\documentclass[11pt,letterpaper]{article}
\pdfoutput=1
\usepackage[letterpaper,margin=1in]{geometry}
\usepackage[T1]{fontenc}
\usepackage[utf8]{inputenc}
\usepackage{microtype}
\usepackage{comment}
\usepackage{amsmath}
\usepackage{amssymb}
\usepackage{geometry}
\usepackage{amsthm}
\usepackage{bm}

\usepackage{tikz}
\usepackage{graphicx}
\usepackage{mathrsfs}
\usepackage{braket}
\usepackage{bbm}
\usepackage{xcolor}
\usepackage{cite}
\usepackage{doi}
\usepackage{tikz}
\usepackage{booktabs}
\usepackage{threeparttable}
\usepackage{multirow}

\newcommand{\cT}{\mathcal{T}}

\newcommand{\cP}{\mathcal{P}}
\newcommand{\poly}{\mathrm{poly}}
\newcommand{\polylog}{\mathrm{polylog}}

\newtheorem{theorem}{Theorem}[section]
\newtheorem{lemma}[theorem]{Lemma}
\newtheorem{corollary}[theorem]{Corollary}

\usetikzlibrary{quantikz2,calc,tikzmark}
\usepackage{theoremref}

\newcommand{\UnderArrow}[4][->]{\begin{tikzpicture}[remember picture,overlay]
    \draw[#1]
      ($(pic cs:#2)+(0,-0.9ex)$) to[bend right=18]
      node[below]{\scriptsize #4}
      ($(pic cs:#3)+(0,-0.9ex)$);
  \end{tikzpicture}}

\usepackage{algorithm}
\usepackage{algorithmic}

\theoremstyle{plain}

\theoremstyle{plain}
\newtheorem{rem}{\protect\remarkname}
\theoremstyle{plain}
\newtheorem*{lem*}{\protect\lemmaname}
\theoremstyle{plain}
\newtheorem*{thm*}{\protect\theoremname}
\theoremstyle{plain}
\newtheorem{prop}{\protect\propositionname}
\theoremstyle{plain}

\theoremstyle{plain}
\newtheorem*{cor*}{\protect\corollaryname}

\newtheorem{defn}{Definition}
\newtheorem*{defn*}{Definition}

\usepackage[USenglish]{babel}
  \providecommand{\corollaryname}{Corollary}
  \providecommand{\lemmaname}{Lemma}
  \providecommand{\propositionname}{Proposition}
  \providecommand{\remarkname}{Remark}
\providecommand{\theoremname}{Theorem}

\usepackage[normalem]{ulem}
\usepackage{authblk}

\hypersetup{colorlinks=true, linkcolor=blue, urlcolor=blue, citecolor = blue}

\newcommand{\RR}{\mathbb{R}}
\newcommand{\CC}{\mathbb{C}}
\newcommand{\wt}{\widetilde}

\newcommand{\dd}{\mathrm{d}}
\newcommand{\Tr}{\mathrm{Tr}}
\newcommand{\Or}{\mathcal{O}}

\renewcommand{\Re}{\operatorname{Re}}
\renewcommand{\Im}{\operatorname{Im}}

\title{Convergence of the Cumulant Expansion and Polynomial-Time Algorithm for Weakly Interacting Fermions}
\author[4]{Hongrui Chen}
\author[6]{Cambyse Rouz{\'e}}
\author[9,10]{Jielun Chen}
\author[8]{Jiaqing Jiang}
\author[11]{Samuel O. Scalet}
\author[9,10]{Yongtao Zhan}
\author[7]{Garnet Kin-Lic Chan}
\author[4,5]{Lexing Ying}
\author[1,2,3]{Yu Tong}
\affil[1]{Department of Mathematics, Duke University, Durham, North Carolina 27708, USA }
\affil[2]{Department of Electrical and Computer Engineering, Duke University, Durham, North Carolina 27708, USA}
\affil[3]{Duke Quantum Center, Duke University, Durham, North Carolina 27708, USA}
\affil[4]{Department of Mathematics, Stanford University, Stanford, CA 94305, USA}
\affil[5]{Institute for Computational and Mathematical Engineering, Stanford University, Stanford, CA 94305, USA}
\affil[6]{Inria, Télécom Paris - LTCI, Institut Polytechnique de Paris, France}
\affil[7]{Division of Chemistry and Chemical Engineering, California Institute of Technology,
Pasadena, California 91125, USA}
\affil[8]{Simons Institute for the Theory of Computing,
University of California, Berkeley, Berkeley, CA 94720, USA}
\affil[9]{Division of Physics, Mathematics, and Astronomy, California Institute of Technology, Pasadena CA 91125, USA} 
\affil[10]{Institute for Quantum Information and Matter, Pasadena CA 91125, USA}
\affil[11]{Department of Computer Science, University of California, Davis, CA, 95616, USA}
\date{}

\begin{document}

\maketitle

\begin{abstract}
    We propose a randomized algorithm for computing the log-partition function of weakly interacting fermionic systems with runtime polynomial in both the system size and the desired precision. While substantial algorithmic progress has been made for weakly interacting fermions, including the development of diagrammatic quantum Monte Carlo methods, a mathematically rigorous polynomial-time algorithm for computing log-partition functions has remained unavailable. In this work, we establish such an algorithm.

    Our approach begins with a convergence analysis of the cumulant expansion. A central ingredient is a tree-determinant expansion, which reorganizes the sum of connected Feynman diagrams and exposes an underlying tree structure. This structure leads to a randomized importance-sampling algorithm augmented by belief propagation for estimating the log-partition function. Unlike diagrammatic quantum Monte Carlo methods based on Markov chain Monte Carlo, whose efficiency depends on generally unknown mixing times, our method admits a provable polynomial runtime guarantee.

\end{abstract}

\section{Introduction}

\paragraph{Motivation.}
Computing partition functions is a fundamental problem in theoretical computer science as it captures many of the central challenges of approximate counting and sampling.  For classical and quantum spin systems, this problem has been the subject of extensive study and has led to many fundamental algorithmic and complexity-theoretic results \cite{JerrumSinclair1993ising,Weitz2006counting,sly2010computational,PatelRegts2017deterministic,Mann2021efficient,bravyi2022quantum}.
Fermionic systems, however, remain far less understood, despite their importance in both quantum complexity theory and physics. Our motivation for studying the partition functions of interacting fermionic systems is twofold.
First, recent progress on quantum Gibbs sampling has shown that thermal states of weakly interacting fermions can be prepared in polynomial time \cite{TongZhan2025fast,DingZhanPreskillLin2025end,SmidMeisterBertaBondesan2025polynomial}. From a computational-complexity perspective, this yields a polynomial-time quantum approach to estimating the partition function through thermodynamic integration, and immediately raises a fundamental question: can classical algorithms match this performance, or is there a superpolynomial quantum advantage in this regime?
Second, from a physics perspective, interacting fermions are indispensable in theories of electrons and in paradigmatic models such as the Sachdev-Ye-Kitaev model \cite{Martin2020electronic,Dagotto1994correlated,SachdevYe1993gapless}. The weakly-interacting regime is practically relevant because dielectric screening and orbital delocalization often lead to effective weak interaction \cite{hybertsen1986electron, aryasetiawan1998gw, onida2002electronic}.

Throughout, weak interaction means that the non-quadratic part of the fermionic Hamiltonian is small. This differs from the usual spin-system notion of weak inter-site coupling (sometimes also referred to as weak interaction \cite{BravyiDiVincenzoLoss2008polynomial,mann2026efficient}): even in the absence of non-quadratic interactions, a fermionic system need not factorize over sites, but instead is described by a Gaussian state. {Since the Gaussian state is exactly solvable, one might intuitively expect the system to remain computationally tractable under a weak non-quadratic perturbation.}
However, computing the partition function in this regime remains a highly non-trivial challenge that lies beyond the direct reach of existing tools, whether from the perspective of rigorous algorithm design or many-body physics. 

On the side of rigorous algorithms, cluster expansion has been a powerful tool for computing the partition function. However, its applications have so far focused mainly on high-temperature spin systems \cite{Mann2021efficient,MannMinko2024algorithmic,BakshiLiuMoitraTang2024high,HaahKothariTang2022optimal,RamkumarCaiTongJiang2025high,SanchezSegoviaSchneiderAlhambra2025}.
These high-temperature guarantees do not directly transfer to weakly interacting fermions. In our setting, the natural expansion is in the interaction strength around a general free-fermion state, which may already exhibit nontrivial correlations even in the absence of interactions. This is fundamentally different from an expansion in inverse temperature around the infinite-temperature state, which is a product state and therefore has no correlations. Nor can this difficulty be removed simply by changing the fermionic basis, as a nonlocal basis transformation would generally destroy the locality of the non-quadratic interaction terms, an essential ingredient for algorithmic efficiency for large fermionic systems. On the theoretical physics side, by contrast, a variety of methods naturally target fermionic systems, including diagrammatic expansions \cite{BenfattoGiulianiMastropietro2003low,benfatto2006fermi,giuliani2009rigorous,Salmhofer2019renormalization}, quantum Monte Carlo \cite{BlankenbeclerScalapinoSugar1981,Hirsch1985hubbard,AssaadEvertz2008qmc,Prokofev1998DiagMC,Rossi2017determinant}, and quasi-adiabatic continuation \cite{HastingsWen2005quasiadiabatic,Hastings2007quasi,BachmannMichalakisEtAl2012automorphic,BravyiHastingsMichalakis2010topological}. Yet, bridging these physics-inspired tools with mathematically rigorous runtime guarantees remains difficult; the few approaches that do admit rigorous guarantees currently result in a quasi-polynomial, rather than polynomial, runtime (as discussed in the ``Related Works'' section).

In this work, we give the first polynomial-time classical algorithm for computing the partition function of weakly interacting fermions. In the following, we present the problem setup and our main results.

\paragraph*{Problem Statement.}
We consider a fermionic system with $N$ modes supported on a $d$-dimensional lattice. Writing $\Omega$ for the set of modes (so that $|\Omega|=N$), we associate to each mode $a\in\Omega$ fermionic annihilation and creation operators $\psi_a$ and $\psi_a^\dagger$, which obey the canonical anti-commutation relations
\begin{equation}
\label{eq:fermion_ops_intro}
    \{\psi_a,\psi_b^\dagger\}=\delta_{ab},\qquad
    \{\psi_a,\psi_b\}=0,\qquad
    \{\psi_a^\dagger,\psi_b^\dagger\}=0 .
\end{equation}
For readers more accustomed to qubits, the Jordan-Wigner transformation provides a concrete way to represent the above operators on a qubit system. After fixing an ordering of the modes, each operator can be mapped in the following way,
\[
\psi_a \;\longmapsto\; \Big(\prod_{b<a} Z_b\Big)\,\frac{X_a+iY_a}{2}.
\]
{where $\{X_a,Y_a,Z_a\}$ denote the Pauli matrices on site $a$.}
Each $\psi_a$ is then an operator acting on $(\mathbb{C}^2)^{\otimes N}$. However, it is worth noting that this transformation does not preserve the fermionic notion of locality.

The Hamiltonian we study splits into a free part and an interaction part, 
\begin{align}
   H=H_0+V, \quad H_0=\sum_{a,b\in\Omega}h_{ab}\,\psi_a^\dagger\psi_b, \quad V=\sum_{P\in\mathcal{P}}v_P\,\Psi_P \label{eq:weakH} 
\end{align}
The free part $H_0$ is set by a Hermitian matrix $h=(h_{ab})$ of size $N\times N$. It describes non-interacting fermions and is exactly solvable. The interaction $V$ describes the interaction between fermionic particles, such as the Coulomb interaction between electrons. $\Psi_P$ is a product of an equal number of creation and annihilation operators whose indices form the tuple $P$. A typical example is $\Psi_P=\psi_a^\dagger\psi_b^\dagger\psi_c\psi_d$, where $P=(a,b,c,d)$.
We will initially focus on \textit{geometrically local Hamiltonians}, where the fermionic modes are arranged on a finite-dimensional lattice, and each Hamiltonian term, with spectral norm bounded by $1$, involves only modes that are within a constant radius. The strength of the interaction is controlled by a parameter $U>0$, which satisfies 
\[
|v_P|\le U \text{ for all } P\in \mathcal{P}.
\]
We note that this condition ensures the interaction part $V$ is \emph{locally small}, but its spectral norm $\|V\|$ can be arbitrarily large as $N$ grows.
Our results and techniques also naturally extend to long-range interactions whose strength decays sufficiently fast with distance. We will later specify the required decay conditions.
 
 At inverse temperature $\beta$, the partition function is $Z=\Tr[e^{-\beta H}]$, and our goal is to estimate the log-partition function $\log Z$. This quantity is extensive: it scales linearly with the system size $N$. We therefore aim to approximate $\log Z$ to within additive error $\epsilon N$, i.e., error $\epsilon$ per mode, the natural target for an extensive quantity.

\paragraph{Main Results.} Our primary contribution is the first polynomial-time (in the system size $N$ and inverse precision $1/\epsilon$) classical algorithm for the partition function of weakly interacting fermions. At its heart is a rigorous analysis of the cumulant expansion, the Taylor series of $\log Z$ in the interaction coefficients. 
\begin{align} \label{eq:cumulant_expansion_intro}
 \log (Z/Z_0)= \sum_{s=1}^\infty \frac{(-1)^{s}}{s!} \sum_{P_1,\cdots,P_s \in \mathcal{P}} v_{P_1}\cdots v_{P_s} \int_{[0,\beta]^s} \dd \tau_1\cdots \dd \tau_s  \,\mathcal{E}_{c}(\{P_i,\tau_i\}_{i\in[s]})
 \end{align}
where the cumulant function $\mathcal{E}_{c}(\{P_i,\tau_i\}_{i\in[s]})$ is defined in \eqref{eq:defcumulanttext}.
$Z_0$ { is the free partition function $\Tr[e^{-\beta H_0}] $ } and can be efficiently evaluated as
\[
Z_0 = \mathrm{det}(I+e^{-\beta h}).
\]
Building on techniques from the rigorous study of fermionic quantum field theory in mathematical physics \cite{GentileMastropietro2001renormalization,PedraSalmhofer2008determinant}, we show that this series converges exponentially fast when the interaction strength is below a threshold that is \emph{independent of the system size}. 
\begin{theorem}[Convergence of the cumulant expansion; informal]\label{thm:convergence_informal}
Let $H=H_0+V$ be a geometrically local fermionic Hamiltonian as above, and fix an inverse temperature $\beta$. There is a constant $C(\beta)>0$, independent of $N$, such that whenever the interaction strength satisfies $U\le C(\beta)$, the cumulant expansion for $\log(Z/Z_0)$ converges exponentially with a constant rate.
\end{theorem}

This exponential convergence result is nontrivial: unlike many standard polymer cluster-expansion arguments, it does not follow from absolute convergence after a combinatorial count of connected terms. Instead, the analysis must exploit cancellations arising from the fermionic anti-commutation relations in \eqref{eq:fermion_ops_intro}. To see the difficulty, consider the analogous perturbative expansion in a classical Gaussian setting:
\[
Z_{\mathrm{classical}}(\lambda)
=
\int_{-\infty}^{\infty}
e^{-\left(\frac{x^2}{2}+\lambda x^4\right)} \dd x .
\]
Although the Hamiltonian $\frac{x^2}{2}+\lambda x^4$ has the quadratic-plus-quartic form similar to a typical fermionic Hamiltonian, the perturbative expansion of $\log Z_{\mathrm{classical}}(\lambda)$ in $\lambda$ has no positive radius of convergence around $\lambda=0$. Indeed, any negative value of $\lambda$, however small, makes the integrand non-integrable. The same obstruction would arise if the fermionic creation and annihilation operators in \eqref{eq:weakH} were replaced by bosonic ones. Thus, the positive-radius exponential convergence we prove, uniformly in the system size $N$, is possible only because of cancellations enforced by the fermionic anti-commutation relations \eqref{eq:fermion_ops_intro}.

The exponential convergence allows us to truncate the cumulant expansion at order $\Or(\log(1/\epsilon))$, reducing the problem to evaluating a finite sum. This, however, does not automatically lead to an efficient algorithm. As will be explained in the ``Techniques'' section, a term-by-term evaluation of the cumulant expansion leads to a quasi-polynomial runtime, rather than the polynomial runtime we aim for. Because of the importance of the weakly interacting fermions, many methods besides term-by-term evaluation have been proposed, as discussed in the ``Related Works'' section. These works, however, also only manage to achieve quasi-polynomial runtime when mathematically rigorous guarantees are available.

Instead of term-by-term evaluation, we evaluate the truncated finite sum by \emph{importance sampling}: we first sample a term from the expansion represented by a tree, and then sample the variables associated to the tree using a belief-propagation algorithm, which runs efficiently because of the tree structure.
This yields our first main result, an efficient algorithm for the log-partition function at finite temperature.
\begin{theorem}[Finite-temperature algorithm; informal]\label{thm:ft-informal-draft}
Let $H=H_0+V$ be a geometrically local fermionic Hamiltonian as above, and fix an inverse temperature $\beta$. There is a constant $C(\beta)>0$, independent of $N$, such that whenever the interaction strength satisfies $U\le C(\beta)$, a randomized classical algorithm estimates $\log Z$ to within additive error $\epsilon N$ (i.e., error $\epsilon$ per mode), with probability at least $2/3$, in time $\wt{\Or}(N\epsilon^{-2})$. For translation-invariant systems the runtime improves to $\wt{\Or}(\epsilon^{-2})$.
\end{theorem}
{Here, $\wt{\mathcal{O}}$ hides the polylogarithm factors.}
The formal statement is the log-partition-function part of Theorem~\ref{thm-algorithm-geo-local}, with the end-to-end runtime stated in Corollary~\ref{cor-finite-range-ham}.

We note that the runtime is nearly linear in $N$, which is essentially optimal, since merely reading the Hamiltonian already takes linear time. Moreover, this result is not confined to high temperature: for every fixed inverse temperature $\beta$, it suffices that the interaction strength $U$ lie below a $\beta$-dependent threshold $C(\beta)$.
In general, the threshold $C(\beta)$ vanishes as $\beta \to \infty$.

Beyond the partition function, the same algorithm also implies efficient computation of the expectation value of local observables that are supported on $\Or(1)$ sites. This relies on the log-partition function being a generating function: for any local observable $O$, perturbing the Hamiltonian by $H\mapsto H+\lambda O$, one has
\[
-\frac{1}{\beta}\,\frac{\partial}{\partial\lambda}\log\Tr\!\big[e^{-\beta(H+\lambda O)}\big]\Big|_{\lambda=0}=\frac{1}{Z}\Tr(Oe^{-\beta H}),
\]
so the thermal expectation value of $O$ is a derivative of the (perturbed) log-partition function, which we estimate by a finite difference. Since $O$ is local, $H+\lambda O$ remains a local Hamiltonian of the same form, so the same expansion and algorithmic techniques behind Theorem~\ref{thm:ft-informal-draft}  apply to it directly. 
\begin{corollary}[Local observables; informal]\label{cor:obs-informal-draft}
Let $O$ be a bounded-size local observable. Under the assumptions of Theorem~\ref{thm:ft-informal-draft}, a randomized classical algorithm running in time $\widetilde{\Or}(\epsilon^{-2})$ estimates the thermal expectation $\frac{1}{Z}\Tr(Oe^{-\beta H})$ to additive error $\epsilon$ with probability at least $2/3$.
\end{corollary}

The formal statement is the local-observable part of Theorem~\ref{thm-algorithm-geo-local} (with end-to-end runtime in Corollary~\ref{cor-finite-range-ham}), resting on the observable algorithm detailed in Section~\ref{sec:alg-observable}.
We remark that computing a local observable is in fact cheaper than computing the full log-partition function: the runtime is $\wt{\Or}(\epsilon^{-2})$, which is independent of the system size $N$, since only the terms of the expansion supported near $O$ contribute and the computation is confined to a bounded neighborhood of $O$.

Moreover, our results are not limited to the strictly finite-range setting; the algorithms naturally extend to {Hamiltonians that contain long-range but decaying terms. For the interaction $V$, the condition we need is $L_V$-summability: the total interaction strength incident to any fixed mode is uniformly bounded by $L_V$, independent of the system size (see Definition~\ref{defn:conditions-on-v}). To state the condition on $H_0$, we introduce the non-interacting Green's function:}
\begin{equation}
    \label{eq:non_int_green_intro}
    g_\tau(a,b)=\big\langle \psi_a(\tau)\,\psi_b^\dagger\big\rangle_0,\qquad \psi_a(\tau)=e^{\tau H_0}\psi_a e^{-\tau H_0}, \quad \langle \cdot \rangle_0 =  \Tr(\cdot e^{-\beta H_0})/Z_0
\end{equation}
(see Definition~\ref{def:Green} for the precise definition). { The quantity $g_\tau(a,b)$ is the two-point correlation of the free system: intuitively, it measures the amplitude for a fermionic excitation created at mode $b$ to propagate to mode $a$ after imaginary-time evolution $\tau$.}
The non-interacting Green's function, because it only involves the quadratic part $H_0$, can be explicitly computed as
\[
g_\tau(a,b) = (e^{-\tau h}(I+e^{-\beta h})^{-1})_{ab},
\]
where $h=(h_{ab})$ is the matrix containing the coefficients of quadratic terms in $H_0$.
{The condition on $H_0$ is expressed through the spatial decay of this Green's function: finite-range $H_0$ gives exponential spatial decay of $g_\tau$ uniformly for $\tau\in[-\beta,\beta]$, while the general long-range result only assumes $L_g$-summability, namely $\max_{\tau\in[-\beta,\beta]}\max_{a\in\Omega}\sum_{b\in\Omega}|g_\tau(a,b)|\le L_g$ (see Definition~\ref{defn:conditions-on-h0}).}

\begin{theorem}[General interactions; informal]\label{thm:general-informal-draft}
{Assume that the interacting terms in $V$ are $L_V$-summable as in Definition~\ref{defn:conditions-on-v}, and that the non-interacting Green's function of $H_0$ is $L_g$-summable as in Definition~\ref{defn:conditions-on-h0}. There is a constant $C(\beta) > 0$,
independent of $N$, such that if $L_VL_g\le C(\beta)$, the cumulant expansion converges, and the log-partition function and local observable expectation values can be computed in time $\mathrm{poly}(N,1/\epsilon)$.}
\end{theorem}

{Theorem~\ref{thm:general-informal-draft} is made precise in Theorem~\ref{thm-algorithm-general} and Corollary~\ref{cor-finite-range-ham}. While the complexity in the more general summable setting of Theorem~\ref{thm:general-informal-draft} does not retain the nearly optimal dependence on $N$ from the geometrically local setting of Theorem~\ref{thm:ft-informal-draft} and Corollary~\ref{cor:obs-informal-draft}, it is still polynomial for both the log-partition function and local observable expectation values.} This contrasts with the long-range quantum spin setting, where known high-temperature results only achieve quasi-polynomial runtime \cite{SanchezSegoviaSchneiderAlhambra2025}, while our fermionic framework naturally yields a polynomial-time algorithm in the long-range setting.

\paragraph{Techniques.} 

Our starting point is the cumulant expansion \eqref{eq:cumulant_expansion_intro}.
The cumulant expansion is the Taylor expansion of the interaction contribution $\log(Z/Z_0)$ to the log-partition function around the free-fermion point, i.e. $V=0$, with the interaction coefficients $v_P$ as the expansion variables. Because the unperturbed Hamiltonian $H_0$ is quadratic, the background state is a free-fermion, or Gaussian, state. As in a classical Gaussian distribution, its higher-order correlations are completely determined by its two-point correlations,
and every cumulant  $\mathcal{E}_{c}(\{P_i,\tau_i\}_{i\in[s]})$ in this  series can be computed via Wick contraction from the non-interacting Green's function \eqref{eq:non_int_green_intro} \cite[Section~IV]{bravyi2004lagrangian}.
The convergence of such expansions has been addressed in certain regimes, for example for the Fermi-Hubbard model in certain parameter ranges \cite{BenfattoGiulianiMastropietro2003low,benfatto2006fermi,giuliani2009rigorous}. For our algorithmic purpose, however, even if convergence allows the series to be truncated at logarithmic order, directly computing the truncated coefficients can still be computationally expensive due to the combinatorial growth of the number of terms one needs to evaluate for each $\mathcal{E}_{c}(\{P_i,\tau_i\}_{i\in[s]})$, leading to a quasi-polynomial runtime, as will be discussed below. 

The most direct way to compute these coefficients is the standard \emph{diagrammatic expansion}, where each cumulant function $\mathcal{E}_{c}(\{P_i,\tau_i\}_{i\in[s]})$ is expressed as a sum of \emph{Feynman diagrams}. In this approach, Feynman diagrams systematically enumerate the Wick contractions contributing to each coefficient of the Taylor series above. Due to the linked-cluster theorem \cite[Section~2.3]{negele2018quantum}, we in fact only need to consider \emph{connected Feynman diagrams}. Suppose the convergence is rapid enough so that we can truncate the expansion at order $s=\Or(\log(1/\epsilon))$ to ensure an error of at most $N\epsilon$ (hence an error of $\epsilon$ per mode). At the $s$th order, the number of labeled connected Feynman diagrams for $s$ 4-local interaction terms (a \textit{configuration} using the terminology of \cite{LiWallerbergerGull2020diagrammatic}) is on the order of $(2s)!$ \cite{castro2018equivalence}, and there are asymptotically $N^s$ choices of these interaction terms. Thus, the brute-force evaluation of all labeled Feynman diagrams at the required orders takes
\[
\Or(N^s(2s)!)=N^{\Or(\log(1/\epsilon))}2^{\Or(\log(1/\epsilon)\log\log(1/\epsilon))}
\]
runtime.
Diagrammatic Monte Carlo methods avoid evaluating diagrams one by one by grouping many diagram contributions into determinants and evaluating the resulting sum by Markov-chain Monte Carlo (MCMC) \cite{Rossi2017determinant,RossiProkofevEtAl2017polynomial}.
However, a rigorous polynomial-time guarantee would require controlling the mixing time of this MCMC procedure, which is not known in this setting; even ergodicity can be subtle \cite{LiWallerbergerGull2020diagrammatic}.

To obtain a provable polynomial-time algorithm, we use the cluster-expansion viewpoint developed in \cite{Lesniewski1987effective,abdesselam1997explicitfermionictreeexpansions,Salmhofer2000,BenfattoGallavottiEtAl1994beta,GentileMastropietro2001renormalization} and widely used in rigorous fermionic expansions for quantum field theory \cite{BenfattoGiulianiMastropietro2003low,benfatto2006fermi,giuliani2009rigorous,Metzner2012functional,Giuliani2021gentle,Husemann2012frequency,Giuliani_2009,MastropietroPorta2022multi,Salmhofer2019renormalization,PortaMastropietroGiuliani2023response,Mastropietro2014weyl,Giuliani2011ground,GiulianiMastropietroPora2010anomalous,Mastropietro2008luttinger,Benfatto2009ultraviolet,GentileMastropietro2001renormalization}. The key object is a \textit{{tree-}determinant expansion} \eqref{eq:determinant_expansion}, which comes from the tree-graph identities of cluster expansions \cite{brydges1978new,Brydges1984short,Battle1984ANO} and is closely related to the anchored tree graphs used in \cite{benfatto2006fermi}. In this expansion, each cumulant $\mathcal{E}_{c}(\{P_i,\tau_i\}_{i\in[s]})$ is reorganized as a sum over labeled trees rather than over connected Feynman diagrams.
Specifically, at the $s$th order, the tree-determinant expansion organizes the terms in the cumulant expansion as
\begin{equation}
\label{eq:determinant_expansion_intro}
    \mathcal{E}_c(\{P_i,\tau_i\}_{i\in[s]} ) = \sum_{T \in \mathcal{T}([s])} \sum_{\chi \in \mathcal{A}(T)} \alpha_{T,\chi}\prod_{(i,j)\in T} g_{\tau_i,\tau_j}(P_i,P_j,\chi_{ij})\, h_{\boldsymbol{\tau}}(P_1,\cdots,P_s,T,\chi),
\end{equation}
where 
\begin{equation}
\label{eq:determinant_integral_intro}
    h_{\boldsymbol{\tau}}(P_1,\cdots,P_s,T,\chi) =\sum_{\omega \in S(T)}\int_{[0,1]^{s-1}} \dd \mathbf{t}\, p_{T,\omega}(\mathbf{t}) \det \mathbf{G}(T,\chi,\omega,\mathbf{t},\{P_i,\tau_i\}_{i\in[s]}). 
\end{equation}
This expansion is proved in Theorem~\ref{cumulant-representation}. Here $\mathcal{T}([s])$ is the set of labeled trees on the $s$ interaction vertices, each $P_i$ labels an interacting term in $V$, $v_{P_i}$ is its strength, and $g_{\tau_i,\tau_j}$ is the non-interacting Green's function connecting two interaction terms along an edge of the tree. $\chi$ is an ``assignment'' that selects a pair of modes from $P_i$ and $P_j$ for $(i,j\in T)$, as defined in Section~\ref{sec:determinant-expansion}. The factor $h_{\boldsymbol{\tau}}$ is the determinant contribution, which contains the remaining contractions and fermionic signs.

The tree-determinant expansion \eqref{eq:determinant_expansion_intro} is used in two crucial ways in this work. First, it helps us establish exponential convergence of the cumulant expansion \eqref{eq:cumulant_expansion}, leading to Theorem~\ref{thm:convergence_informal}. Second, it enables us to construct the efficient algorithms in Theorem~\ref{thm:ft-informal-draft} and Corollary~\ref{cor:obs-informal-draft}.

\textit{Exponential convergence.}
For simplicity we will focus on the case where the interaction $V$ is { geometrically local}, i.e., only involve modes that are within some constant distance from each other.
A key observation is that the number of trees in
\eqref{eq:determinant_expansion_intro} grows like $s!$, up to an
exponential factor, due to Cayley's formula. In the cumulant expansion, the prefactor $1/s!$ cancels this factorial growth. Hence the $s$th-order contribution is
bounded by $|\mathcal{P}| e^{O(s)} U^s$, provided that the determinant
$
\det \mathbf{G}\bigl(T,\chi,\omega,\mathbf{t},
\{P_i,\tau_i\}_{i\in[s]}\bigr)
$
grows at most exponentially in $s$. 
The factor $|\mathcal{P}|=\Or(N)$ (number of interaction terms in $V$) is a result of choosing the first vertex label in the tree, while the decay of the non-interacting Green's function $g_\tau$ controls the contribution from all other vertex labels.
While determinants can grow
superexponentially in the matrix size in general, the matrices
$\mathbf{G}$ inherit a Gram-Hadamard structure from the non-interacting
Green's function. Lemma~\ref{lem-det-bound} uses this structure to prove
a uniform $e^{O(s)}$ bound on the determinant. These observations combined enable us to bound the $s$th order term in the cumulant expansion \eqref{eq:cumulant_expansion_intro} by $N e^{O(s)}U^s$, thus guaranteeing exponential convergence for $U$ below a system-size independent threshold. 
This extends the standard argument for periodic systems \cite{PedraSalmhofer2008determinant}, where
the system consists of identical unit cells, to our nonperiodic setting. Intuitively, this result indicates that the determinant is able to capture the fermionic sign cancellation.

\textit{Efficient algorithm.} 
The tree-determinant expansion \eqref{eq:determinant_expansion_intro} also helps us perform importance sampling without using MCMC. More specifically, to evaluate the cumulant, we first reorganize the cumulant expansion \eqref{eq:cumulant_expansion_intro} and the tree-determinant expansion \eqref{eq:determinant_expansion_intro} into the following form
\[
\log(Z/Z_0) = \sum_{s=0}^{\infty} \Xi_s,
\]
where
\[
\Xi_s = \frac{{(-1)^s}}{s!} \sum_{T \in \mathcal{T}([s])}  \int_{[0,\beta]^s} \dd \boldsymbol{\tau} \sum_{P_1,\cdots,P_s } v_{P_1} \cdots v_{P_s}\sum_{\chi \in  \mathcal{A}(T)}\alpha_{T,\chi}\prod_{(i,j)\in T}g_{\tau_i,\tau_j}(P_i,P_j,\chi) h_{\boldsymbol{\tau}}(P_1,\cdots,P_s,T,\chi),
\]
where $\alpha_{T,\chi}$ is a sign that needs to be accounted for in the algorithm but does not affect the analysis.
Note that from the convergence analysis we know that the above sum over $s$ can be truncated at order $S=\Or(\log(1/\epsilon))$ to achieve average free energy accuracy $\epsilon$.

We will construct an unbiased estimator for each term $\Xi_s$ from the following sampling procedure. We first uniformly sample a tree (which is done efficiently using the Pr\"ufer code). Conditional on the tree, we then sample the interaction terms $P_1,\ldots,P_s$ according to their contribution, in absolute value, to $\Xi_s$, using the tree structure. Indeed, after the tree and the imaginary times are fixed, and replacing the determinant part $h_{\boldsymbol{\tau}}(P_1,\cdots,P_s,T,\chi)$ by a constant factor, the absolute weights for the interaction labels factorize into vertex factors $|v_{P_i}|$ and edge factors coming from the non-interacting Green's function, so these labels form a Markov random field (MRF) on the tree, which will serve as the proposal distribution in our importance sampling procedure. We will sample from the corresponding normalized proposal distribution:
\begin{align}
\mathbb{P}(P_1,\cdots,P_s) = \frac{1}{Z_s(T)}|v_{P_1} \cdots v_{P_s}|\prod_{(i,j)\in T} M_{\tau_i,\tau_j}(P_i,P_j),
\end{align}
where $M_{\tau_i,\tau_j}(P_i,P_j)$ are factor potentials computed from the non-interacting Green's function $g_{\tau_i,\tau_j}(P_i,P_j,\chi)$.
Due to the tree structure ($T$ in the above expression is a tree), this MRF sampling can then be done efficiently using belief propagation (BP).

For each sampled instance, which includes the configuration $(P_1,P_2,\cdots,P_s)$ together with other random variables introduced in Section~\ref{sec:sampling}, we evaluate a weight $w_s$ as defined in \eqref{eq:estimator}, which we also include here:
\begin{align}  
\label{eq:sampled_weight_intro}
w_s = {\frac{(-1)^s\,s^{s-2}}{s!}}\,\alpha_{T,\chi}\, \mathrm{sgn}\left( {v_{P_1}\cdots v_{P_s}}\prod_{(i,j) \in T}  \,g_{\tau_i,\tau_j}(P_i,P_j,\chi) \right) Z_{{s}}(T)\,\beta^s \det(\mathbf{G}(T,\chi,\omega,\mathbf{t},\{P_i,\tau_i\})).\end{align}
These samples allow us to estimate $\Xi_s$ in an unbiased manner, since
\[
\mathbb{E}[w_s] = \Xi_s.
\]

The explicit form of $w_s$ allows us to estimate the variance of the estimator. Note that 
\[
|w_s|\leq {\frac{s^{s-2}}{s!}} Z_{{s}}(T)\,\beta^s |\det(\mathbf{G}(T,\chi,\omega,\mathbf{t},\{P_i,\tau_i\}))|.
\]
The convergence discussion already allows us to bound $Z_{{s}}(T)$  by $Ne^{\Or(s)}U^s$, and $|\det(\mathbf{G}(T,\chi,\omega,\mathbf{t},\{P_i,\tau_i\}))|$ by $e^{\Or(s)}$, while $\frac{s^{s-2}}{s!}=e^{\Or(s)}$ by Stirling's approximation. As a result, we can see that
\[
|w_s|\leq Ne^{\Or(s)}U^s,
\]
which allows us to bound the variance of $w_s$ by $N^2e^{\Or(s)}U^{2s}$, which is independent of the system size. Consequently, this leads to $N^2/\epsilon^2$ sampling overhead when $U$ is below the threshold. Since each sample can be generated efficiently using the BP algorithm, we have an efficient algorithm for estimating $\log(Z)$ to $\epsilon$ precision.

\paragraph{Related Works.} Below we will review related algorithms and convergence results for weakly-interacting fermions and similar quantum many-body systems.

\emph{Numerical algorithms for weakly interacting fermions.} The most common approach to address weak interactions is through diagrammatic expansion, which has been discussed in the preceding paragraphs. Brute-force summation based on diagrammatic expansion can only lead to a quasi-polynomial time algorithm, due to the combinatorially growing number of Feynman diagrams at a given order.
Techniques in the diagrammatic quantum Monte Carlo (QMC) method \cite{Prokofev1998DiagMC,Prokofev2007BoldDiagMC,VanHoucke2010DiagMC,Kozik2010DiagMCFermions} are available to circumvent the inefficiencies in the brute-force approach. Importantly, Feynman diagrams are not evaluated one by one. Instead, contributions from a large number of Feynman diagrams are evaluated together in the form of determinants, as shown in \cite{Rossi2017determinant}. Moreover, the summation over interaction terms is evaluated through importance sampling aided by Markov-chain Monte Carlo (MCMC). The efficiency of the method then depends crucially on the mixing time of the random walk in the MCMC procedure, i.e., the time it takes for the random walk to relax to the target distribution. Assuming the efficiency of the MCMC sampling, a polynomial runtime can be achieved \cite{RossiProkofevEtAl2017polynomial}. However, due to the complexity of the interactions present in a fermionic quantum system, it is unclear how one can quantify this mixing time. In fact, even ensuring the ergodicity of the random walk is a non-trivial task \cite{LiWallerbergerGull2020diagrammatic}. This difficulty has so far prevented a rigorous proof that the diagrammatic QMC algorithm runs in time $\mathrm{poly}(N,1/\epsilon)$ even in the weak interaction setting.

\emph{Rigorous results in restricted regimes.} Even though no polynomial time algorithm was available, previous works have produced rigorous results in the non-interacting limit, where the interaction strength tends to zero. As an example, it has been shown that the density-matrix embedding theory (DMET) achieves first-order accuracy in the interaction strength in this limit \cite{CancesFaulstichEtAl2025analysis}. In the case where the Hamiltonian is close to being decoupled, i.e., becoming a sum of non-overlapping terms, a rigorous polynomial time algorithm has been found to compute ground state properties for quantum spin systems \cite{BravyiDiVincenzoLoss2008polynomial}. This algorithm does not apply to our setting because, in addition to the obvious differences in the problem setup, it relies crucially on the fact that in their setting the unperturbed Hamiltonian has a product ground state, which is not true in our setting.

\emph{Cluster expansion.}
The cluster expansion is a powerful tool that provides a crucial idea for this work. It has also been used in many other scenarios, in particular for computing the log-partition function of classical systems at high temperature or low density \cite{penrose1963convergence,ruelle1963correlation,ruelle1969statistical,kotecky1986cluster,Brydges1984short,fernandez2007cluster,fernandez2007analyticity,Procacci1998,bricmont1996high}. 
Its application to quantum systems is typically limited to the high-temperature regime \cite{Mann2021efficient,BakshiLiuMoitraTang2024high,HaahKothariTang2022optimal,RamkumarCaiTongJiang2025high,SanchezSegoviaSchneiderAlhambra2025}.
Our work goes beyond this setting since, with a fixed quadratic part, our algorithm can handle arbitrary temperatures and densities with sufficiently weak interaction.
Ref.~\cite{helmuth2023efficient} proposes an efficient algorithm for computing the log-partition function of quantum spin systems at low temperature, where the Hamiltonian consists of a classical spin Hamiltonian supplemented by small nearest-neighbor interactions that may be non-diagonal in the classical basis. Extending this approach to the setting considered in our work would be highly nontrivial, if possible at all, since the unperturbed Hamiltonian $H_0$ in our framework may not be diagonalizable in any local basis. Moreover, diagonalizing $H_0$ in terms of its single-particle orbitals can completely destroy the locality of the perturbation $V$, which is essential for the method in \cite{helmuth2023efficient}.

\emph{Quasiadiabatic continuation.} Another approach that brings us close to a polynomial time algorithm is quasiadiabatic continuation \cite{HastingsWen2005quasiadiabatic,Hastings2007quasi,Hastings2010quasi,BachmannMichalakisEtAl2012automorphic,BravyiHastingsMichalakis2010topological}. In this approach, one starts from a Hamiltonian whose ground state is easy to prepare. For ground state we can start from the ground state of the non-interacting (quadratic) Hamiltonian. From this state a quasiadiabatic evolution can be generated using a quasi-local Hamiltonian. If the initial Hamiltonian is adiabatically connected to the target Hamiltonian, i.e., a spectral gap that is independent of the system size exists for a path linking these two Hamiltonians, then the quasiadiabatic dynamics will take the initial ground state to the ground state of the target Hamiltonian in constant time. Therefore computing observable expectation values of the target ground state is reduced to implementing the constant-time quasiadiabatic dynamics in the Heisenberg picture, which can be done in time $\Or(2^{\mathrm{polylog}(1/\epsilon)})$. 
Though originally developed for the ground state, this approach can be extended to the thermal state using a parent Hamiltonian \cite{FirankoGoldsteinArad2024area,chen2023efficient}. Moreover, the log-partition function can be computed using thermodynamic integration \cite{kirkwood1935statistical} in time $\Or(2^{\mathrm{polylog}(N/\epsilon)})$, where $N$ appears in the exponent because one needs to estimate the energy (Hamiltonian expectation value), which is an extensive quantity, to precision $\Or(\epsilon)$ for multiple points in time.

\emph{System size truncation through decay of correlation.} 
Uniform correlation decay of Gibbs states also brings us closer to a polynomial-time classical algorithm. For Gibbs states, uniform correlation decay implies local indistinguishability \cite{brandao2019finite,rouze2024efficient}, meaning that the expectation value of any local observable can be well approximated by its expectation value in the Gibbs state of the Hamiltonian restricted to a local patch. The approximation error decays exponentially with the radius of this patch.
On a $d$-dimensional lattice, this property yields a classical algorithm for estimating local observable expectation values with runtime $\mathcal{O}(2^{\log^d(1/\epsilon)})$. Consequently, the log-partition function can be computed in time $\mathcal{O}(2^{\log^d(N/\epsilon)})$. Moreover, uniform correlation decay has been rigorously established for Gibbs states of short-range Hamiltonians at sufficiently high temperatures in arbitrary dimensions, as well as for all temperatures in one-dimensional systems \cite{capel2025decay, kimura2025clustering, kliesch2014locality}. 

\emph{Previous convergence results.} The convergence of Feynman diagrams in weakly interacting fermionic systems has been extensively studied, but previous works have mostly focused on the periodic case \cite{abdesselam1997explicitfermionictreeexpansions,Salmhofer2000,BenfattoGiulianiMastropietro2003low,benfatto2006fermi,giuliani2009rigorous,Metzner2012functional,Giuliani2021gentle,Husemann2012frequency,Giuliani_2009,MastropietroPorta2022multi,Salmhofer2019renormalization,PortaMastropietroGiuliani2023response,Mastropietro2014weyl,Giuliani2011ground,GiulianiMastropietroPora2010anomalous,RivasseauWang2023honeycomb,Mastropietro2008luttinger,Kashima2014renormalization,Kashima2015zero,Benfatto2009ultraviolet,Wang2019nonfermi,GentileMastropietro2001renormalization}. Because of the interest in addressing non-periodic systems in computational problems, we modify the determinant bound in \cite{Salmhofer2019renormalization} to adapt to this more general case. As noted before, the convergence of the Feynman diagrams expansion does not automatically lead to a polynomial-time algorithm. To the best of our knowledge, our work is the first to take this step.

\emph{Quantum algorithms.} Recent works on Lindbladian-based quantum algorithms \cite{chen2023efficient,ChenKastoryanoBrandaoGilyen2023quantum} have resulted in efficient quantum algorithms for the high-temperature \cite{rouze2024efficient,RouzeFrancaAlhambra2024optimal} setting and for weakly interacting fermions \cite{TongZhan2025fast,DingZhanPreskillLin2025end,SmidMeisterBertaBondesan2025polynomial,Smid2025rapidmixingquantumgibbs}. 
The results in this work show that we should also not expect significant quantum advantage in the weakly interacting setting. However, given that the Lindbladian-based quantum algorithms used in the above works apply beyond weak interactions, this in no way rules out possible quantum advantage based on these quantum algorithms.

\paragraph{Open Problems and Outlook.} Our results and techniques open up many possible future research directions. One obvious extension is to try to extend our method to zero-temperature setting to propose efficient algorithms for the ground state. Although in classical algorithms for quantum many-body systems the zero-temperature setting often requires completely different methodology \cite{BravyiDiVincenzoLoss2008polynomial}, our approach may be generalizable to the zero-temperature case with only modest modifications. This is because our algorithm relies mainly on the decay of the non-interacting Green's function rather than the smallness of $\beta$. Given decay in $\tau$ and space, both of which can be guaranteed by assuming $H_0$ is gapped, we expect to extend the analysis and algorithm construction in this work to study the ground state of weakly interacting fermionic systems.

One surprising feature of our algorithm is that it not only runs in polynomial time, but the dependence on $N$ is in fact linear when both $H_0$ and $V$ are finite-ranged. This is in contrast with the polynomial time algorithm for the high-temperature spin case, where the polynomial degree of the system size dependence is only known to be $\Or(1)$ \cite{Mann2021efficient}. This raises the question whether this good scaling is due to some structure that is inherent to fermionic systems or whether the techniques used in this work can be used to improve the algorithm for the high-temperature spin setting as well.

Our convergence result (Theorem~\ref{thm-convergence}) can be equivalently stated as the absence of a complex zero of the partition function in the vicinity of the origin. This has been studied in the context of quantum spin systems \cite{LeeYang1952I,LeeYang1952II} and has resulted in quasi-polynomial time classical algorithm for certain Hamiltonians \cite{HarrowMehrabanSoleimanifar2020classical}. It is of interest to explore whether the tree-graph identity used in this work can lead to a polynomial time algorithm for this type of quantum spin system.

\paragraph{Organization.} The rest of the paper is organized as follows. In Section~\ref{sec.2} we introduce the notations and the problem setup. In Section~\ref{sec:main_results} we summarize the main results of this paper and then present the proofs. Specifically, we present the tree-determinant expansion in Section~\ref{sec:determinant-expansion}, prove the convergence of the cumulant expansion (Theorem~\ref{thm-convergence}) in Section~\ref{sec:convergence}, describe the polynomial-time algorithm and bound its runtime in Section~\ref{sec:alg-partition}, extend the result to computing observable expectation values in Section~\ref{sec:alg-observable}, and discuss the computation of the non-interacting Green's function in Section~\ref{sec:compute_non_int_greens_func}. Important technical details, such as the derivation of the tree-determinant expansion, the determinant bound, the decay of the non-interacting Green's function, and the sampling algorithms used in Theorem~\ref{thm:ft-informal-draft} (BP and tree sampling) are included in Section~\ref{sec:technicaldetails}.

\section{Problem Setup and Notations}\label{sec.2}

We consider a fermionic system on a lattice $\Lambda \subset \mathbb{Z}^d$. The single-particle state space is given by $\Omega = \Lambda \times \{\uparrow, \downarrow \}$, where $\{\uparrow, \downarrow \}$ denote the spin degrees of freedom. The total number of single-particle states is $|\Omega| = N$. The system is described by a unital algebra generated by the identity operator $I$ and fermionic creation ($\psi_a^\dag$) and annihilation ($\psi_a$) operators for each state $a \in \Omega$. These operators satisfy the canonical anticommutation relations: for any two $a,b\in\Omega$,
$$
\{\psi_a, \psi_b^\dag\} = \delta_{ab}, \quad \{\psi_a, \psi_b\} = 0, \quad \{\psi_a^\dag, \psi_b^\dag\} = 0,
$$
The total Hamiltonian of the system is decomposed as
$$ H = H_0 + V, $$
where $H_0$ is the non-interacting part and $V$ represents the interaction potential. The free Hamiltonian $H_0$ is quadratic in the fermion operators and is given by
\begin{align}\label{eq:hdef}
 H_0 = \sum_{a,b \in \Omega} h_{ab} \psi_a^\dag \psi_b, 
 \end{align}
where $h = (h_{ab})_{a,b\in \Omega}$ is a Hermitian matrix. The interaction potential $V$ is given by
$$ V = \sum_{P \in \mathcal{P}} v_P \Psi_P. $$
Here, $\mathcal{P}$ is a list of multi-indices enumerating the allowed interactions. Each $P \in \mathcal{P}$ defines a term involving an equal number, $m_P$, of creation and annihilation operators, where $1 \le m_P \le {M}$. The index $P$ specifies an ordered list of sites for the creation operators, $P^+ = (P^+(1), \dots, P^+(m_P)){\in \Omega^{m_P}}$, and a corresponding list for the annihilation operators, $P^- = (P^-(1), \dots, P^-(m_P)){\in \Omega^{m_P}}$. The operator $\Psi_P$ is the product of these operators in normal order:
$$ \Psi_P := \left(\prod_{j=1}^{m_P} \psi_{P^+(j)}^\dag\right) \left(\prod_{j=1}^{m_P} \psi_{P^-(j)}\right). $$

\begin{defn}[Maximal Interaction Size]
\label{defn:max_int_size}
    We define the \emph{maximal interaction size} of $V$ to be $M=\max_{P\in \mathcal{P}}m_P$.
\end{defn}
Although our framework is developed for general many-body interactions, the physically motivated case of primary interest is that of four-fermion interactions ($m_P=M=2$), which can describe the Coulomb interaction between electrons. The general formulation is maintained for two reasons: first, the general case may be of interest to the theoretical computer science community and is useful for comparison with related literature on quantum or classical algorithms for quantum system; and second, it provides the necessary structure for analyzing arbitrary-size observables, even for a thermal state whose Hamiltonian contains only four-body interactions, see Section \ref{sec:alg-observable}. The partition function of the system is defined by \[Z = \Tr(e^{-\beta H})\,,\] where $\beta\in [0,\infty)$ is the inverse temperature. The non-interacting partition function is $Z_0 = \Tr(e^{-\beta H_0})$.

\subsection{Expansion of the Partition Function}
The ratio of the interacting and non-interacting partition functions can be formally expanded using the interaction picture, yielding the Dyson series (see Lemma \ref{lem:dyson}):

\begin{align}\label{eq:Dyson}
\frac{Z}{Z_0} = \sum_{s=0}^\infty \frac{(-1)^s}{s!}\sum_{P_1,\cdots,P_s \in \mathcal{P}} v_{P_1}\cdots v_{P_s} \int_{[0,\beta]^s} \dd \tau_1 \cdots \dd \tau_s \langle \cT (\Psi_{P_1}(\tau_1) \cdots \Psi_{P_s}(\tau_s)) \rangle_0.
\end{align}
Here, $\Psi_P(\tau) = e^{\tau H_0} \Psi_P e^{-\tau H_0}$ is the operator in the imaginary-time interaction picture. The operator $\cT$ enforces time-ordering: if the permutation $\sigma$ satisfies ${ \tau_{\sigma(1)} \ge \tau_{\sigma(2)}\ge \cdots \ge \tau_{\sigma(s)}}$, then 
$$
\cT (\Psi_{P_1}(\tau_1) \cdots \Psi_{P_s}(\tau_s)) = \Psi_{P_{\sigma(1)}}(\tau_{\sigma(1)}) \cdots \Psi_{P_{\sigma(s)}}(\tau_{\sigma(s)}).
$$
The expectation $\langle \cdot \rangle_0$ is taken with respect to the non-interacting thermal state $\frac{1}{Z_0}{\mathrm{Tr}(\cdot \, e^{-\beta H_0})}$.

The central objects in this expansion are the time-ordered multi-point correlation functions, which we denote by $\mathcal{E}$:
$$
\mathcal{E}(\{P_i,\tau_i\}_{i\in[s]}) := \langle \cT (\Psi_{P_1}(\tau_1) \cdots \Psi_{P_s}(\tau_s)) \rangle_0.
$$
Since the expectation is over a Gaussian state, these correlators can be evaluated via Wick's theorem, which expresses them in terms of two-point correlators corresponding to a time-ordered Green's function.

\begin{defn}[Non-Interacting Green's Function]\label{def:Green}
The (anti-symmetric) time-ordered Green's function\footnote{Since we will not use the interacting Green's function in this work, we sometimes refer to the non-interacting Green's function simply as the Green's function.} is given by
{ 
\begin{align*}
g_\tau(a,b)  =\begin{cases} -\langle  \psi_b^\dag \psi_a(\tau)  \rangle_0, & \tau < 0 \\\langle \psi_a(\tau)\psi_b^\dag \rangle_0, & \tau \ge0 \end{cases} ,\quad \forall \,a,b \in \Omega.
\end{align*}
}
It also admits a matrix representation (see Lemma \ref{lem:Greenfunction}):

$$ g_\tau :={-} \mathbf{1}_{\tau {<} 0}e^{-\tau h}(1+e^{\beta h})^{-1} {+} \mathbf{1}_{\tau{\ge} 0} e^{-\tau h}(1+e^{{-}\beta h})^{-1},  $$ 
where we denoted the identity matrix over $\mathbb{C}^{\Omega\times\Omega}$ by $1$. Here, $g_\tau(a,b)$ is the $(a,b)$-th element of the matrix $g_\tau$, and $h$ is the matrix of coefficients encoding the interactions in the unperturbed Hamiltonian $H_0$ in \eqref{eq:hdef}.
\end{defn}

\noindent To apply Wick's theorem to the correlator $\mathcal{E}(\{P_i,\tau_i\}_{i\in[s]})$, all required pairwise contractions are organized into the determinant of a matrix built from the Green's function. More precisely, we define the Green's function matrix $\mathbf{G}$ by the following block form:
\begin{align}\label{green-matrix}
\mathbf{G}(\{P_i,\tau_i\}_{i \in [s]}):= \begin{bmatrix}
G_0(P_1^-,P_1^+) & G_{\tau_1-\tau_2}(P_1^-,P_2^+) & \cdots & G_{\tau_1-\tau_s}(P_1^-,P_s^+) \\
G_{\tau_2-\tau_1}(P_2^-,P_1^+) & G_{0}(P_2^-,P_2^+) & \cdots & G_{\tau_2-\tau_s}(P_2^-,P_s^+) \\
\vdots & \vdots & \ddots & \vdots \\
G_{\tau_s-\tau_1}(P_s^-,P_1^+)& G_{\tau_s-\tau_2}(P_s^-,P_2^+) & \cdots & G_{0}(P_s^-,P_s^+)
\end{bmatrix},
\end{align}
where each block $G_{\tau}(P^{{-}}_i,P^{{+}}_j)$ is a $m_{P_i} \times m_{P_j}$ matrix corresponding to interactions between the annihilation operators in $P_i$ and the creation operators in $P_j$. Its entries are given by: 
\begin{align}\label{eq:Gtog}
\big(G_{\tau_i-\tau_j}(P^{{-}}_i,P^{{+}}_j)\big)_{kl} := g_{\tau_i-\tau_j}(P_i^-(k), P_j^+(l) ), \quad \text{for } 1\le k \le m_{P_i},\, 1\le l \le m_{P_j}.
\end{align}
The application of Wick's theorem then yields the identity (see Lemma \ref{lemm:Wick}):

\begin{align}\label{Wickidentity}
\mathcal{E}(\{P_i,\tau_i\}_{i \in [s]}) = (-1)^{\sum_{i=1}^sm_{P_i}(m_{P_i} - 1)/2}\det \mathbf{G}(\{P_i,\tau_i\}_{i \in [s]}).
\end{align}
The significance of this identity lies in expressing the many-body correlator as a determinant, which induces sign cancellations.

\subsection{Expansion of the Log-Partition Function}
By the linked-cluster theorem, the log-partition function can be written as the following series expansion (see Lemma \ref{lem:dysonlog}).
 \begin{align} \label{eq:cumulant_expansion}
 \log (Z/Z_0)= \sum_{s=1}^\infty \frac{(-1)^{s}}{s!} \sum_{P_1,\cdots,P_s \in \mathcal{P}} v_{P_1}\cdots v_{P_s} \int_{[0,\beta]^s} \dd \tau_1\cdots \dd \tau_s  \,\mathcal{E}_{c}(\{P_i,\tau_i\}_{i\in[s]})
 \end{align}
  where the cumulant function $\mathcal{E}_c$ is defined as
\begin{align}\label{eq:defcumulanttext}
\mathcal{E}_{c}(\{P_i,\tau_i\}_{i\in[s]}):=\sum_{\Pi\in\mathbf{P}_s}(-1)^{|\Pi|-1}(|\Pi|-1)!\prod_{B=\{j_1,\cdots , j_{|B|}\}\in\Pi}\langle \mathcal{T}(\Psi_{P_{j_1}}(\tau_{j_1})\cdots \Psi_{P_{j_{|B|}}}(\tau_{j_{|B|}}) )\rangle_0,
\end{align}
where $\mathbf{P}_s$ 
stands for the set of partitions of $\{1,\cdots ,s\}$ and $|\Pi|$ denotes the number of sets in a partition $\Pi$. Moreover, $\mathcal{E}_c$ is the connected part of the moment function $\mathcal{E}$: for instance, given a potential $V=\Psi_{P_1}+\Psi_{P_2}$ with $\mathcal{E}(\{P_i,\tau_i\}_{i\in[2]}):=\langle\mathcal{T}( \Psi_{P_1}(\tau_1)\Psi_{P_2}(\tau_2))\rangle_0$, 
\begin{align*}
\mathcal{E}_c(\{P_i,\tau_i\}_{i\in[2]})=\mathcal{E}(\{P_i,\tau_i\}_{i\in[2]})-\langle \Psi_{P_1}(\tau_1)\rangle_0\langle \Psi_{P_2}(\tau_2)\rangle_0.
\end{align*}
In general, $\mathcal{E}_c$ satisfies:
\begin{align}\label{eq:connectedtext}
\mathcal{E}(\{P_i,\tau_i\}_{i\in[s]}) = \sum_{\Pi\in \mathbf{P}_s} \prod_{B\in\Pi} \mathcal{E}_c(\{P_i,\tau_i\}_{i\in B}).
\end{align}

\subsection{Locality and Decay of the Non-Interacting Green's Function}

Our analysis relies on a set of assumptions addressing the locality of the interaction potential $V$ and the {decay of correlations for the non-interacting Green's function $g$}:
First, we assume the size of all interaction terms is uniformly bounded, i.e., $m_P \le M$ for all $P \in \mathcal{P}$ and $M=\Or(1)$ (i.e., the maximal interaction size as defined in Definition~\ref{defn:max_int_size} is a constant). We also impose conditions to control the range and strength of the interactions. The following definitions cover two main cases: a strict, finite-range locality and a more general summability condition that can accommodate long-range interactions provided their strength decays sufficiently fast.

\begin{defn}(Distance on the Lattice)
 For $a,b\in \Omega$ with $a=(x_a,\sigma_a)$, $b=(x_b,\sigma_b)$ with $x_a,x_b\in\Lambda$, $\sigma_a,\sigma_b\in\{\uparrow,\downarrow\}$, we define the distance $\mathrm{dist}(a,b)$ as the graph distance between $x_a$ and $x_b$ on the lattice $\Lambda$. For two subsets or list $A,B \subset \Omega$, we define 
 $$
 \mathrm{dist}(A,B) = \min_{a \in A ,b \in B} \mathrm{dist}(a,b)
 $$
 and $\mathrm{dist}(a,B)=\mathrm{dist}(\{a\},B)$ for $a\in\Lambda$ and $B\subseteq \Lambda$. With a slight abuse of notation, we also apply the notation $\mathrm{dist}(\cdot, \cdot)$ to lists of elements of $\Omega$, where the distance is taken over the associated sets.
\end{defn}

 \begin{defn}(Conditions on $V$)\label{defn:conditions-on-v}
\begin{enumerate}
    \item We say the potential $V$ is \textbf{$(\mathsf{n},r_0)$-geometrically local} if each interaction is supported on a ball of radius $r_0$, and each site $x \in \Omega$ is involved in at most $\mathsf{n}$ distinct interaction terms.

    \item We say the potential $V$ is \textbf{$L_V$-summable} if the total interaction strength experienced by any single site is bounded by $L_V$. Formally, we require
    $$
    \max_{x \in \Omega} \sum_{P \in \mathcal{P},\,P^-\owns x} |v_P| \,{\le \frac{L_V}{2}},\quad  \max_{y\in \Omega} \sum_{P \in \mathcal{P},\,P^+\owns y}|v_P| {\le \frac{L_V}{2}}.
    $$
\end{enumerate}

\end{defn}

We note that the geometrically local condition implies the summable condition when the interaction strength is uniformly bounded.  Specifically, if $V$ is $(\mathsf{n},r_0)$-geometrically local and the interaction strengths are bounded by $|v_P| \le U$, then $V$ is $L_V$-summable with $L_V \le {2U\mathsf{n}}$. Throughout our analysis, the $L_V$-summable property serves as the general assumption for proving the convergence bounds and designing polynomial-time algorithms. However, when the stronger geometrically local condition holds, the algorithm can be simplified, leading to a more favorable computational complexity.  

 {
Next, we impose a condition on the non-interacting part $H_0$. Again, we consider two primary cases: 
a finite-range locality and a decay condition on the non-interacting Green's function.
\begin{defn} (Conditions on $H_0$)\label{defn:conditions-on-h0}
\begin{itemize}
\item We say $H_0$ is $r_1$-finite-ranged, if  $|h_{ab}| \le 1$ and $h_{ab} = 0$ unless $\mathrm{dist}(a,b) < r_1$.
\item We say that the non-interacting Green's function is \textbf{$(K,\xi)$-exponentially decaying} if
 $$ \max_{\tau \in [-\beta,\beta]}|g_\tau(a,b)| \le Ke^{-\mathrm{dist}(a,b)/\xi} .$$ 
 We say the Green's function is \textbf{$L_g$-summable} if 
 $$ \max_{\tau \in [-\beta,\beta]}\,\max_{a \in \Omega} \,\sum_{b \in \Omega}|g_{\tau}(a,b)| \le L_g.$$
 \end{itemize}
\end{defn}
}
\noindent  Note that the $(K,\xi)$-exponential decay condition for the Green's function implies it is also $L_g$-summable, with $L_g \le C_d K\xi^d$ for some constant $C_d$ that only depends on the lattice dimension. While the general summability condition is sufficient to prove convergence and establish a polynomial-time algorithm, the stronger exponential decay condition allows for a significant simplification of the algorithm.

 The decay properties of the Green's function are a direct consequence of the locality structure of the non-interacting Hamiltonian $H_0$. The decay of $g$ arises if the matrix elements $h_{ab}$ decay sufficiently fast \cite{Hastings2004decay,lin2016decayestimatesdiscretizedgreens}.  The following lemma provides two concrete sufficient conditions on $H_0$ that guarantee decay, one for strictly finite-range Hamiltonians and a more general condition for Hamiltonians with long-range, but rapidly decaying, matrix elements (see Corollaries \ref{lemma-green-decay1} and \ref{secondpartlemmadecay} below for a proof):
 \begin{lemma} \label{lemma-green-decay}
If the non-interacting part is $r_1$-finite-ranged, then the Green's function is $(K,\xi)$-exponentially decaying with 
$$
 K = O(\beta r_1^d)\qquad\text{ and } \qquad  \xi =O(\beta r_1^{d+1}).
$$
More generally, if $H_0$ satisfies 
$$
 \max_{a \in \Omega}\sum_{b\in \Omega} |h_{ab}| (e^{\theta \mathrm{dist}(a,b)} -1) \le \frac{\pi}{4\beta\|h\|},
$$
for some constant $\theta > 0$, the Green's function is $(K,\xi)$-exponentially decaying with
$$
K = O(\beta\|h\|^2)\qquad\text{ and }\qquad  \xi = \theta^{-1}.
$$
\end{lemma}

\paragraph{Scaling considerations} Throughout our analysis, the model-defining parameters, including the maximum interaction size $M$, the lattice dimension $d$, and others introduced for specific conditions (whenever we assume the corresponding condition holds), namely, $L_V, L_g, \mathsf{n}, r_0, r_1, K, \xi$ are all assumed to be constants, independent of the system size $N$. 
We use the following asymptotic notations allowing us to distinguish between different dependencies:
\begin{itemize}
\item $O(\cdot),\,\Omega(\cdot)$: These denote bounds that hold up to a universal constant.
\item $O_s(\cdot)$: Here, the subscript $s$ indicates that the implicit constant may depend on the specified parameter (e.g., the lattice dimension $d$).
\item $\mathcal{O}(\cdot)$: This denotes a bound where the implicit constant may depend on $M,\,d$, and any other model parameters ($L_V, L_g, \mathsf{n}, r_0, r_1, K, \xi$) that are assumed to hold for the specific statement being made.
\end{itemize}

\section{Main Results}
\label{sec:main_results}

We are now ready to state our main results, namely:
\begin{itemize}
    \item We prove that under general summability conditions for both the potential and the Green's function, the cumulant expansion converges exponentially fast within a parameter regime where the product $\beta L_V L_g < C_M$ for certain constant $C_M$ that only depends on $M$.
    \item 
    In addition, in the regime where the convergence of cumulant expansion is proved, we develop an importance sampling algorithm to approximate the log-partition function and the expectation of local observables. We prove that:
\begin{itemize}
\item Under the general summability conditions, there exists a classical algorithm for computing the log-partition function to within $N\epsilon$ error (equivalent to computing the free energy per site to within $\Or(\epsilon)$ error) or expectations of local observables within $\epsilon$ error in $\mathcal{O}(|\mathcal{P}|^2N^2\polylog(1/\epsilon)/\epsilon^2)$ time.
\item Under the strictly local condition on both the potential $V$ and the non-interacting part $H_0$, the algorithm's complexity can be improved significantly. The log-partition function can be computed in $\mathcal{O}(N \epsilon^{-2}{\mathrm{polylog}(N}/\epsilon))$ time in general or $\mathcal{O}(\epsilon^{-2}{\mathrm{polylog}(N}/\epsilon))$ for translation invariant systems, to within $N\epsilon$ error. The local observables can be computed in {$\mathcal{O}(\epsilon^{-2}{\mathrm{polylog}}(N/\epsilon))$ time.} 
\end{itemize}
\end{itemize}
The formal statements are given below, and their proofs are presented in the remainder of this section. 

\begin{theorem} \label{thm-convergence}
Suppose the non-interacting Green's function is $L_g$-summable, the potential $V$ is $L_V$-{summable} and the interacting terms are at most $M$-body (i.e., the maximal interaction size is $M$ as defined in Definition~\ref{defn:max_int_size}). Then the $s$-th order term in the cumulant expansion \eqref{eq:cumulant_expansion} is bounded by 
$$
\frac{1}{s!}\sum_{P_1,\cdots,P_s \in \mathcal{P}}|v_{P_1}\cdots v_{P_s} | \int_{[0,\beta]^s} \dd \tau_1 \cdots \dd\tau_s\left| \mathcal{E}_{c}(\{P_i,\tau_i\}_{i \in [s]})\right| \le N (C\beta  M{2^{2M}}L_g L_V)^s.
$$
for some absolute constant $C$. 
 Consequently, 
 whenever $C\beta M{2^{2M}} L_VL_g < 0.99$, by choosing $S =\Omega(\log(1/\epsilon))$ the truncation error of the expansion is bounded by
$$
\left|\log (Z/Z_0) - \sum_{s=1}^S \frac{(-1)^{s}}{s!} \sum_{P_1,\cdots,P_s \in \mathcal{P}} v_{P_1}\cdots v_{P_s} \int_{[0,\beta]^s} \dd \tau_1\cdots \dd \tau_s  \,\mathcal{E}_{c}(\{P_i,\tau_i\}_{i\in[s]})\right| \le \epsilon.
$$

\end{theorem}

This convergence result provides some of the necessary components of the efficient algorithm to compute the log-partition function.

In Theorems \ref{thm-algorithm-general} and \ref{thm-algorithm-geo-local} below, we first establish the efficiency of our algorithm using a query complexity model. Specifically, we assume exact access to the non-interacting Green’s function $g_{ab}(\tau)$ and the interaction coefficients $V_P$, counting only the total number of queries while excluding the cost of computing individual entries. We then combine this query count with the cost of specific evaluations and a stability analysis in Corollary \ref{cor-finite-range-ham} to derive the total end-to-end runtime complexity.

\begin{theorem} \label{thm-algorithm-general}
 Suppose the non-interacting Green's function is $L_g$-summable and the potential $V$ is $L_V$-{summable}, and the interacting term are at most $M$-body, then there exists an absolute constant $C>0$, such that whenever $\beta M{2^{2M}} L_VL_g < C$, a classical algorithm with query complexity $\mathcal{O}(|\mathcal{P}|^2 \epsilon^{-2} \polylog(1/\epsilon))$ can, with probability at least $2/3$, compute the log-partition function up to $N\epsilon$ error and compute the $m$-body observables with $m=O(1)$ within $\epsilon$ error. 
 \end{theorem}

When the potential $V$ and the non-interacting Green's function $g_\tau (a,b)$ have stronger locality properties, the query complexity of our algorithm can be further reduced:

\begin{theorem} \label{thm-algorithm-geo-local}
In addition to the assumptions in Theorem~\ref{thm-algorithm-general}, we assume that the potential $V$ is $(\mathsf{n},r_0)$-geometrically local and the Green's function is $(K,\xi)$-exponentially decaying, then there exists an absolute constant $C>0$, such that whenever $\beta M{2^{2M}} L_VL_g < C$, there exists a classical algorithm with query complexity $\mathcal{O}(N \epsilon^{-2}{\mathrm{\polylog}}(1/\epsilon)) $ for computing the log-partition function within $N\epsilon$ error with probability at least $2/3$, and the query complexity can be reduced to $\mathcal{O}(\epsilon^{-2}\mathrm{\polylog}(1/\epsilon))$ for translation-invariant systems.
 There is a classical algorithm with query complexity { $\mathcal{O}(\epsilon^{-2}\mathrm{\polylog}(1/\epsilon))$} for computing local observables within $\epsilon$ error with probability at least $2/3$.

\end{theorem}

To derive the total runtime complexity, we combine the query complexity established above with the computational cost of realizing the non-interacting Green's function oracle. In the general case, the complexity of computing a single entry $g_\tau(a,b)$ scales as $\mathcal{O}(N^2 \log(1/\epsilon))$; thus, the final complexity is the product of the query count and this $\mathcal{O}(N^2 \log(1/\epsilon))$ overhead. However, when the non-interacting Hamiltonian $H_0$ is finite-ranged, we achieve a significant reduction in complexity through two distinct mechanisms. First, the sparsity of $H_0$ allows individual entries of the Green's function to be computed in $\mathcal{O}(\mathrm{polylog}(N/\epsilon))$ time. Second, as established in Lemma \ref{lemma-green-decay}, a finite-ranged $H_0$ guarantees that the Green's function decays exponentially. This decay satisfies the stronger assumptions of Theorem \ref{thm-algorithm-geo-local}, allowing us to utilize the reduced query complexity scaling. These combined factors lead to the following corollary:

\begin{corollary}
    \label{cor-finite-range-ham}
\begin{itemize}
    \item Under the assumptions of Theorem~\ref{thm-algorithm-general}, there exists an absolute constant $C>0$, such that whenever $\beta M{2^{2M}} L_VL_g < C$, an $\mathcal{O}(|\mathcal{P}|^2N^2 \epsilon^{-2} \polylog(N/\epsilon))$ classical algorithm can, with probability at least $2/3$, compute the log-partition function up to $N\epsilon$ error and compute the $m$-body observables with $m=O(1)$ within $\epsilon$ error.
    \item If, in addition, we assume that the potential $V$ is $(\mathsf{n},r_0)$-geometrically local and the non-interacting part $H_0$ is $r_1$-finite-ranged, then there exists an absolute constant $C>0$, such that whenever $\beta M{2^{2M}} L_VL_g < C$, there exists an $\mathcal{O}(N \epsilon^{-2}{\mathrm{\polylog}}(N/\epsilon)) $-time 
 classical algorithm for computing the log-partition function within $N\epsilon$ error with probability at least $2/3$, and the runtime can be reduced to $\mathcal{O}(\epsilon^{-2}\mathrm{\polylog}(N/\epsilon))$ for translation-invariant systems.
 There is a { $\mathcal{O}(\epsilon^{-2}\mathrm{\polylog}(N/\epsilon))$}-time classical algorithm for computing local observables within $\epsilon$ error with probability at least $2/3$.
\end{itemize}
\end{corollary}

\noindent In the following sections, we present our main approach and provide the proofs for the main theorems. In Section \ref{sec:determinant-expansion}, we introduce the tree-determinant expansion for the cumulant function, whose algebraic structure is amenable to rigorous bounds. Using this representation, we prove the convergence of the expansion under general summability conditions {(Theorem \ref{thm-convergence})} in Section \ref{sec:convergence}. 
{Theorems \ref{thm-algorithm-general} and \ref{thm-algorithm-geo-local} are the subject of the subsequent sections}: in Section \ref{sec:alg-partition}, we develop an algorithm to compute the log-partition function by evaluating the truncated cumulant expansion and analyze its query complexity. We then discuss how to simplify this algorithm in the case of a geometrically local potential and an exponentially decaying Green's function. In Section \ref{sec:alg-observable}, we extend this approach to compute local observables. This is achieved using an identity that connects the expectation of an observable to the log-partition function.  { While in Sections \ref{sec:alg-partition} and \ref{sec:alg-observable} we assume an exact access to the non-interacting Green's function and only consider the query complexity model, we complete the end-to-end analysis (Corollary \ref{cor-finite-range-ham}) in Section \ref{sec:compute_non_int_greens_func}. There, we detail the algorithm for approximating the non-interacting Green's function and establish the stability of our algorithm.}  The technical details supporting these proofs are deferred to Section \ref{sec:technicaldetails}. 
\subsection{The Tree-Determinant Expansion}\label{sec:determinant-expansion}
Our key analytical and algorithmic tool is a tree-determinant expansion for the cumulant function $\mathcal{E}_c$, which reorganizes the perturbative series into a sum over tree structures \cite{GentileMastropietro2001renormalization}, {and whose proof is postponed to Section \ref{sec:determunantexpansion} (see Theorem \ref{cumulant-representation}):} 
\begin{equation}
\label{eq:determinant_expansion}
    \mathcal{E}_c(\{P_i,\tau_i\}_{i\in[s]} ) = \sum_{T \in \mathcal{T}([s])} \sum_{\chi \in \mathcal{A}(T)} \alpha_{T,\chi}\prod_{(i,j)\in T} g_{\tau_i,\tau_j}(P_i,P_j,\chi_{ij})\, h_{\boldsymbol{\tau}}(P_1,\cdots,P_s,T,\chi),
\end{equation}
where 
\begin{equation}
\label{eq:determinant_integral}
    h_{\boldsymbol{\tau}}(P_1,\cdots,P_s,T,\chi) =\sum_{\omega \in S(T)}\int_{[0,1]^{s-1}} \dd \mathbf{t}\, p_{T,\omega}(\mathbf{t}) \det \mathbf{G}(T,\chi,\omega,\mathbf{t},\{P_i,\tau_i\}_{i\in[s]}). 
\end{equation}
This formula re-expresses the cumulant into a sum over tree structures connecting the $s$ interaction terms $\{P_1, \dots, P_s\}$.  The contribution of each tree in the sum is then factorized into two components. The first is the product of Green's functions $g_{{\tau_i,\tau_j}}$ representing the propagators along the tree's edges. The second component, $h_{\boldsymbol{\tau}}$, encapsulates the remaining interactions into a linear combination of determinants, a structure essential for managing the series' convergence.
Let us now break down the components of this expansion.

\paragraph*{Tree Structures and Contractions}

\begin{itemize}
\item \textbf{Sum over Trees ($T \in \mathcal{T}([s])$):} The outer sum is over $\mathcal{T}([s])$, the set of all \textbf{labeled (undirected) trees} on the vertex set $[s] = \{1, \dots, s\}$. Each vertex $i$ corresponds to an interaction term $\Psi_{P_i}$, and each tree $T$ represents a ``structural backbone'' of connections between them. This is a crucial combinatorial simplification. By Cayley's formula, the number of such trees is $|\mathcal{T}([s])| = s^{s-2}$ (see Section \ref{sec:prufer}). This growth is far more manageable than the double-factorial growth ($\sim (2s)!$) of connected Feynman diagrams. When combined with the $1/s!$ prefactor from the cumulant expansion, the scaling becomes exponential $s^{s-2}/s! \sim e^s$ rather than factorial. 
We say $(i,j) \in T$ if $i<j$ and sites $P_i$ and $P_j$ are connected by $T$. {Though the tree is undirected, we always assume $i<j$ when we say $(i,j) \in T$ for convenience.}
    \item \textbf{Sum over Contraction Assignments ($\chi \in \mathcal{A}(T)$) and Propagators $g_{\tau_i,\tau_j}(P_i,P_j,\chi_{ij})$:}  For each edge $(i,j) \in T$, the assignment $\chi$ selects an annihilation operator from one vertex (e.g., $P_i$) and a creation operator from the other ($P_j$), or vice versa, with the constraint that each operator can only be contracted once. The propagator $g_{\tau_i,\tau_j}(P_i,P_j,\chi_{ij})$ then denotes the Green's function between the two contracted operators. The tree $T$ and the assignment $\chi$ specify a ``backbone'' of interactions, and the product of propagators $\prod_{(i,j)\in T}g_{\tau_i,\tau_j}(P_i,P_j,\chi_{ij})$ gives the amplitude of the specific structural backbone. See Figure  \ref{fig:diagram} for an illustrative example.
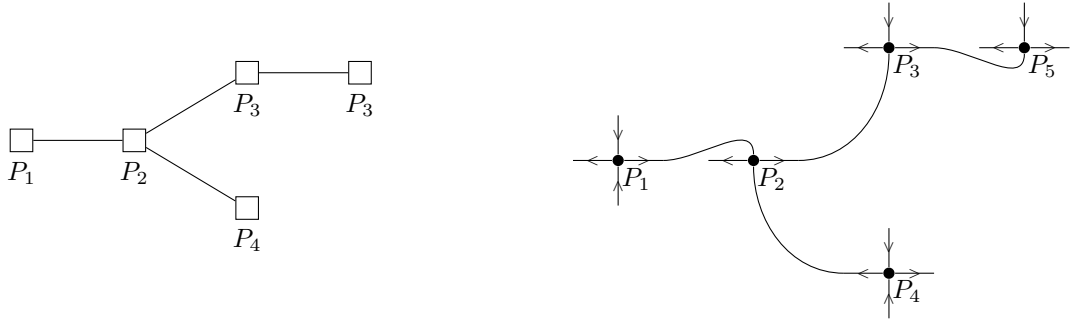
\begin{figure}[h!]
\centering

\begin{minipage}{.45\textwidth}
\centering
\begin{tikzpicture}[scale = 0.6]
    \tikzstyle{box} = [draw, rectangle, minimum height=0.3cm, minimum width=0.3cm, label distance=2pt]

    \node (n1) at (0, 0)   [box, label=below:{$P_1$}] {};
    \node (n2) at (2.5, 0) [box, label=below:{$P_2$}] {};
    \node (n3) at (5, 1.5) [box, label=below:{$P_3$}] {};
    \node (n4) at (5, -1.5) [box, label=below:{$P_4$}] {};
    \node (n5) at (7.5, 1.5) [box, label=below:{$P_3$}] {};

    \draw (n1) -- (n2);
    \draw (n2) -- (n3);
    \draw (n2) -- (n4);
    \draw (n3) -- (n5);

\end{tikzpicture}
\end{minipage}
\hfill
\begin{minipage}{.45\textwidth}
\centering
\begin{tikzpicture}[scale=0.6]
    
    \node (n1) at (0, 0) [circle, fill=black, inner sep=1.5pt] {};
    \node (n2) at (3, 0) [circle, fill=black, inner sep=1.5pt] {};
    \node (n3) at (6, 2.5) [circle, fill=black, inner sep=1.5pt] {};
    \node (n4) at (6, -2.5) [circle, fill=black, inner sep=1.5pt] {};
    \node (n5) at (9, 2.5) [circle, fill=black, inner sep=1.5pt] {};

    \node at (0.4, -0.4) {$P_1$};
    \node at (3.4, -0.4) {$P_2$};
    \node at (6.4, 2.1) {$P_3$};
    \node at (6.4, -2.9) {$P_4$};
    \node at (9.4, 2.1) {$P_5$};

    \draw (n1) -- ++(-1,0) node[midway, sloped] {$\scriptscriptstyle <$}; 
    \draw (n1) -- ++(0,1)  node[midway, sloped] {$\scriptscriptstyle <$};
    \draw (n1) -- ++(1,0)  node[midway, sloped] {$\scriptscriptstyle >$};  
    \draw (n1) -- ++(0,-1) node[midway, sloped] {$\scriptscriptstyle <$};

    \draw (n2) -- ++(-1,0) node[midway, sloped] {$\scriptscriptstyle <$}; 
    \draw (n2) -- ++(1,0)  node[midway, sloped] {$\scriptscriptstyle >$};  

    \draw (n3) -- ++(-1,0) node[midway, sloped] {$\scriptscriptstyle <$}; 
    \draw (n3) -- ++(0,1)  node[midway, sloped] {$\scriptscriptstyle >$};
    \draw (n3) -- ++(1,0)  node[midway, sloped] {$\scriptscriptstyle >$};  

    \draw (n4) -- ++(1,0)  node[midway, sloped] {$\scriptscriptstyle >$};
    \draw (n4) -- ++(-1,0) node[midway, sloped] {$\scriptscriptstyle <$};  
    \draw (n4) -- ++(0,-1) node[midway, sloped] {$\scriptscriptstyle >$};
    \draw (n4) -- ++(0,1)  node[midway, sloped] {$\scriptscriptstyle >$};

    \draw (n5) -- ++(-1,0) node[midway, sloped] {$\scriptscriptstyle <$}; S
    \draw (n5) -- ++(1,0)  node[midway, sloped] {$\scriptscriptstyle >$};  
    \draw (n5) -- ++(0,1)  node[midway, sloped] {$\scriptscriptstyle >$};

    \draw (n1) ++(1,0) to[out=0, in=90] (n2) ++(0,-1) node[pos=0, sloped] {$\scriptscriptstyle >$};
    \draw (n2) ++(1,0) to[out=0, in=-90] (n3) ++(0,-1); 
    \draw (n4) ++(-1,0) to[out=180, in=-90] (n2) ++(0,-1); 
    \draw (n3) ++(1,0) to[out=0, in=-90] (n5) ++(0,-1);

\end{tikzpicture}
\end{minipage}

\caption{An example of a ``backbone'' of the interactions specified by the tree $T$ and an assignment $\chi$ in case $m_{P_i} =2$ at the order $s =5$. (Left) Each box corresponds to an interacting term. The tree structure specifies how the interactions among boxes happen; (Right) Each vertex denotes an interacting term. The inner legs and the outer legs of the vertices correspond to the annihilation and creation operators respectively. For each edge $(i,j) \in T$, the assignment specifies which pair of operators on $P_i$ and $P_j$ are contracted. Each leg (operator) can only be contracted once.  }
\label{fig:diagram}
\end{figure}

Formally, an assignment $\chi$ maps each edge $(i,j) \in T$ to an element $\chi_{ij} \in \{ (\sigma,k,l):\sigma \in \{\pm1\},\,1\le k \le m_{P_i},\,1\le l \le m_{P_j}\}$. { Here, the sign $\sigma \in \{+1, -1\}$ specifies the direction of the contraction, while the labels $(k,l)$ identify the indices of the operators within their respective lists: if $\sigma = +1$, the contraction is from an annihilation operator in $P_j$ to a creation operator in $P_i$, specifically from $P_j^-(l)$ to $P_i^+(k)$. Conversely, if $\sigma = -1$, the contraction is from $P_i^-(k)$ to $P_j^+(l)$.} The set $\mathcal{A}(T)$ contains all possible ways to select $(\chi_{ij})_{(i,j)\in T}$ such that a single creation or annihilation operator $P_i^\sigma(k)$ is selected at most once. The propagator between the selected operators is then given by
  \begin{align} \label{eq:gxitog}
  g_{\tau_i,\tau_j}(P_i,P_j,\chi_{ij}) =\begin{cases}
   &  g_{\tau_i-\tau_j}(P_i^-(k),P_j^+(l)),\,\quad \quad\chi_{ij} := ({-1},k,l), \\
     &  g_{\tau_j-\tau_i}(P_j^-(l),P_i^+(k)),\quad \quad\chi_{ij} = ({+1},k,l).
        \end{cases}
    \end{align}
    {
  In other words, each $\chi_{ij}$ selects a single entry, $g_{\tau_i,\tau_j}(P_i,P_j,\chi_{ij})$, from the corresponding Green's function blocks in \eqref{green-matrix} (i.e., an element from $G_{\tau_i-\tau_j}(P_i^-,P_j^+)$ or $G_{\tau_j-\tau_i}(P_j^-,P_i^+)$). The entire set of selections $\{\chi_{ij}\}$ is constrained such that no row or column is selected more than once.}
 
    \item \textbf{The Sign:} $\alpha_{T,\chi}$ is a sign factor ($\pm 1$) that depends on the tree $T$ and the contraction assignment $\chi$.  While this sign does not affect the convergence analysis, it must be included in the algorithm. Its precise definition is 
$$
\alpha_{T,\chi} = (-1)^{s-1}\prod_{i=1}^s \alpha_{P_i} \prod_{(i,j) \in T}\alpha_{\chi_{ij}}, 
$$
where $\alpha_{P_i} = (-1)^{m_{P_i}(m_{P_i}-1)/2}$ and $\alpha_{\chi_{ij}} = (-1)^{\sum_{r=1}^{i-1}m_{P_r} + \sum_{r=1}^{j-1}m_{P_r} +k+l }$ if $\chi_{ij} = (\sigma,k,l)$ for some $\sigma \in \{\pm 1\}$.
\end{itemize}

\paragraph*{The Remainder Determinant and Growing Paths}
\begin{itemize} 
    \item \textbf{Sum over Growing Paths ($\omega \in S(T)$):} A growing path $\omega$ is a permutation of the tree's vertices that respects the parent-to-child order. Fixing 1 as the root, a growing path requires $\omega(1) = 1$ and that for any vertex $i$, its unique parent $a(i)$ in the tree must have appeared earlier in the path, i.e.,  $a(i) \in \{\omega(1), \dots, \omega(i-1)\}$. This provides an ordered way to ``build'' the tree from the root 1. We use $S(T)$ to denote the set of all growing paths.
    \item \textbf{Weighting and Integration ($\mathbf{t},p_{T,\omega}(\mathbf{t})$):} The average over growing paths is defined using interpolation parameters, $\mathbf{t} \in [0,1]^{s-1}$, and the path-dependent function, $p_{T,\omega}(\mathbf{t})$ defined as follows: If $b_k$ is the number of available choices for the { $k+1$}-th vertex in a growing path given the first { $k$} vertices $\omega(1),\cdots,\omega(k) $, then $p_{T,\omega}(\mathbf{t}) = t_1^{{b_1-1}} \cdots t_{s-1}^{{b_{s-1}-1}}$. The function $p_{T,\omega}$ integrated over $\mathbf{t}$ assigns weights for each growing path $\omega $ according to 
    $$\int_{[0,1]^{s-1}} \dd \mathbf{t}\, p_{T,\omega}(\mathbf{t}) = \frac{1}{b_1 b_2 \cdots b_{s-1}}\,.$$

    One can show that these weights normalize to a probability measure on the space of all growing paths, see Proposition \ref{sequence-weight} for details:
$$\sum_{\omega \in S(T)} \int_{[0,1]^{s-1}} \dd \mathbf{t}\, p_{T,\omega}(\mathbf{t}) = 1.$$
 This structure is highly advantageous because it means the complex average over all paths can be performed via a simple sequential sampling procedure. To draw a path according to this distribution, one can iteratively construct the vertex sequence by choosing the next vertex uniformly at random from all available candidates at each step. 
    \item \textbf{The Determinant $\det \mathbf{G}(T,\chi,\omega,\mathbf{t},\{P_i,\tau_i\})$:} The determinant is taken over a matrix representing the correlations among the fermions not used to form the tree connections. To construct this matrix, we begin by defining a general weighted Green's function matrix, where the weights $a(\omega,\mathbf{t})_{ij}$ depend on the specific growing path and interpolation parameters, see Equation \eqref{eq:weightsa} below. More precisely, the matrix $ \mathbf{G}(T,\chi,\omega,\mathbf{t},\{P_i,\tau_i\}) $ is the submatrix of the weighted Green's function matrix
     \begin{align}\label{eq:matrixGdettobound}
     \begin{bmatrix}
a(\omega,\mathbf{t})_{11}G_0(P_1^-,P_1^+) & a(\omega,\mathbf{t})_{12}G_{\tau_1-\tau_2}(P_1^-,P_2^+) & \cdots & a(\omega,\mathbf{t})_{1s}G_{\tau_1-\tau_s}(P_1^-,P_s^+) \\
a(\omega,\mathbf{t})_{21}G_{\tau_2-\tau_1}(P_2^-,P_1^+) & a(\omega,\mathbf{t})_{22}G_{0}(P_2^-,P_2^+) & \cdots &a(\omega,\mathbf{t})_{2s} G_{\tau_2-\tau_s}(P_2^-,P_s^+) \\
\vdots & \vdots & \ddots & \vdots \\
a(\omega,\mathbf{t})_{s1} G_{\tau_s-\tau_1}(P_s^-,P_1^+)& a(\omega,\mathbf{t})_{s1}G_{\tau_s-\tau_2}(P_s^-,P_2^+) & \cdots & a(\omega,\mathbf{t})_{ss}G_0(P_s^-,P_s^+)
\end{bmatrix}
\end{align}
 obtained by deleting the rows and columns corresponding to the annihilation and creation operators already contracted along the tree's edges by the assignment $\chi$.

We now define the weights $a(\omega,\mathbf{t}) $. For $\omega  \in S(T)$, define $X_i = \{\omega(1),\cdots,\omega(i)\}$, representing the set of vertices that have been ``visited'' after $i$ steps of the growing path. A pair of vertices $(j,k)$ is said to cross $X_i$ if $j \in X_i,\,k \notin X_i$, or vice versa. The weight function is then defined by 
\begin{align}\label{eq:weightsa}
a(\omega,\mathbf{t})_{jk} = \prod_{\substack{1\le i\le s-1 \\ (j,k)\,\text{crosses}\,X_i}}t_i.   
\end{align}
\end{itemize}

\paragraph*{The Derivation} The tree-determinant expansion for fermionic systems originates from early work on the rigorous analysis of quantum field theory \cite{abdesselam1997explicitfermionictreeexpansions,Salmhofer2000,GentileMastropietro2001renormalization}. Its derivation utilizes two main techniques: the Grassmann integral representation of the determinant and an earlier combinatorial tool known as the tree-graph identity \cite{brydges1978new, Brydges1984short, Battle1984ANO}, which was originally motivated from the study of Euclidean field theory. A self-contained proof is provided in Section \ref{sec:determunantexpansion}.

\subsection{Convergence of the Cumulant Expansion (Proof of Theorem \ref{thm-convergence})}\label{sec:convergence}
Substituting the determinant expansion to the $s$-th order term of the cumulant expansion \eqref{eq:cumulant_expansion}, we get 
$$
\frac{(-1)^s}{s!} \sum_{T \in \mathcal{T}([s])}  \int_{[0,\beta]^s} \dd \boldsymbol{\tau} \sum_{P_1,\cdots,P_s } v_{P_1} \cdots v_{P_s}\sum_{\chi \in  \mathcal{A}(T)}\alpha_{T,\chi}\prod_{(i,j)\in T}g_{\tau_i,\tau_j}(P_i,P_j,\chi_{ij}) h_{\boldsymbol{\tau}}(P_1,\cdots,P_s,T,\chi).
$$
To establish the convergence of the cumulant expansion, we need to bound this term by $\rho^s$ for some $\rho < 1$ in the weakly interacting regime. The combinatorial prefactor $1/s!$ cancels the number of distinct trees $|\mathcal{T}([s])| = s^{s-2}$ up to an acceptable exponential factor (see Section \ref{sec:prufer}), which allows us to bound the contribution from each tree individually. The integral over the imaginary-time variable $\int_{[0,\beta]^s} \dd \boldsymbol{\tau}$ only gives a $\beta^s$ factor. Therefore, it remains to bound
$$  \sum_{P_1,\cdots,P_s } v_{P_1} \cdots v_{P_s}\sum_{\chi \in  \mathcal{A}(T)}\alpha_{T,\chi}\prod_{(i,j)\in T}g_{\tau_i,\tau_j}(P_i,P_j,\chi_{ij}) h_{\boldsymbol{\tau}}(P_1,\cdots,P_s,T,\chi) .$$
The core idea for this relies on two key steps. First, we show that the function $h_{\boldsymbol{\tau}}$ is uniformly bounded by $C^s$ for some constant $C$. Second, we control the summation over the boxes $P_1, \dots, P_s$. This summation can be viewed as a partition function of a {graphical model on the tree} $T$ and we can bound it by reformulating the summation in an order following a growing path of the tree.  

\paragraph*{Boundedness of $h_{\boldsymbol{\tau}}$}

The bound on $h_{\boldsymbol{\tau}}$ defined in Equation \eqref{eq:determinant_integral} hinges on bounding its integrand, as the weighted linear combination $\sum_{\omega \in S(T)} \int_{[0,1]^{s-1}} \dd \mathbf{t} \, p_{T,\omega}(\mathbf{t})$ is normalized. The determinant bound is accomplished using a general bound for matrices constructed from Green's functions, which we state in the following lemma.

\begin{lemma}(The determinant bound) \label{lem-det-bound} 
For any positive integer $n$, $x_1,\cdots, x_n,y_1,\cdots,y_n \in \Omega $,  $\tau_1,\cdots,\tau_n \in [{0},\beta]$, and unit vectors $u_1,\cdots,u_n \in \RR^n$, the matrix $\mathbf{G} \in \RR^{n\times n}$ with coefficients
\begin{align} \label{green-matrix-general}
\mathbf{G}_{ij} = g_{\tau_i-\tau_j}(x_i,y_j)\cdot \langle u_i,u_j\rangle
\end{align}
satisfies $|\det \mathbf{G}\,| \le 2^{2n}$.
\end{lemma}
\noindent Above, the inner product term $\langle u_i,u_j\rangle $ serves as a Gram representation of the weights $a(\omega,\mathbf{t})_{jk}$ for each block. See Proposition \ref{gram} for details. 

The standard tool for proving such determinant bounds is based on the Gram inequality. It states that if a matrix $\mathbf{G}$ can be written as a kernel matrix $\mathbf{G}_{ij} = \langle v_i,w_j\rangle_{\mathcal{H}} $ for some Hilbert space $\mathcal{H}$, then $\det \mathbf{G} \le \prod_{i,j} \|v_i\|_{\mathcal{H}}\|w_j\|_{\mathcal{H}}$. Since our matrices are constructed from the Green's function, this strategy suggests constructing a pair of embedding maps $\varphi_1,\varphi_2: \Omega \times [{0},\beta] \to \mathcal{H}$ for a suitable Hilbert space $\mathcal{H}$. These maps would need to satisfy two conditions: first, they must reproduce the Green's function as an inner product  $ \langle \varphi_1(x_i,\tau_i),\varphi_2(x_j,\tau_j)\rangle_{\mathcal{H}} = \mathbf{G}_{ij}$; second, the images of the maps must have uniformly bounded norm. If such maps existed, a uniform bound on the determinant would follow directly from the Gram inequality. However, the discontinuity of the Green's function at $\tau_i - \tau_{j} = 0$ prevents the construction of a single, simple embedding. 

To resolve this issue, our proof closely follows the pioneering method developed by  \cite{PedraSalmhofer2008determinant}. Their key insight is to use a generalized version of Gram inequality. This generalized Gram inequality allows the determinant bound to be proven by constructing separate Hilbert space embeddings for $\tau_i -\tau_{j}\le 0$ and $\tau_i -\tau_j >0 $. While the original proof was developed for translationally invariant systems, we follow their method and generalize the results to systems that do not require translational invariance. The detailed proof of Lemma \ref{lem-det-bound} can be found in Section \ref{sec:detbound}.
{Applying the above Lemma to the matrix $\det \mathbf{G}(T,\chi,\omega,\mathbf{t},\{P_i,\tau_i\})$ from Equation~\eqref{eq:matrixGdettobound}, then yields the bound
\begin{align}\label{htaubound}
h_{\boldsymbol{\tau}}(P_1,\cdots,P_s,T,\chi)\le 2^{2Ms}\,.
\end{align}
}

\paragraph*{The Tree Summation}
Next, we consider the bound for the summation over $P_1,\cdots,P_s$. Using the uniform bound over $h_{\boldsymbol{\tau}}$, it remains to bound 
$$
\sum_{P_1,\cdots,P_s}|v_{P_1} \cdots v_{P_s}|\sum_{\chi \in  \mathcal{A}(T)}\prod_{(i,j)\in T}|g_{\tau_i,\tau_j}(P_i,P_j,\chi_{ij})|. $$
Note that, since the summation over $\chi$ can be bounded by summing over all possible pairs of creation and annihilation for each edge separately, we can bound the above summation as
$$\sum_{P_1,\cdots,P_s}|v_{P_1} \cdots v_{P_s}|  \prod_{(i,j)\in T} M_{\tau_i,\tau_j}(P_i,P_j) $$
where $M_{\tau_i,\tau_j}(P_i,P_j)$ is the total amplitude from all assignments on $(i,j)$
$$M_{\tau_i,\tau_j}(P_i,P_j) := \sum_{\chi_{ij}}|g_{\tau_i,\tau_j}(P_i,P_j,\chi_{ij}) | =  \sum_{k=1}^{m_{P_i}}\sum_{l=1}^{m_{P_j}} |g_{\tau_i-\tau_j}(P_i^-(k),P_j^+(l))|  +  |g_{\tau_j-\tau_i}(P_j^-(l),P_i^+(k)) |.$$
This summation over all assignments involves $2m_{P_i}m_{P_j}$ terms. 

We reformulate this summation in an order, beginning at the root, and proceeding along a {growing path} of the tree. Formally, without loss of generality we assume $(1,2,\cdots,s)$ is a growing path, and we can write the summation as 
$$
\sum_{P_1}|v_{P_1}| \sum_{P_2} |v_{P_2}| \,M_{\tau_2,\tau_{a(2)}}(P_2,P_{a(2)}) \cdots \sum_{P_s} |v_{P_s}| \,M_{\tau_s,\tau_{a(s)}}(P_s,P_{a(s)}),
$$
where $a(s)$ is the parent of $s$ on $T$. The initial step, which involves summing over the root, contributes a factor of $NL_V$ corresponding to the system size (see Lemma \ref{lem:tree-sum}). For all subsequent steps, summability of the potential $V$ and of the Green's function $g$ ensures that each of the $s-1$ remaining summations are uniformly bounded by a constant factor, $ML_VL_g$ (see Lemma \ref{lem:tree-sum-iterate} below). Combining these factors, this procedure yields a total bound of $\Or(N) \cdot (ML_VL_g)^{s}$ for the $s$-th order term, which is sufficient for proving the convergence of the expansion.
\begin{lemma} \label{lem:tree-sum-iterate}
Suppose the non-interacting Green's function $g$ is $L_g$-summable and the potential $V$ is $L_V$-summable. For a fixed box $P_*\in \cP$, and imaginary-time parameters $\tau,\tau^*\in [0,\beta]$, we have
$$
\sum_{P \in \cP} |v_{P}|\, M_{\tau,\tau^*}(P,P_*) \le ML_VL_g.
$$

\end{lemma}
\begin{proof}
Recall that 
$$
M_{\tau,\tau^*}(P,P_*)=  \sum_{k=1}^{m_{P}}\sum_{l=1}^{m_{P_*}} \left|g_{\tau-\tau^*}(P^-(k),P_*^+(l))\right| + \left| g_{\tau^*-\tau}(P_*^-(l),P^+(k))\right|  .
$$
We bound each term separately. For fixed $m_P=m$ and $1\le k\le m$, let $\mathcal{P}(m) = \{P \in \mathcal{P},\,m_{P} = m\}$, we have
$$
\sum_{P \in \mathcal{P}(m)}  |v_{P}|\,|g_{\tau-\tau^*} (P^-(k),P_*^+(l))| \le \sum_{x\in \Omega}|g_{\tau-\tau^*}(x,P_{*}^+(l)) | \sum_{\substack{P \in \mathcal{P}(m),P^-(k) = x}} |v_{P}|.$$
Similarly,
$$
\sum_{P \in \mathcal{P}(m)}  |v_{P}|\,| g_{\tau^*-\tau}(P_*^-(l),P^+(k) )|\le \sum_{x\in \Omega}|g_{\tau^*-\tau}(P_{*}^-(l),x) | \sum_{\substack{P \in \mathcal{P}(m),\, P^+(k) = x}} |v_{P}|.$$
Combining them together we have
\begin{align}
\begin{aligned}
 \sum_{P \in \mathcal{P}} |v_{P}|\,M_{\tau,\tau^*}(P,P_{*}) \le   &\sum_{l = 1}^{m_{P_{*}}} \sum_{x \in \Omega}\left(\sum_{m=1}^M\sum_{k=1}^{m}\sum_{\substack{P \in \mathcal{P}(m),\,P^-(k) = x}} |v_{P}| \right) |g_{\tau-\tau^*}(x,P_{*}^+(l))| \\
 & +  \sum_{l = 1}^{m_{P_{*}}} \sum_{x \in \Omega}\left(\sum_{m=1}^M\sum_{k=1}^{m}\sum_{\substack{P \in \mathcal{P}(m),\,P^+(k) = x}} |v_{P}| \right) |g_{\tau^*-\tau}(P_{*}^-(l),x)|  \\
 \le  &ML_VL_g,
 \end{aligned} \label{123}
\end{align}
where the last inequality comes from 
\begin{align*}
&\sum_{m=1}^M\sum_{k=1}^{m}\sum_{\substack{P \in \mathcal{P}(m),\,P^-(k) = x}} |v_{P}| = \sum_{P \in \mathcal{P},\,P^-\owns x}|v_P| \le \max_{y \in \Omega} \sum_{P \in \mathcal{P},\,P^-\owns y}|v_P|\le \frac{L_V}{2},\\ &\sum_{m=1}^M\sum_{k=1}^{m}\sum_{\substack{P \in \mathcal{P}(m),\,P^+(k) = x}} |v_{P}| = \sum_{P \in \mathcal{P},\,P^+\owns x}|v_P|\le \max_{y \in \Omega}\sum_{P \in \mathcal{P},\,P^+\owns y}|v_P|\le \frac{L_V}{2}
\end{align*}
for any $x \in \Omega$ and 
$$
\sum_{x \in \Omega}|g_{\tau-\tau^*}(x,P_{*}^+(l))| \le L_g,\quad \sum_{x\in \Omega} |g_{\tau^*-\tau}(P_{*}^-(l),x)| \le L_g,\,\forall  1\le l\le m_{P_{*}}\,.
$$

\end{proof}
\begin{lemma} \label{lem:tree-sum}
Suppose the non-interacting Green's function $g$ is $L_g$-summable and the potential $V$ is $L_V$-summable. Then, for any $T \in \mathcal{T}([s])$ we have
\begin{align}\label{eq:tree-sum-1}
\sum_{P_2,\cdots,P_s}|v_{P_2} \cdots v_{P_s}|\prod_{(i,j)\in T}\,M_{\tau_i,\tau_j}(P_i,P_j) \le (ML_VL_g)^s.
\end{align}
for fixed $P_1 \in \mathcal{P}$, and 
\begin{align}\label{eq:tree-sum-2} \sum_{P_1,\cdots,P_s}|v_{P_1} \cdots v_{P_s}|\prod_{(i,j)\in T}\,M_{\tau_i,\tau_j}(P_i,P_j) \le N(ML_VL_g)^s.\end{align}
\end{lemma}
\begin{proof} 
Without loss of generality, we assume $(1,2,\cdots,s)$ is a growing path. Therefore, we can write the summation as 
$$
\sum_{P_1 \in \cP}|v_{P_1}| \sum_{P_2 \in \cP} |v_{P_2}| \,M_{\tau_2,\tau_{a(2)}}(P_2,P_{a(2)}) \cdots \sum_{P_s \in \cP} |v_{P_s}| \,M_{\tau_s,\tau_{a(s)}}(P_s,P_{a(s)}),
$$
where $a(s)$ is the parent of $s$ on $T$. For each $i \ge 2$ and fixed $P_1, \cdots,P_{i-1}$, we bound  
$$
\sum_{P_i \in \cP} |v_{P_i}|\,M_{\tau_i,\tau_{a(i)}}(P_i,P_{a(i)}).
$$
by $ML_VL_g$ using Lemma \ref{lem:tree-sum-iterate}. 
Iterating for $i=2,\cdots,s$ completes the proof of \eqref{eq:tree-sum-1}.  

Summing over the root, we have
$$
\sum_{P_1 \in \cP} |v_{P_1}| \le \sum_{x \in \Omega}\sum_{P \in \mathcal{P},\, P\owns x} |v_P|,
$$
which is bounded by $NL_V$ under the $L_V$-summability condition. Combining this with \eqref{eq:tree-sum-1} we complete the proof of \eqref{eq:tree-sum-2}. 
\end{proof}

\noindent Combining the uniform bound of $h_\tau$ \eqref{htaubound} and Lemma \ref{lem:tree-sum} ends the proof of Theorem \ref{thm-convergence}.

\subsection{Efficient Algorithm for the Log-Partition Function}

\label{sec:alg-partition}

Now we explain the algorithm based on evaluating the truncated cumulant expansion. As established in Theorem \ref{thm-convergence} and explained in Section \ref{sec:convergence}, the series converges exponentially within the regime $\beta M L_V L_g \le C$ for some appropriately chosen absolute constant $C$. This exponential decay allows us to approximate the series by truncating it at a finite order $S$. For a target precision $N\epsilon$, a truncation order of $S = {\mathrm{log}(1/\epsilon)}$ is sufficient to ensure the truncation error is at most $\epsilon/2 $. The computational task thus reduces to evaluating the partial sum:
$$
 \sum_{s=1}^S \frac{{(-1)^s}}{s!} \sum_{T \in \mathcal{T}([s])}  \int_{[0,\beta]^s} \dd \boldsymbol{\tau} \sum_{P_1,\cdots,P_s } v_{P_1} \cdots v_{P_s}\sum_{\chi \in  \mathcal{A}(T)}\alpha_{T,\chi}\prod_{(i,j)\in T}g_{\tau_i,\tau_j}(P_i,P_j,\chi) h_{\boldsymbol{\tau}}(P_1,\cdots,P_s,T,\chi),
$$
where 
$$
h_{\boldsymbol{\tau}}(P_1,\cdots,P_s,T,\chi) =\sum_{\omega \in S(T)}\int_{[0,1]^{s-1}} \dd \mathbf{t}\, p_{T,\omega}(\mathbf{t}) \det \mathbf{G}(T,\chi,\omega,\mathbf{t},\{P_i,\tau_i\}_{i\in[s]}). $$

\noindent Let $\Xi_s$ denote the numerical value of the $s$-th order term in the sum above. Our approach is to estimate each $\Xi_s$ individually for $s=1, \dots, S$ using importance sampling.

\subsubsection{The Importance Sampling Framework}\label{sec:sampling}
To estimate each term $\Xi_s$, we design an importance sampling scheme. The goal is to sample multiple collections of variables $(T, \boldsymbol{\tau}, \{P_i\}_{i=1}^s, \chi, \omega, \mathbf{t})$ from a probability distribution that prioritizes the most significant contributions to the sum. In particular, we generate $L$ independent samples. Each sample $l=1, \dots, L$ consists of:
\begin{itemize}
\item Drawing a new collection of variables $(T, \boldsymbol{\tau}, \{P_i\}_{i=1}^s, \chi, \omega, \mathbf{t})_l$ from the distribution.
\item Computing a single unbiased estimator, $w_s^{(l)}$, which is a function of that specific collection.
\end{itemize}
The final estimator for $\Xi_s$ will be the average $\frac{1}{L}\sum_{l=1}^L w_s^{(l)}$.

For each iteration $l=1,\cdots,L$, we sample the collection of variables in the following steps:
\begin{enumerate}
\item \textbf{Sample the tree $T$.} A tree $T \in \mathcal{T}([s])$ is sampled uniformly at random. This can be done efficiently in $O(s)$ time by generating a random Prüfer sequence and constructing the corresponding tree \cite{prufer1918neuer,CChen2000} (see Section \ref{sec:prufer}). Recall that there are $s^{s-2}$ labeled trees. { If the maximum degree of the drawn tree is greater than $2M$, set $w_s^{(l)}$ as 0 and start with a new sampling iteration.}
\item \textbf{Sample imaginary time variables $\boldsymbol{\tau} \in [0,\beta]^s$}. The $s$ imaginary time variables $\tau_1, \dots, \tau_s$ are drawn independently and uniformly from the interval $[0, \beta]$.
\item \textbf{Sample the interaction terms $\{P_i\}_{i=1}^s$.} The interaction terms $P_1, \dots, P_s$ are drawn from a distribution weighted by the graphical model on the tree: 
\begin{align}
\mathbb{P}(P_1,\cdots,P_s) = \frac{1}{Z_s(T)}|v_{P_1} \cdots v_{P_s}|\prod_{(i,j)\in T} M_{\tau_i,\tau_j}(P_i,P_j)
\end{align}
 and we compute the normalization constant $Z_s(T)$. Since the sampling distribution is a graphical model on the tree $T$, we can use a standard belief propagation algorithm to both sample from the distribution and compute its normalization constant $Z_s(T)$ exactly and efficiently. 
 \begin{itemize}
\item  We describe the procedure of the BP algorithm in detail in Section \ref{sec:bp}. The computational cost of the BP algorithm is $O(s|\mathcal{P}|^2)$. In the case that $V$ is geometrically local, we have $|\mathcal{P}| = \mathcal{O}(N)$, leading to a quadratic dependence on $N$. 
\item We will discuss in Section \ref{sec:truncation} how to reduce the dependence on $N$ to linear in the case that the potential $V$ is geometrically local and the Green's functions $g$ are exponentially decaying by truncating the long-range interactions. 
\item We also note that the partition function $Z_s$ is uniformly bounded by $N(ML_VL_g)^{{s}}$ according to Lemma \ref{lem:tree-sum}.
 \end{itemize}
 \item \textbf{Sample the contraction assignment $\chi$.} 
 We independently select  for each edge $(i,j) \in T$ a specific fermion contraction (a pair of creation/annihilation operators) with probability distribution proportional to the magnitude of the corresponding Green's function:
$$\mathbb{P}(\chi_{ij}) = \frac{|g_{\tau_i,\tau_j}(P_i,P_j,\chi_{ij})|}{M_{\tau_i,\tau_j}(P_i,P_j)}.$$
 Then if $(\chi_{ij})_{(i,j)\in T}$ does not form a valid assignment, that is, a single operator is contracted more than once, we set the weight $w_s^{(l)} = 0$ and start with a new sampling iteration. 
\begin{rem}
We note that sampling each $\chi_{ij}$ independently does not encode the constraint that each operator can only be selected once. Setting the weight of invalid samples as 0 results in a loss in the convergence radius by an absolute constant compared to the optimal radius achievable with our framework. However, this does not affect the complexity scaling of the algorithm.

For practical implementation, one might be able to use a refined approach. If the size of the interacting terms is a constant (i.e., $m_P$ is a constant for all $P \in \mathcal{P}$), one could sample the assignments $\chi$ uniformly from $\mathcal{A}(T)$ using a sequential scheme before sampling $P_1,\cdots,P_s$. In the general case, the sampling of $\chi$ and that of $P_1, \dots, P_s$ can be implemented jointly using a graphical model. Nevertheless, for simplicity, we will confine our analysis to the approach of sampling $\chi_{ij}$ independently. 
\end{rem}

 \item \textbf{Sample the growing path $\omega = (\omega(1),\cdots,\omega(s))$.} 
  As established in Section \ref{sec:determinant-expansion}, the average over growing paths can be implemented by a simple sequential sampling procedure. We construct the path $\omega = (\omega(1), \dots, \omega(s))$ iteratively: set $\omega(1)=1$, and for each subsequent step $k$, choose $\omega(k+1)$ uniformly at random from the set of all valid vertices that can be added to the path $(\omega(1), \dots, \omega(k))$.

\item \textbf{Sample the interpolation variable $\mathbf{t}$.} The $s-1$ interpolation variables $t_1, \dots, t_{s-1}$ are drawn independently from the distribution $p(t_i) = b_it_i^{b_i-1} 1_{[0,1]}$, where $b_i$ is the number of choices available at step $i$ of the growing path construction.
\end{enumerate}
After sampling the full collection of variables $(T,\boldsymbol{\tau},P_1,\cdots,P_s,\chi,\omega,\mathbf{t})$, we compute the following weight function:
\begin{align} \label{eq:estimator} w_s^{(l)} = {\frac{(-1)^s\,s^{s-2}}{s!}}\,\alpha_{T,\chi}\, \mathrm{sgn}\left( {v_{P_1}\cdots v_{P_s}}\prod_{(i,j) \in T}  \,g_{\tau_i,\tau_j}(P_i,P_j,\chi) \right) Z_{{s}}(T)\,\beta^s \det(\mathbf{G}(T,\chi,\omega,\mathbf{t},\{P_i,\tau_i\})).\end{align}
By construction, this random weight is an unbiased estimator of the term $\Xi_s$:  
$$\Xi_s = \mathbb{E}[w_s].$$
Finally, we estimate $\Xi_s$ by the average $\frac{1}{L}\sum_{l=1}^L w_s^{(l)}$. Summing this over each order $s \in \{1, \dots, S\}$, we approximate the truncated series by 
$$
\sum_{s=1}^S \hat{\Xi}_s \approx \frac{1}{L} \sum_{i=1}^L(w_1^{(i)} + \cdots +  w_s^{(i)}).
$$

\subsubsection{The Query Complexity}
To determine the overall complexity of our Monte Carlo algorithm, we analyze two key components: 
\begin{itemize}
\item The computational cost of generating a single random estimate.
\item The number of samples required to achieve the desired statistical accuracy.
\end{itemize}

Regarding the first component, we distinguish between the query complexity of the algorithm and the internal cost of evaluating the Green's function. Recall that the query complexity is the number of of times we query entries of the non-interacting Green's function $g_{ab}(\tau)$ and the interaction coefficients $V_P$. In this and the subsequent section, we only analyze the algorithm in terms of the query complexity to complete the proof of Theorem \ref{thm-algorithm-general} and \ref{thm-algorithm-geo-local}.
The detailed algorithm for evaluating a single query, along with its stability analysis from the inexact computation in each query, is reserved for Section \ref{sec:compute_non_int_greens_func}. Ultimately, the total algorithmic complexity—formalized in Corollary \ref{cor-finite-range-ham}—is the product of the query complexity (derived in Theorem \ref{thm-algorithm-general} and \ref{thm-algorithm-geo-local}) and the evaluation cost established in Section \ref{sec:compute_non_int_greens_func}, which results in a $N^2 \polylog(\epsilon^{-1})$ overhead for the general $H_0$ and a $\polylog(N\epsilon^{-1})$ overhead for finite-ranged $H_0$.

We now analyze the query complexity of generating a single random estimate, $W = \sum_{s=1}^S w_s$. The expansion is truncated at $S = \polylog(1/\epsilon)$ orders to achieve an $N\epsilon$ target accuracy for the final computation of the log-partition function $\log N$. At each order $s$ (where $1 \le s \le S$), the query complexity of steps 1, 2, 4, 5, and 6 is at most $\mathcal{O}(s)$ and the query complexity of computing the determinant is at most $\mathcal{O}(s^3)$. Given that $S$ is logarithmic, the total contribution of these steps across all orders is a polylogarithmic factor. 

The computational bottleneck is the BP algorithm in Step 3. We consider two scenarios: In the general case, the BP algorithm requires $\mathcal{O}(|\mathcal{P}|^2 s)$ steps (see Section~\ref{sec:bp}). Summing this from $s=1$ to $S$, the total query complexity to generate one sample is $\sum_{s=1}^S \mathcal{O}(|\mathcal{P}|^2 s) = \mathcal{O}(|\mathcal{P}|^2 S^2) = \mathcal{O}(|\mathcal{P}|^2 \polylog(1/\epsilon))$. For a geometrically local potential and an exponentially decaying Green's function, the query complexity across all $S$ orders can be reduced to $\mathcal{O}(N \polylog(1/\epsilon))$ (see Section \ref{sec:truncation}).

{
We then analyze the required number of samples. Note that the weight $w_s$ is bounded by 
\begin{align*}
|w_s| & \le \frac{s^{s-2}\beta^s }{s!}|Z_s(T)||\det (\mathbf{G}(T,\chi,\omega,\mathbf{t},\{P_i,\tau_i\})) | \\
& \le N(e2^{2M}M\beta L_VL_g)^s.
\end{align*}
When $e2^{2M}M \beta L_VL_g\le1$, we have $\sum_{s=1}^S|w_s| \le N$. By the standard Hoeffding's inequality, $O(\log(1/\delta)/\epsilon^2)$ samples are sufficient to estimate $\sum_{s=1}^S {\Xi}_s$ within $N\epsilon$-error, where $1-\delta$ is the success probability. 
}

Finally, the overall query complexity is the product of the query complexity per sample and the required number of samples. We conclude that the query complexity is $\mathcal{O}(|\mathcal{P}|^2\epsilon^{-2}\mathrm{polylog}(N/\epsilon))$ in the general case, and $\mathcal{O}(N\epsilon^{-2}\polylog(1/\epsilon))$ for a geometrically local potential and an exponentially decaying Green's function. {We will also discuss how to improve this to independent of $N$ if we further assume the system is translation invariant.} Combining the analysis of the query complexity with the cost of evaluating the non-interacting Green's function, we obtain the final algorithm complexity stated in Theorem \ref{thm-algorithm-general} and Theorem \ref{thm-algorithm-geo-local}.

\subsubsection{Improved Strategy for Local Potential and Exponentially Decaying Green's Function} \label{sec:truncation}
The bottleneck of the algorithm described above is step 3, namely that of sampling $P_1,\cdots,P_s$ and computing the normalization constant $Z_s$ by BP. In this section, we explain how to simplify this step by truncating the summation for the geometrically local potential and exponentially decaying Green's function, {leading to the improved runtime claimed in Theorem \ref{thm-algorithm-geo-local}.} 

Recall that at order $s$, after sampling the tree $T$ and the imaginary-time variables $\boldsymbol{\tau}$, we need to evaluate
\begin{equation} \label{exact-sum}
\sum_{P_1,\cdots,P_s } v_{P_1} \cdots v_{P_s}\sum_{\chi \in  \mathcal{A}(T)}\alpha_{T,\chi}\prod_{(i,j)\in T}g_{\tau_i,\tau_j}(P_i,P_j,\chi_{ij}) \,h_{\boldsymbol{\tau}}(P_1,\cdots,P_s,T,\chi).
\end{equation}
Rather than summing over all possible paths $(P_1,\cdots,P_s) \in \mathcal{P}^s$, we can truncate paths with small weights to reduce the computational complexity. In particular, since the Green's function is exponentially decaying, any paths $(P_1,\cdots,P_s)$ where at least one pair $(P_i, P_j)$ connected by an edge in $T$ is ``far apart'' on the lattice $\Lambda$ will have an exponentially small $g_{\tau_i,\tau_j}(P_i, P_j,\chi_{ij})$ term. This makes the total contribution of that entire path set negligible.

We can therefore truncate these ``small weight'' paths by enforcing a maximum distance $R$ for all linked pairs. This leads to the following formal definition of the paths we keep.
\begin{defn}
We say $P_1,\cdots,P_s$ forms an $R$-\textbf{good path} if for any $(i,j) \in T$, $\mathrm{dist}(P_i,P_j) \le R$. We use {$\cP_R^s$} to denote the set of all $R$-good path with length $s$.
\end{defn}
\noindent By truncating the paths that violate this $R$-good condition, we obtain the following summation over the truncated path set $\cP_R^s$:\begin{align} \label{eq:truncate-sum} \hat{\Xi}_s \!:=\! \frac{{(-1)^s}}{s!}\!\!\sum_{T \in \mathcal{T}([s])}\int_{[0,\beta]^s} \!\!\!\!\!\!\dd \boldsymbol{\tau}\!\!\!\sum_{(P_1,\cdots,P_s) \in \cP_R^s} \!\!\!\!\!v_{P_1} \cdots v_{P_s}\sum_{\chi \in \mathcal{A}(T)}\!\!\!\alpha_{T,\chi}\prod_{(i,j)\in T}g_{\tau_i,\tau_j}(P_i,P_j,\chi) h_{\boldsymbol{\tau}}(P_1,\cdots,P_s,T,\chi).\end{align}We can then use the same importance sampling framework as described in Section \ref{sec:sampling} to evaluate this truncated summation.
In the remainder of the section, we will analyze the truncation error and the resulting query complexity. 

\paragraph*{Error Bound for the Truncation}
To bound the truncation error, we start with a Lemma to bound the summation over a single edge $(i,j)$ when $\mathrm{dist}(P_i,P_j)$ is large.  
\begin{lemma} \label{lem:truncated-sum-iterate}
Suppose that the potential $V$ is $L_V$-summable, and the Green's function $g$ is $(K,\xi)$-exponentially decaying. For a fixed $P_* \in \mathcal{P}$, and imaginary-time parameters $\tau,\tau^*$, we have
$$
\sum_{P \in \mathcal{P}|\,\mathrm{dist}(P,P_*)>R} |v_P|M_{\tau,\tau^*}(P,P_*) \le C_dML_V K{\xi R^{d-1}} \exp(-R/\xi),
$$
\end{lemma}
\begin{proof}
Similar to the proof of Lemma \ref{lem:tree-sum-iterate}, denoting  and letting $\mathcal{P}(m) = \{P \in \mathcal{P},\,m_{P} = m\}$,
\begin{align*}
 \sum_{P\in\mathcal{P}|\mathrm{dist}(P,P_*)>R} \!\!\!\!\!\!\!\!\!\!\!|v_{P}|M_{\tau,\tau^*}(P,P_*)\le   &\sum_{l = 1}^{m_{P_{*}}} \sum_{x\in\Omega|\mathrm{dist}(x,P_{*})>R}\!\!\!\left(\sum_{m=1}^M\sum_{k=1}^{m}\sum_{\substack{P \in \mathcal{P}(m),\,P^-(k) = x}} |v_{P}| \right) |g_{\tau-\tau^*}(x,P_{*}^+(l))| \notag \\
 &\!\!\!\!\!\!\!\!\!\! +  \sum_{l = 1}^{m_{P_{*}}} \sum_{x\in\Omega|\mathrm{dist}(x,P_{*})>R}\left(\sum_{m=1}^M\sum_{k=1}^{m}\sum_{\substack{P \in \mathcal{P}(m),\,P^+(k) = x}} |v_{P}| \right) |g_{\tau^*-\tau}(P_{*}^-(l),x)| \\
\le  & ML_V \sum_{r \ge R}C_d'{r}^{d-1} K\exp(-{r}/\xi)  \\
\le & M L_V  C_d K {\xi R^{d-1}}\exp(-R/\xi).
\end{align*}
\end{proof}
\begin{lemma} \label{lem:truncated-sum}(Error bound of the truncation) 
Suppose that the potential $V$ is $L_V$-summable, and the Green's function $g$ is $(K,\xi)$-exponentially decaying. Let $\mathcal{P}^s$ denote the set of all paths $(P_1,\cdots,P_s)$ with $P_i \in \mathcal{P} $. We have
\begin{align} \label{eq:truncation-error}
\sum_{(P_1,\cdots,P_s) \in \mathcal{P}^s\setminus \cP_R^s} |v_{P_1} \cdots v_{P_s} |  \prod_{(i,j)\in T} |M_{{\tau_i,\tau_j}}(P_i,P_j)|  &
&{\le s\,C_d^{{ s-1}}\,L_V^sM^{s-1}NK^{s-1}\xi^{1+d(s-2)}R^{d-1}e^{-R/\xi}}
\end{align}
for some constant $C_d$ that only depends on $d$.

Consequently, consider the truncated term in the summation
$
\hat{\Xi}_s 
$. One can choose $R  = \mathcal{O}(\log(1/\epsilon)) $ so that when $\beta 4^{M}ML_VK\xi^{d}C_d \le C $ for an appropriate absolute constant $C$, the truncation error is bounded by 
$$
\left|\sum_{s=1}^S \hat{\Xi}_s - \sum_{s=1}^S \Xi_s\right| = \Or(N\epsilon).
$$
\end{lemma}
\begin{proof}
Without loss of generality, assume $1,\cdots,s$ forms a growing path and denote by $a(s)$ the parent of $s$.
Let $\mathcal{P}_{j,R}^s$ denote the set of all paths $(P_1,\cdots,P_s)$ that violate the condition of $R$-goodness at the edge $\{a(j),j\}$, i.e., $\mathrm{dist}(P_j,P_{a(j)})> R$. Then the truncation error in \eqref{eq:truncation-error} is bounded by 
$$
\sum_{j=2}^s\sum_{(P_1,\cdots,P_s) \in \cP_{j{, R}}^s} |v_{P_1} \cdots v_{P_s} |  \prod_{(i,j)\in T} |M_{{\tau_i,\tau_j}}(P_i,P_j)|.
$$
For each $j$ we can reformulate the summation over $(P_1,\cdots,P_s) \in \cP_{j,R}^s$ as 
$$
\sum_{P_1 \in \cP}|v_{P_1}|\sum_{P_2  \in \cP}|v_{P_2}|M_{\tau_2,\tau_{a(2)}}(P_2,P_{a_2}) \cdots\!\!\!\!\!\!\!\!\!\!\!\!\!\! \sum_{\mathrm{dist}(P_{j}, P_{a(j)})> R}\!\!\!\!\!\!\!\!\!\!\!\!\!|v_{P_s}|M_{\tau_j,\tau_{a(j)}}(P_j,P_{a(j)}) \cdots \sum_{P_s \in \cP}|v_{P_s}|M_{\tau_s,\tau_{a(s)}}(P_s,P_{a(s)})
$$
The summation over $P_i$, $i \in [s]\setminus\{j\}$ follows the same bound as in Lemmas \ref{lem:tree-sum-iterate} and \ref{lem:tree-sum}:
\begin{align*}
\sum_{P_1 \in \cP}|v_{P_1}| \le N L_V, \quad
\sum_{P_i \in \cP}|v_{P_i}|M_{\tau_i,\tau_{a(i)}}(P_i,P_{a(i)}) \le ML_VL_g \le C_dML_VK\xi^d.
\end{align*}
The summation over $P_j$ follows from Lemma \ref{lem:truncated-sum-iterate}:
$$
\sum_{\mathrm{dist}(P_j,P_{a(j)})>R} |v_{P_j}|M_{\tau_j,\tau_{a(j)}}(P_j,P_{a(j)})\le ML_VC_dK {\xi R^{d-1}} \exp(-R/\xi) .$$ 
Combining the bound for the summation over $P_1,\cdots,P_s$ together, we complete the proof of \eqref{eq:truncation-error}. Therefore, combining with the bound \eqref{htaubound} and $|\mathcal{T}([s])|=s^{s-2}$ we get that
\begin{align*}
&\left|\sum_{s=1}^S\hat{\Xi}_s-\sum_{s=1}^S\Xi_s\right|\\
&\qquad \le\sum_{s=1}^S\frac{1}{s!}\, \sum_{T\in\mathcal{T}([s])}\,  \int_{[0,\beta]^s}\dd\boldsymbol{\tau}\sum_{(P_1,\cdots,P_s)\in \mathcal{P}^s\backslash \mathcal{P}_R^s}|v_{P_1}\cdots v_{P_s}|\prod_{(i,j)\in T} |M_{\tau_i,\tau_j}(P_i,P_j)| 2^{2Ms}\\
&\qquad \le \sum_{s=1}^S\frac{1}{s!}\, s^{s-1}\,  \beta^s\,C_d^{{ s-1}}\,L_V^sM^{s-1}NK^{s-1}\xi^{1+d(s-2)}R^{d-1}e^{-R/\xi}2^{2Ms}\\
&\qquad \le
 e^{-R/\xi}\xi^{1-d}\,L_V\beta N R^{d-1}4^M \sum_{s=1}^S\frac{1}{s!}\, (4^M\xi^dKL_VM\beta C_ds)^{s-1}\,\\
&\qquad \le e^{-R/\xi}e\xi^{1-d}\,L_V\beta N R^{d-1}4^M \sum_{s=1}^S(4^M\xi^dKL_VM\beta C_de)^{s-1}.
\end{align*}
When $4^M \xi^d KL_VM\beta C_d e \le 0.99$, we can choose $ R =\mathcal{O}(\log(1/\epsilon))$ for some absolute constant $C$ such that the above error is bounded by $N\epsilon$.
\end{proof}

\paragraph*{The Query Complexity after Truncation}
The truncation modifies the graphical model probability in step 3, which is handled by the BP subroutine. The new probability distribution is given by:$$\mathbb{P}(P_1,\cdots,P_s) = \frac{1}{\hat{Z}_s(T)}|v_{P_1} \cdots v_{P_s}| \prod_{(i,j) \in T}M_{\tau_i,\tau_j}(P_i,P_j) \mathbf{1}_{\mathrm{dist}(P_i,P_j) \le R},$$for some different normalization constant $\hat{Z}_s(T)$. Truncating out the long-range terms in $M_{\tau_i,\tau_j}(P_i,P_j)$ improves the query complexity of the BP subroutine: Recall that the $O(|\mathcal{P}|^2)$ cost per vertex arises from computing a message vector of size $|\mathcal{P}|$, where each entry requires a sum over $|\mathcal{P}|$ states. The truncation now sparsifies this computation. The sum for each entry is restricted to child states $P_i$ within distance $R$ of the parent's state $P_{a(i)}$. For an $(\mathsf{n},r_0)$-geometrically local state space $V$, this sum is over at most $\mathsf{n}C_d(R+r_0)^d$ states. This reduces the query complexity of computing one message vector from $\Or(|\mathcal{P}|^2)$ to $\Or(|\mathcal{P}| \cdot \mathsf{n}C_d(R+r_0)^d)$. Given $|\mathcal{P}| \le N\mathsf{n}$, we obtain the $O(N\mathsf{n}^2 C_d(R+r_0)^d)$ query complexity.

\smallskip

As established in Lemma \ref{lem:truncated-sum-iterate}, an $N\epsilon$-approximation error can be guaranteed by choosing a radius $R = \mathcal{O}(\text{polylog}(1/\epsilon))$. This result, combined with the per-vertex cost derived above, confirms that the BP step has a query complexity of $\mathcal{O}(N\polylog(1/\epsilon))$. Multiplying this per-sample query complexity by the number of required samples $1/\epsilon^2$ (as analyzed in Section \ref{sec:sampling}) and the $\mathcal{O}(\polylog(N/\epsilon)) $ cost for computing the non-interacting Green's function (as analyzed in Section \ref{sec:compute_non_int_greens_func}) gives an overall algorithm complexity of $\mathcal{O}(N\epsilon^{-2}\polylog(1/\epsilon))$.
{
\paragraph*{An Improved Algorithm in the Translation Invariant Setting}
If we further assume the system is translation-invariant—in addition to the exponential decay of the non-interacting Green's function and the locality of the potential—we can simplify the procedure described above. This simplification improves the query complexity of the BP subroutine, making it independent of the system size.

We fix the first term in the summation, $P_1$, to be some specific $P_*$. Using translation invariance, we can rewrite the truncated term in $\hat{\Xi}_s$ as
\begin{align*}  &\hat{\Xi}_s= \\
&\! \frac{{(-1)^s|\mathcal{P}|v_{P_1}}}{s!}\!\!\sum_{T \in \mathcal{T}([s])}\int_{[0,\beta]^s} \!\!\!\!\!\!\dd \boldsymbol{\tau}\!\!\!\sum_{(P_2,\cdots,P_s) \in \cP_R^{s,P_1}} \!\!\!\!\!\!\!\!\!\!\!\!v_{P_2} \cdots v_{P_s}\sum_{\chi \in \mathcal{A}(T)}\!\!\!\alpha_{T,\chi}\prod_{(i,j)\in T}g_{\tau_i,\tau_j}(P_i,P_j,\chi) h_{\boldsymbol{\tau}}(P_1,\cdots,P_s,T,\chi),\end{align*}
where $\mathcal{P}_R^{s,P_1}$ denotes all possible $(P_2,\cdots,P_s)$ such that $P_1, P_2, \cdots,P_s$ form an $R$-good path starting from $P_1$. Recall from Lemma \ref{lem:truncated-sum} that when $\beta L_V $ is bounded by some constant $C(M,K,\xi,d)$, choosing $R = \mathcal{O}(\log(1/\epsilon))$ and $S = O(\log N)$ ensures the truncated summation $\sum_{s=1}^S \hat{\Xi}_s$ approximates the log-partition function up to $N\epsilon$ accuracy. 

Now we evaluate $\sum_{s=1}^S \Xi_s$ with a modified BP procedure: we sample $P_2,\cdots,P_s$ from
\begin{align}\label{eq:bp-conditional}
\mathbb{P}(P_2,\cdots,P_s|P_1) = \frac{1}{\hat{Z}_s^{P_1}(T)}|v_{P_2}\cdots v_{P_s}| \prod_{(i,j) \in T}M_{\tau_i,\tau_j}(P_i,P_j)1_{\mathrm{dist}(P_i,P_j) \le R}.
\end{align}
and compute the normalization constant $\hat{Z}_s^{P_1}(T)$. The query complexity is analyzed as below:
\begin{itemize}
\item \textbf{Query complexity per sample.} We show that query complexity for generating a single sample can now be reduced to being independent of $N$. Indeed, any path in $\mathcal{P}_R^{s,P_1}$ is supported on a region within distance at most $sR$ from $P_1$. Denoting
$$
B(P_1, \eta ) = \{P \in \mathcal{P}: \mathrm{dist}(P,P_1) \le \eta \},
$$
we note that in the distribution $\mathbb{P}(P_2,\cdots,P_s|P_1)$ each $P_i(i\ge 2)$ is supported on
 $B(P_1,sR)$. Since $|B(P_1,sR)| = \mathcal{O}(\mathrm{polylog}(1/\epsilon))$, the query complexity of BP can be reduced to $\mathcal{O}(\mathrm{polylog}(1/\epsilon))$. Thus the query complexity of BP sampler is independent of $N$.
 \item \textbf{The number of required samples.} The unbiased estimator now becomes 
 \begin{align*}
 &\hat{w}_s^{P_1} =\\
 &{\frac{(-1)^s\,s^{s-2}}{s!}}\,\alpha_{T,\chi}\, \mathrm{sgn}\left(\! {v_{P_1}\cdots v_{P_s}}\!\!\!\prod_{(i,j) \in T} \!\!\! \,g_{\tau_i,\tau_j}(P_i,P_j,\chi) \right) \!|\mathcal{P}| Z_{{s}}^{P_1}(T)\,\beta^s \det(\mathbf{G}(T,\chi,\omega,\mathbf{t},\{P_i,\tau_i\}));
 \end{align*}
 the only difference between $w_s$ in \eqref{eq:estimator} in and $w_s^{P_1}$ is that we have replaced $Z_s$ with a factor $|\mathcal{P}|Z_s^{P_1}$. Note that $|\mathcal{P}| = \mathcal{O}(N)$ and $|Z_s^{P_1}| \le M\beta L_V L_g \le M\beta L_V C_dK \xi^d $  by Lemma \ref{lem:tree-sum-iterate}. Thus when $\beta L_V$ is bounded by some constant $C(M,K,\xi,d)$, we also have 
 $ |w_s^{P_1}| \le O(N)$ and by Hoeffding's inequality $O(\log(1/\delta))$ samples are sufficient to estimate $\sum_{s=1}^S\Xi_s$ within $N\epsilon$ error. 
\end{itemize}
Combining the analysis of the cost per sample and the number of required samples, and multiplying this with the complexity for evaluating the Green's function (analyzed in Section \ref{sec:compute_non_int_greens_func}), the algorithm complexity is $\mathrm{polylog}(N/\epsilon)$ for translation invariant systems. 
 }
\subsection{Efficient Algorithm for Computing Local Observables} \label{sec:alg-observable}
The partition function encodes many physically interesting quantities. In that spirit, the algorithm for computing the log-partition function can also be applied to compute a wide range of these quantities. In this section, we will adapt this algorithm to compute local observables of weakly interacting Fermionic systems.

    Let $O$ be an observable with $P_*$ denoting the list of creation and annihilation operators in $O$, and $m_{{P_{*}}}$ denoting the number of creation/annihilation operators in $P_*$, which also satisfy $m_{{P_*}} \le M$. The expectation of $O$ over the thermal state is $\mathrm{Tr}(Oe^{-\beta H})/Z$, which can be related to the log-partition function by the following standard result:
    \begin{lemma}
\begin{align} \label{eq:identity}
 \left.\frac{\partial}{\partial \lambda }\log \mathrm{Tr}(e^{-\beta (H + \lambda O)})\right|_{\lambda =0} = -\frac{\beta}{Z}\Tr(Oe^{-\beta H}) 
\end{align}
\end{lemma}
\begin{proof}
Let $Z(\lambda) = \mathrm{Tr}(e^{-\beta (H + \lambda O)})$. The left-hand side of the identity is $\frac{\partial}{\partial \lambda} \log Z(\lambda)|_{\lambda=0}$. Using the chain rule, this is equal to
\begin{align} \label{232}
\left.\frac{1}{Z(0)} \frac{\partial Z(\lambda)}{\partial \lambda}\right|_{\lambda=0}.\end{align} Now we compute the derivative of $Z(\lambda)$. By Duhamel's formula, we get
$$
\frac{\partial}{\partial \lambda} e^{-\beta (H + \lambda O)} = \int_0^1 e^{-s\beta(H + \lambda O)} (-\beta O) e^{-(1-s)\beta(H + \lambda O)} ds.
$$
Taking the trace of this expression and using its cyclic property:
\begin{align*}
\frac{\partial Z(\lambda)}{\partial \lambda} &= -\beta \int_0^1 \mathrm{Tr}\left( e^{-s\beta(H + \lambda O)} O e^{-(1-s)\beta(H+ \lambda O)} \right) ds \\
&= -\beta \int_0^1 \mathrm{Tr}\left( O e^{-\beta(H + \lambda O)} \right) ds \\
& = -\beta \mathrm{Tr}(O e^{-\beta(H + \lambda O)}).
\end{align*}
Evaluating this at $\lambda=0$ gives
$$
\frac{\partial Z(\lambda)}{\partial \lambda}\bigg|_{\lambda=0} = -\beta \mathrm{Tr}(O e^{-\beta H}).
$$
Substituting this to \eqref{232} we complete the proof.
\end{proof}
We apply the cumulant expansion to the $\log \mathrm{Tr}(\dots)$ term. The derivative $\frac{\partial}{\partial \lambda}$ at $\lambda=0$ filters the expansion to keep only the terms that contain exactly one factor of the observable $O$. This process gives the following series expansion for the observable:
\begin{align*}
\frac{1}{Z}\Tr(Oe^{-\beta H})& = -\frac{1}{\beta}\sum_{s=1}^\infty \frac{(-1)^s}{s!}\sum_{P_1,\cdots,P_{s} \in \mathcal{P}} \frac{\partial}{\partial v_{P_*}}v_{P_1} \cdots v_{P_{s}} \int_{[0,\beta]^s} { \dd \boldsymbol{\tau}~} \mathcal{E}_c(\{P_i,\tau_i\}_{i\in [s]}) \\ =& \sum_{s=1}^\infty \frac{(-1)^{s-1}}{\beta(s-1)!}\sum_{P_1 = P_*,\,P_2,\cdots,P_{s} \in \mathcal{P}}v_{P_2} \cdots v_{P_{s}} \int_{[0,\beta]^s} { \dd \boldsymbol{\tau}~} \mathcal{E}_c(\{P_i,\tau_i\}_{i\in [s]}).
\end{align*}
\noindent Next, we use the determinant expansion on $\mathcal{E}_c$. The $s$-th order term in the corresponding expansion will be 
\begin{align} \label{O-series}
&\Xi_s^O:=\nonumber\\
&\frac{(-1)^{s-1}}{\beta(s-1)!} \sum_{T \in \mathcal{T}([s])}  \int_{[0,\beta]^s}\!\!\!\! \dd \boldsymbol{\tau} \!\!\!\sum_{P_1 = P^*,\,P_2\cdots,P_{s}\in\mathcal{P} } \!\!\!\!v_{P_2} \cdots v_{P_{s}}\!\!\!\!\!\sum_{\chi \in  \mathcal{A}(T)}\!\!\!\alpha_{T,\chi}\!\!\!\!\prod_{(i,j)\in T}\!\!\!\!g_{\tau_i,\tau_j}(P_i,P_j,\chi_{ij}) h_{\boldsymbol{\tau}}(P_1,\cdots,P_s,T,\chi).
\end{align}
Now we consider an algorithm for computing the local observables using a truncated series summation $\sum_{s=1}^S \Xi_s^O$, where $\Xi_s^O$ is the $s$-th order term in the series expansion given by \eqref{O-series}. 

\subsubsection{The General Case}
\paragraph*{Convergence of the Expansion}
As before, the convergence of the tree-determinant expansion follows from the boundedness of $h_{\boldsymbol{\tau}}$ and the tree-summation structure. In the following Lemma, we show that truncating at $S = O(\log (1/\epsilon))$ yields an approximation within $\epsilon$ error.
\begin{lemma}
When $\beta4^MML_VL_g < C$ for some absolute constant $C$, by choosing $S = \Omega(\log(1/\epsilon))$, we have 
$$
\left|\frac{1}{Z}\mathrm{Tr}(Oe^{-\beta H}) - \sum_{s=1}^S\Xi_s^O\right| = O(\epsilon).
$$
\end{lemma}
\begin{proof}
Recall again that 
$$
M_{\tau,\tau^*}(P,P_*)=  \sum_{k=1}^{m_{P}}\sum_{l=1}^{m_{P_*}} \left|g_{\tau-\tau^*}(P^-(k),P_*^+(l))\right| + \left| g_{\tau^*-\tau}(P_*^-(l),P^+(k))\right|.
$$
Applying the first part of Lemma \ref{lem:tree-sum}, we have
$$
\sum_{P_2,\cdots,P_{s}}|v_{P_2} \cdots v_{P_{s}}|\prod_{(i,j)\in T}|M_{\tau_i,\tau_j}(P_i,P_j)| \le (ML_VL_g)^{s-1}
$$
Combining this with the bound $h_{\boldsymbol{\tau}}\le 2^{2Ms}$ derived in \eqref{htaubound} and the fact that the number of trees $T$ of size $s$ scales as $|\mathcal{T}([s])| = s^{s-2}$ (cf. Section \ref{sec:prufer}), we obtain 
$$
|\Xi_s^O| \le (\beta 2^{2M}ML_VL_g{e})^s.
$$
Thus we complete the proof.
\end{proof}
\paragraph*{The Algorithm}
The algorithm for evaluating $\sum_{s=1}^S \Xi_s^O$ follows the importance sampling strategy similar to the algorithm for computing the log-partition function. The sampling problem at the BP subroutine becomes
$$
\mathbb{P}(P_2,\cdots,P_{s}|P_1) = \frac{1}{Z_s^O}|v_{P_2} \cdots v_{P_{s}}| \prod_{(i,j) \in T} M_{\tau_i,\tau_j}(P_i,P_j),
$$
where $Z_s^O$ is the normalization constant. We note that in general the BP procedure cannot be simplified, so the query complexity for generating a single sample is $\mathcal{O}(|\mathcal{P}|
^2{s})$. Combining this with the fact that the unbaised estimator is bounded by $O(1)$ when $\beta ML_VL_g < C$ since $|Z_s^O| \le (ML_VL_g)^s$ from Lemma \ref{lem:tree-sum}, we obtain an $\mathcal{O}(|\mathcal{P} |^2\epsilon^{-2}\polylog(1/\epsilon))$ query complexity for estimating the expectation of observable $O$.

\subsubsection{Truncations for Local Potential and Exponentially Decaying Green's Function}
Similar to the algorithm for log-partition function, when the potential term is local and the Green's function is exponentially decaying, we can truncate the summation to simplify the algorithm. Define $P_R^{s,O}$ to be all $P_2,P_3,\cdots,P_{s}$ such that $P_1,P_2,\cdots,P_{s}$ forms a $R$-good path, we consider the truncated series
\begin{align}
&\hat{\Xi}_s^O :=  \frac{(-1)^{s-1}}{\beta(s-1)!} 
  \!\sum_{T \in \mathcal{T}([s])}  \!\!\int_{[0,\beta]^s} \!\! \dd \boldsymbol{\tau}\!\!\!\!\!\!\!\!\sum_{(P_2,\cdots,P_{s}) \in \mathcal{P}_R^{s,O}}\!\!\!\!\!\!\!\!\!v_{P_2}\! \cdots v_{P_{s}}\!\!\!\!\sum_{\chi \in  \mathcal{A}(T)}\!\!\alpha_{T,\chi}\!\!\prod_{(i,j)\in T}\!\!\!g_{\tau_i,\tau_j}(P_i,P_j,\chi) h_{\boldsymbol{\tau}}(P_1,\cdots\!,P_s,T,\chi).\nonumber
\end{align}
\paragraph*{The Truncation Error} 
As before, we can truncate out the paths that violate the $R$-good condition, by setting $R = \mathcal{O}(\log(1/\epsilon))$. We have the following Lemma.
\begin{lemma}
Suppose that the potential $V$ is $L_V$-summable, and the Green's function $g$ is $(K,\xi)$-exponentially decaying. 
We can choose $R = \mathcal{O}(\log(1/\epsilon))$ so that when $\beta 4^MML_VK\xi^dC_d \le C$ for some absolute constant $C$, the truncation error is bounded by
$$
\left|\sum_{s=1}^S \hat{\Xi}_s^O - \sum_{s=1}^S  {\Xi}_s^O\right| = O(\epsilon)
$$
\end{lemma}
\begin{proof}
Similar to the proof of Lemma \ref{lem:truncated-sum}, we first show that 
\begin{align} \label{1212}\sum_{(P_2,\cdots,P_{s}) \in \mathcal{P}^{s-1}\setminus \cP_R^{s,O}} |v_{P_2} \cdots v_{P_s} |  \prod_{(i,j)\in T} |M_{{\tau_i,\tau_j}}(P_i,P_j)|  \le s(C_dML_VK)^{s-1}\xi^{1+(s-2)d}R^{d-1}e^{-R/\xi}.\end{align}
The proof of \eqref{1212} follows a similar argument as in Lemma \ref{lem:truncated-sum}. Without loss of generality, assume $P_1,\cdots,P_{s}$ forms a growing path. For $ 2 \le j \le s$, we consider the contribution from paths  $P_2,\cdots,P_{s}$ such that $(P_1,P_2,\cdots,P_{s})$ violates the condition of $R$-good path at the edge $\{a(j),j\}$, i.e., $\mathrm{dist}(P_j,P_{a(j)})\ge R$:
$$\epsilon_j = \sum_{P_2  \in \cP}|v_{P_2}|M_{\tau_2,\tau_{1}}(P_2,P_{1}) \cdots\!\!\!\!\!\!\!\!\!\!\!\!\! \sum_{\mathrm{dist}(P_{j}, P_{a(j)})> R}\!\!\!\!\!\!\!\!\!|v_{P_j}|M_{\tau_j,\tau_{a(j)}}(P_j,P_{a(j)}) \cdots \!\!\!\!\sum_{P_{s} \in \cP}|v_{P_{s}}|M_{\tau_{s},\tau_{a(s)}}(P_s,P_{a(s)}).$$
The summation over $P_i$'s with $i \in \{2,\cdots,s\}\setminus \{j\}$ is bounded by $ML_VC_dK\xi^d$ using Lemma \ref{lem:tree-sum-iterate}; the summation over $P_j$ is bounded by $C_dML_VK \xi R^{d-1} e^{-R/\xi}$ using Lemma \ref{lem:truncated-sum-iterate}. Thus we have $$\epsilon_j \le (C_dML_VK)^{s-1}\xi^{1+(s-2)d}R^{d-1}e^{-R/\xi}. $$
By the union bound, we complete the proof of \eqref{1212}. Thus combining this with bound on $h_{\boldsymbol{\tau}}$ we get 
\begin{align}
&\left|\sum_{s=1}^S\Xi_s^O - \sum_{s=1}^S \hat{\Xi}_s^O \right|\\ & \le \sum_{s=1}^S \frac{1}{s!}\sum_{T \in \mathcal{T}([s])}\int_{[0,\beta]^s}\dd \boldsymbol{\tau} \sum_{(P_2,\cdots,P_s) \in \mathcal{P}^{s-1} \setminus \mathcal{P}_R^{s,O}}|v_{P_2}\cdots v_{P_s}|\prod_{(i,j)\in T }M_{\tau_i,\tau_j}(P_i,P_j,\chi_{ij})2^{2Ms} \\
& \le \sum_{s=1}^S \frac{s^{s-1}}{s!} \beta^s (C_dML_VK)^{s-1}\xi^{1+(s-2)d}R^{d-1}e^{-R/\xi} {2^{2Ms}}\\
& = {4^M}R^{d-1}e^{-R/\xi}\xi^{1-d}\beta e \sum_{s=1}^S (e\beta C_d 4^M ML_V K \xi^d)^{s-1} 
\end{align}
When $4^M\xi^d K L_VM\beta C_d e \le 0.99$, we can choose $R =\mathcal{O}(\log(1/\epsilon))$ such that the above error is bounded by $\epsilon$. 
\end{proof}
\paragraph*{The Algorithm Complexity}
We replace $\Xi_s^O$ with $\hat{\Xi}_s^O$ and run the same importance sampling procedure to evaluate $\sum_{s=1}^S \hat{\Xi}_s^O$. 
{
Note the BP procedure becomes 
$$
\mathbb{P}(P_2,\cdots,P_{s} |P_1)  = \frac{1}{\hat{Z}_s^{O}} |v_{P_2} \cdots v_{P_{s}} | \prod_{(i,j) \in T} M_{\tau_i,\tau_j}(P_i,P_j)1_{\mathrm{dist}(P_i,P_j) < R},
$$
where $\hat{Z}_s^O$ is the normalizing constant. 
This is the same as the distribution \eqref{eq:bp-conditional} in the algorithm for translation invariant systems. Thus following a similar argument we obtain an algorithm with complexity $\mathcal{O}(\epsilon^{-2}\polylog(N/\epsilon))$ for computing the expectation of local observables. 

\subsection{Approximate Computation of the Green's Function and the Stability}
\label{sec:compute_non_int_greens_func}
In practice, the non-interacting Green's function $g_\tau(a,b)$ can only be computed approximately. This raises a critical question: is our series expansion numerically stable with respect to small errors in $g_\tau$? We analyze this using the log-partition function expansion, $\sum_{s=1}^S \Xi_s$, as our primary example, noting that the stability analysis for local observables yields a similar conclusion. The following analysis demonstrates that this approximation does not introduce a prohibitive numerical instability.

Let $\sum_{s=1}^S\tilde{\Xi}_s$ be the expansion computed using an approximate Green's function $\tilde{g}_\tau$, which satisfies a precision bound $\sup |g_{\tau} - \tilde{g}_\tau | \le \epsilon_g$. Our goal is to ensure the total error on the final sum, $|\sum_{s=1}^S\tilde{\Xi}_s - \sum_{s=1}^S\Xi_s|$, remains within our target tolerance, $O(\epsilon)$.

We now apply a straightforward (though not tight) error propagation analysis. At each order $s$, $\Xi_s$ is a summation over at most $M_s \le s^{s-2}s! N^s$ terms (where the $s!$ factor comes from the determinant). Each term $T$ is of the form $T = \frac{1}{s!} \prod_{i=1}^s g_i$, where $g_i$ is an approximately computed Green's function. Since $|g_i| \le 1$, the error $\delta T$ in a single term due to the error in each $g_i$ can be bounded by (assuming $\epsilon_gs\le1$)
$$|\delta T| \le s \left( \frac{e}{s!} \right) \epsilon_g = \frac{e}{(s-1)!} \epsilon_g.$$
The total error at order $s$ is the number of terms $M_s$ times the error per term $|\delta T|$:
$$|\delta \Xi_s| \le M_s \cdot |\delta T| \le (s^{s-2}s! N^s) \times \left(\frac{e}{(s-1)!} \epsilon_g\right) = s^{s-1}e N^s  \epsilon_g.$$ 
Thus to keep the total error within $O(\epsilon)$, it suffices to choose $\epsilon_g = O\left(\frac{\epsilon}{S^{S} N^S}\right)$.

The denominator is superpolynomially large, suggesting that $\epsilon_g$ must be exponentially small. However, this seemingly infeasible precision is achievable because the computational cost to compute $g_\tau(a,b)$ depends only polylogarithmically on the required precision, $1/\epsilon_g$. As we will show in the Lemma immediately following this discussion, this cost scales as $\Or(N^2 \cdot \text{polylog}(1/\epsilon_g))$ for a general $H_0$, and improves to $\Or(\text{polylog}(1/\epsilon_g))$ for finite-ranged $H_0$. Since $\polylog(1/\epsilon_g) = \poly(S,\log N)$ the entire cost to compute $g_\tau$ to the required precision remains $\mathcal{O}(\text{polylog}(N/\epsilon))$. Thus, the cost of computing the approximate Green's function contributes only a polylogarithmic factor to the total computational cost.
\begin{lemma}
   Suppose $H_0 = \sum_{i,j\in \Omega}h_{ij}\psi_i^\dag \psi_j$  satisfies $|h_{ij}|\le 1$ for all $i,j\in\Omega$. Given $\tau \in [-\beta,\beta]$ and $a,b \in \Omega$, there is an algorithm to compute $g_{\tau}(a,b)$ in $\epsilon$-accuracy in $\Or(N^2\log(1/\epsilon) )$ runtime. In addition, if $H_0$ is $r_1$-finite-ranged, that is, $h_{ij}=0$ unless $\operatorname{dist}(i,j)<r_1$, the runtime can be improved to $\Or(\polylog(1/\epsilon) )$.
   \end{lemma}
\begin{proof}
Recall that the non-interacting Green's function is defined by the matrix function $(f_{\ge 0}(h))_{ab}$ for $\tau \ge 0$ and $(f_{< 0}(h))_{ab}$ for $\tau < 0$, where $h$ is the matrix $H_0$. We will assume $\tau \ge 0$ without loss of generality and let $f = f_{\ge 0}$; the case $\tau < 0$ is analogous.

We fix a parameter $b > \|h\|$ (such a $b = \mathcal{O}(r_1^d)$ can be bounded using the sparsity of $\|h\|$, see Lemma \ref{spectral-norm-bound} for details). Define the rescaled matrix $\tilde{h} = h/b$. The spectrum of $\tilde{h}$ is now contained within $[-1, 1]$. We are computing $(f(h))_{ab} = (f(b\tilde{h}))_{ab}$. Let $\tilde{f}(x) = f(bx)$; this is the function we actually need to approximate on the interval $[-1, 1]$. 

By Lemma \ref{poly-approx}, we can choose $n = \mathcal{O}(\log(1/\epsilon))$ such that there is a degree $n$ polynomial $\tilde{f}_n(x) = \sum_{k=0}^n c_k T_k(x)$ that satisfies $|\tilde{f}(x) - \tilde{f}_n(x)| \le \epsilon$ for all $x \in [-1, 1]$, where $T_k$ is the $k$-th Chebyshev polynomial. These $n+1$ coefficients $c_k$ can be efficiently computed by interpolating on $O(n)$ Chebyshev nodes, which costs $O(\poly(n)) = O(\text{polylog}(1/\epsilon))$ time.

Our task thus reduces to computing $(\tilde{f}_n(\tilde{h}))_{ab} = \langle a | \tilde{f}_n(\tilde{h}) | b \rangle$. We compute this using the  recurrence relation of the Chebyshev polynomials. Define a sequence of vectors $|v_k\rangle = T_k(\tilde{h}) |b\rangle$ and compute them recursively:
$$|v_0\rangle = I |b\rangle = |b\rangle,\,|v_1\rangle = \tilde{h} |b\rangle,\quad |v_k\rangle = 2\tilde{h} |v_{k-1}\rangle - |v_{k-2}\rangle \quad \text{for } k \ge 2.$$ 
We first compute all $n+1$ vectors in this sequence. The final answer is then obtained by summing their components along $\langle a |$ weighted by the coefficients $c_k$:$$(\tilde{f}_n(\tilde{h}))_{ab} = \sum_{k=0}^n c_k \langle a | T_k(\tilde{h}) | b \rangle = \sum_{k=0}^n c_k \langle a | v_k \rangle.$$

We now analyze the computation runtime. In the general case, each matrix-vector multiplication step can be computed in $N^2$ time, and thus the overall complexity is $\Or(N^2\log(1/\epsilon))$. For finite-ranged $H_0$, we leverage the sparsity of $\tilde{h}$ to improve this complexity to be independent of $N$. Any degree-$k$ polynomial of $\tilde{h}$, including $T_k(\tilde{h})$, contains at most $k r_1$ non-zero entries on each row (corresponding to sites that are within distance $kr_1$ from the site corresponding to the column). This means the vector $|v_k\rangle = T_k(\tilde{h})|b\rangle$ is sparse, with non-zero entries only for sites $i$ where $\operatorname{dist}(i,b) < k  r_1$, thus the number of such non-zero entries is bounded by $N_k \le C_d (kr_1)^d$ for some constant $C_d$. The computational cost is dominated by the $n-1$ sparse-matrix-vector multiplications $\tilde{h} |v_{k-1}\rangle$ required for the recurrence. The cost of each step $k$ is proportional to the number of non-zero entries in $|v_{k-1}\rangle$ multiplied by the number of non-zero entries per row of $\tilde{h}$, resulting in a $\mathcal{O}(N_kr_1^d) = \mathcal{O}(k^d r_1^{2d})$ cost. Thus the total time to compute all vectors is given by
$$\sum_{k=1}^n \mathcal{O}(k^d r_1^{2d}) = \mathcal{O}(r_1^{2d} n^{d+1}).$$
Since $ n = \mathcal{O}(\log(1/\epsilon))$, the total runtime for the computation is is $\mathcal{O}(\polylog(1/\epsilon))$.
\end{proof}

\section{Technical Details}\label{sec:technicaldetails}

 In this section, we collect several technical results that are used throughout the proofs of the main result.

\subsection{The Tree-Determinant Expansion for the Cumulant}\label{sec:determunantexpansion}

This section is devoted to the proof of the tree-determinant expansion \eqref{eq:determinant_expansion} of the cumulant function. The main result is stated in Theorem \ref{cumulant-representation}. The key element of its proof consists in the use of the so-called tree-graph identity \cite{brydges1978new, Brydges1984short, Battle1984ANO}, see Theorem \ref{tree-graph-thm} below.

\subsubsection{The Tree-Graph Identity} 
We first introduce some notations. We fix an integer $s$. Let $(V_{jk})_{1\le j \le k \le s}$ be a set of elements in a commutative algebra (taken later as the center of the Grassmann algebra, that is, elements with even parity, see Section \ref{sec:grassmann_variables}). For any subset $\mathcal{N} \subset \{1,\cdots,s\}$, define 
\begin{align}\label{eq:V(N)}
V(\mathcal{N}) = \sum_{i,j\in \mathcal{N},i\le j} V_{ij}. 
\end{align}
We are interested in the connected part of $\exp(-V)$.
\begin{defn}\label{def:QQc}
Let $\mathcal{A}$ be a commutative algebra, for a map $Q: 2^{[s]} \to \mathcal{A}$ {with $Q(\emptyset)=1$}, the connected part of $Q$, denoted by $Q_c$, is the map $Q_c:2^{[s]} \to \mathcal{A}$ defined recursively as $Q_c(\emptyset)=0$ and
\begin{equation}\label{eqQcitt}
Q_c(\mathcal{N}) = Q(\mathcal{N}) -\sum_{\substack{B \subsetneq \mathcal{N},\\  \min \mathcal{N} \in B}} Q_c(B)Q(\mathcal{N}\backslash B),
\end{equation}
 where $\min\mathcal{N}=\min\{j\in [s]|j\in\mathcal{N}\}$. 
Most importantly, $Q$ and $Q_c$ satisfy the defining relations
\begin{equation} \label{def-conn}
Q(\mathcal{N}) = \sum_{\substack{B \subseteq \mathcal{N},\\  \min \mathcal{N} \in B}} Q_c(B)Q(\mathcal{N}\backslash B) = \sum_{\Pi\in \mathbf{P}_\mathcal{N}}  \prod_{B \in \Pi} Q_c(B),\quad \forall\,\mathcal{N} \subset [s],
\end{equation}
{where $\mathbf{P}_\mathcal{N}$ stands for the set of partitions of $\mathcal{N}$. Note that this decomposition precisely coincides with that of the equation \eqref{eq:connectedtext} relating the moment function $\mathcal{E}$ to its connected part $\mathcal{E}_c$.}

\end{defn}

\paragraph{Short Justification of Definition \ref{def:QQc}} From \eqref{eqQcitt}, we can directly decompose $Q(\mathcal{N})$ into a sum over partitions of $\mathcal{N}$. Moreover, each partition $\Pi$ appears exactly once, namely by choosing the block $B \in \Pi$ that contains $\min \mathcal{N}$ (the term $Q_c(B)$ in the outer sum) and the partition of $\mathcal{N} \setminus B$ into the remaining blocks. Hence 
\eqref{def-conn} follows. \qed

\smallskip

\noindent {In Section \ref{sec:grassplustreegraph}, we will choose $\mathcal{A}$ as the center of a well-chosen Grassmann algebra. The tree-graph identity gives an explicit representation of the connected part of $\mathcal{N}\mapsto \exp(-V(\mathcal{N}))$. Before presenting the identity, we first introduce some notations:
\begin{itemize}
\item For $\mathcal{N} \subset [s]$ with $|\mathcal{N} | = n$, let $T(\mathcal{N})$ be the set of (undirected) trees on the vertex set $\mathcal{N}$. {We write $\ell =(i,j) \in T$ if $i\in\mathcal{N}$ and $j\in\mathcal{N}$ are connected by an edge in the tree $T\in  T(\mathcal{N})$}. We always assume $i<j$ when we write an edge as $ (i,j)$.
\item For a subset $X \subset \mathcal{N}$, we say that a pair of vertices $\ell =(j,k)$ (not necessarily an edge on $T$) crosses $X$, or $\ell \sim \partial X$, if {$k \in X,\,j\notin X$ or vice versa}. 
\item A growing path $\omega$ is a permutation of the tree's vertices that respects the parent-to-child order: fixing the minimum element in the tree as the root, a growing path requires $\omega(1) = {\min\{j\in[s]|j\in \mathcal{N}\}}$ and that for any vertex $i$, its unique parent $a(i)$ in the tree must have appeared earlier in the path, i.e.,  $a(i) \in \{\omega(1), \dots, \omega(i-1)\}$. We denote the set of growing paths as $S(T)$.

A sequence of subsets $\{\min\mathcal{N} \} = X_1 \subset X_2 \subset \cdots \subset X_{n-1} \subset \mathcal{N}$ is said to be compatible with $T$ if it forms a growing path. That is, there is a growing path $\omega \in S(T)$ such that $X_i = \{\omega(1),\cdots,\omega(i)\}$. We note that $(X_1,\cdots,X_{n-1})$ is compatible with $T$ if and only if the edges of the tree can be ordered as $\ell_1,\cdots,\ell_{n-1}$ such that $\ell_i \sim \partial X_i$.
\item Given a sequence $(X_1,\cdots,X_{n-1})$ that is compatible with a tree $T$, let $b_i$ denote the number of edges within $T$ crossing $X_i$, define $p_{T,\omega}:[0,1]^{n-1} \to\RR $ by
\begin{align*}
p_{T,\omega}(\mathbf{t}) = t_1^{b_1-1}\cdots t_{n-1}^{b_{n-1}-1},
\end{align*}
We also define $a_{\omega}:[0,1]^{n-1} \to\RR^{n\times n} $ by \begin{align} 
a_{\omega}(\mathbf{t})_{jk} =
\begin{cases}
\prod_{i=1}^{n-1} t_i((j,k)), \,j\ne k \\
1,\,j = k.
\end{cases}.
\end{align}
 Here,
 \begin{align}
 t_i((j,k)) = \begin{cases} t_i \in [0,1] & \text{if } (j,k) \text{ crosses } X_i \\ 1 & \text{otherwise} .\end{cases} \label{def-tl}
 \end{align}
 Note that $a_\omega$ does not depend on the tree structure. 
\end{itemize}
\begin{theorem}(The tree-graph identity) \label{tree-graph-thm} \cite[Theorem 3.1]{Brydges1984short} \cite[Theorem 2]{Procacci1998}
The connected part of {the map $2^{[s]}\to \mathcal{A}$, $\mathcal{N}\mapsto \exp(-V(\mathcal{N}))$} is given by
\begin{align}
\exp(-V)_c(\mathcal{N}) = \sum_{T \in \cT(\mathcal{N})} \prod_{(i,j)\in T} (-V_{ij}) \sum_{\omega \in S(T)} \int_{[0,1]^{n-1}} \dd \mathbf{t}\,  p_{T,\omega}(\mathbf{t})\exp\left(- \sum_{u,v\in\mathcal{N},u\le v} a_{\omega}(\mathbf{t})_{uv} V_{uv}\right).\label{tree-graph}
\end{align}
\end{theorem}
\begin{proof}
For any subset $X \subset \mathcal{N}$, define $L(X) = \{(i,j): i\le j,\,i,j\in X\}$. For $\ell = (i,j) \in L(X)$, let $V_\ell = V_{ij}$. Recall that $V(X) = \sum_{\ell \in L(X)} V_\ell$. Let $X_1 = \{\min\mathcal{N}\}$, $X_2, X_3, \dots, X_r$ ($r < n$) a sequence of subsets of $\mathcal{N}$ such that $X_1 \subset X_2 \subset \dots \subset X_r$ and $|X_i| = i$ ($1 \le i \le r$) and define for any $X_r \subset X \subset \mathcal{N}$,
\begin{equation}
W_X(X_1, X_2, \dots, X_r; t_1, t_2, \dots, t_r) = \sum_{\ell \in L(X)} t_1(\ell) t_2(\ell) \dots t_r(\ell) V_{\ell}
\label{def-W} 
\end{equation}
where $t_i(\ell)$ is defined in \eqref{def-tl}. Observe that
\begin{align}
W_X(X_1, \dots, X_r; t_1, \dots, t_{r-1}, 0) &= W_{X_r}(X_1, \dots, X_{r-1}; t_1, \dots, t_{r-1}) + V( X \setminus X_r) \label{W0}  \\
W_X(X_1, \dots, X_r; t_1, \dots, t_{r-1}, 1) &= W_X(X_1, \dots, X_{r-1}; t_1, \dots, t_{r-1}) \label{W1}
\end{align}
Now we re-express $\exp(-V)$ in terms of $W$ by interpolation: note that
\begin{equation}
W_{\mathcal{N}}(X_1; t_1) = t_1 V(  \mathcal{N}) + (1-t_1) [V(X_1) + V(\mathcal{N} \setminus X_1)]~;
\end{equation}
that is, $W_{\mathcal{N}}(X_1; t_1)$ interpolates between the full $V({\mathcal{N}})$ and two of its ``proper subsets'', $V(X_1)$, $V({\mathcal{N}} \setminus X_1)$. In this way,
\begin{align}
e^{-V(\mathcal{N})} &= e^{-W_\mathcal{N}(X_1; 0)} + \int_0^1 dt_1 \left[ \frac{\partial}{\partial t_1} e^{-W_\mathcal{N}(X_1; t_1)} \right] \nonumber \\
&= e^{-W_\mathcal{N}(X_1; 0)} +\sum_{\ell_1 \sim \partial X_1} (-V_{\ell_1}) \int_0^1 dt_1 \, e^{-W_\mathcal{N}(X_1; t_1)}~, \label{A1}
\end{align}
where we further used the commutativity of the terms $V_{\ell}$ in the above integral. Let us now iterate this construction by re-expressing $W_{\mathcal{N}}(X_1,t_1)$ through interpolation: for fixed $\ell_1 \sim \partial X_1$, let $X_2 = X_1 \cup \ell_1$, notice that
\begin{equation}
W_{\mathcal{N}}(X_1, X_2; t_1, t_2) = t_2 W_{\mathcal{N}}(X_1; t_1) + (1-t_2) [W_{X_2}(X_1; t_1) + V({\mathcal{N}} \setminus X_2)]~,  
\end{equation} 
so that
\begin{align}
e^{-W_{\mathcal{N}}(X_1; t_1)} &= e^{-W_{\mathcal{N}}(X_1, X_2; t_1, 0)} + \int_0^1 dt_2 \left[ \frac{\partial}{\partial t_2} e^{-W_{\mathcal{N}}(X_1, X_2; t_1, t_2)} \right] \nonumber \\
&= e^{-W_{\mathcal{N}}(X_1, X_2; t_1, 0)} + \sum_{\ell_2 \sim \partial X_2} (-V_{\ell_2})\int_0^1 dt_2 \, t_1(\ell_2) \, e^{-W_{\mathcal{N}}(X_1, X_2; t_1, t_2)}~.
\label{A2}
\end{align}
Substituting \eqref{A2} into \eqref{A1} we get:
\begin{align}
e^{-V(\mathcal{N})} &= e^{-W_{\mathcal{N}}(X_1; 0)} + \sum_{\ell_1 \sim \partial X_1} \int_0^1 dt_1 \, (- V_{\ell_1}) e^{-W_{\mathcal N}(X_1, X_2; t_1, 0)} + \nonumber \\
&+ \sum_{\ell_1 \sim \partial X_1} \sum_{\ell_2 \sim \partial X_2} \int_0^1 dt_1 \int_0^1 dt_2 \, (- V_{\ell_1}) (-V_{\ell_2}) t_1(\ell_2) e^{-W_{\mathcal{N}}(X_1, X_2; t_1, t_2)}.
\end{align}
Note that $X_2$ depends on the choice of $\ell_1$, since $X_2=X_1\cup \ell_1$. We now iterate this procedure for $n-1$ steps: at step $n_0$, 
assume that
\begin{align}
\begin{aligned}
e^{-V(\mathcal{N})} &= \sum_{r=1}^{n_0-1} \sum_{\ell_1 \sim\partial X_1} \dots \sum_{\ell_{r-1} \sim \partial X_{r-1}} \prod_{j=1}^{r-1} (-V_{\ell_j}) \int_0^1 dt_1 \dots \int_0^1 dt_{r-1} \\
&\prod_{j=1}^{r-2} t_1(\ell_{j+1}) \cdots t_j(\ell_{j+1}) e^{-W_{\mathcal{N}}(X_1, \dots, X_{r}; t_1, \dots, t_{r-1},0)}\\
&+\sum_{\ell_1 \sim\partial X_1} \dots \sum_{\ell_{n_0-1} \sim \partial X_{n_0-1}} \prod_{j=1}^{n_0-1} (-V_{\ell_j}) \int_0^1 dt_1 \dots \int_0^1 dt_{n_0-1}\\
&\prod_{j=1}^{n_0-2} t_1(\ell_{j+1}) \cdots t_j(\ell_{j+1}) e^{-W_{\mathcal{N}}(X_1, \dots, X_{n_0-1}; t_1, \dots, t_{n_0-1})}\label{eq:r0step1}
\end{aligned}
\end{align}
where $X_{n_0}=X_{n_0-1}\cup \ell_{n_0-1}$. Rewriting $X_{[1:n_0]}:=X_1,\dots,X_{n_0}$, $t_{[1:n_0]}:=t_1,\dots, t_{n_0}$, we note that in the last term in the sum over $r$ above: 
\begin{align}
W_{\mathcal{N}}(X_{[1,n_0]};t_{[1,n_0]}) =t_{n_0}W_{\mathcal{N}}(X_{[1,n_0-1]};t_{[1:n_0-1]})+(1-t_{n_0})[W_{X_{n_0}}(X_{[1:n_0-1]};t_{[1:n_0-1]})+V(\mathcal{N}\backslash X_{n_0})].
\end{align}
Thus,
\begin{align}
&e^{-W_{\mathcal{N}}(X_{[1,n_0-1]},t_{[1,n_0-1]})}=e^{-W_{\mathcal{N}}(X_{[1,n_0]},t_{[1,n_0-1]},0)}+\int_0^1 dt_{n_0}\left[\frac{\partial}{\partial t_{n_0}}e^{-W_{\mathcal{N}}(X_{[1:n_0]};t_{[1:n_0]})}\right]\nonumber\\
&~~~~~ =e^{-W_{\mathcal{N}}(X_{[1,n_0]},t_{[1,n_0-1]},0)}+\sum_{\ell_{n_0}\sim \partial X_{n_0}}(-V_{\ell_{n_0}})\int_0^1 dt_{n_0}t_{1}(\ell_{n_0})\cdots t_{r_{0-1}}(\ell_{n_0}) e^{-{W}_{\mathcal{N}}(X_{[1:n_0]};t_{[1:n_0]})}.\label{eq:r0step}
\end{align}
Substituting \eqref{eq:r0step} into \eqref{eq:r0step1}, we get
\begin{align}
\begin{aligned}
e^{-V(\mathcal{N})}& = \sum_{r=1}^{n_0} \sum_{\ell_1 \sim\partial X_1} \dots \sum_{\ell_{r-1} \sim \partial X_{r-1}} \prod_{j=1}^{r-1} (-V_{\ell_j}) \int_0^1 dt_1 \dots \int_0^1 dt_{r-1} \\
&\prod_{j=1}^{r-2} t_1(\ell_{j+1}) \cdots t_j(\ell_{j+1}) e^{-W_{\mathcal{N}}(X_1, \dots, X_{r}; t_1, \dots, t_{r-1},0)} \\
&+\sum_{\ell_1 \sim\partial X_1} \dots \sum_{\ell_{n_0} \sim \partial X_{n_0}} \prod_{j=1}^{n_0} (-V_{\ell_j}) \int_0^1 dt_1 \dots \int_0^1 dt_{n_0}\\
&\prod_{j=1}^{n_0-1} t_1(\ell_{j+1}) \cdots t_j(\ell_{j+1}) t_{n_0-1}(\ell_{n_0})e^{-W_{\mathcal{N}}(X_{[1:n_0]};t_{[1:n_0]})},
\end{aligned}
\end{align}
and hence the induction step is verified. Repeating till $n_0=n$, we arrive at
\begin{align*}
e^{-V(\mathcal{N})}&=\sum_{r=1}^{n} \sum_{\ell_1 \sim\partial X_1} \dots \sum_{\ell_{r-1} \sim \partial X_{r-1}} \prod_{j=1}^{r-1} (-V_{\ell_j}) \int_0^1 dt_1 \dots \int_0^1 dt_{r-1} \\
&\prod_{j=1}^{r-2} t_1(\ell_{j+1}) \cdots t_j(\ell_{j+1}) e^{-W_{X_r}(X_1, \dots, X_{r-1}; t_1, \dots, t_{r-1})}e^{-V(\mathcal{N}\backslash X_r)},
\end{align*}
where we used \eqref{W0}, and $X_r$ is once again constructed by $X_{r} = X_{r-1} \cup \ell_{r-1}$. 

Next, we observe that in the summation $\sum_{\ell_1 \sim\partial X_1} \dots \sum_{\ell_{r-1} \sim \partial X_{r-1} }$, the constraint $\ell_i \sim \partial X_i$ implies that $\{\ell_1,\cdots,\ell_{r-1} \}$ forms a tree ${T}$ on $\mathcal{T}(X_{r})$, and the sequence $(X_1,\cdots,X_r)$ is compatible with $T$. We can reformulate the summation as
\begin{align*}
\sum_{\ell_1 \sim\partial X_1} \dots \sum_{\ell_{r} \sim \partial X_r } = \sum_{\substack{X_{r} \subset \mathcal{N},\\ \,|X_{r}| = r,\min\mathcal{N} \in X_r}} \sum_{T \in \mathcal{T}(X_{r})} \sum_{\substack{{(}X_1,\cdots,X_{{r}}{)}\, \\ \text{compatible with}\,{T}}} 
\end{align*}
Thus we get 
\begin{align} \label{re-sum}
e^{-V({\mathcal{N}})} & = \sum_{r=1}^{n} \sum_{\substack{X_{r} \subset \mathcal{N},\\ \,|X_{r}| = r,\min\mathcal{N} \in X_r}} e^{-V(\mathcal{N} \setminus X_{r})} \sum_{T \in \mathcal{T}(X_{r})} \prod_{j=1}^{r-1}(-V_{\ell_j})\int_0^1 dt_1\cdots\int_0^1dt_{r-1} \\ &\sum_{\substack{{(}X_1,\cdots,{X_{r})} \, \\ \text{compatible with}\,{T}}} \prod_{j=1}^{r-2} t_1(\ell_{j+1}) \cdots t_j(\ell_{j+1}) e^{-W_{X_{r}}(X_1, \dots, X_{r-1}; t_1, \dots, t_{r-1})},\nonumber
\end{align}
where $\ell_j$ is the edge in $T$ that satisfies $\ell_j \sim \partial X_j$ {but not $\ell_j \sim \partial X_{j+1}$}. Let $b_j$ be the number of edges in $T$ crossing $X_j$, then {since given a tree $T\in\mathcal{T}(X_r)$ with compatible sequence $(X_1,\cdots ,X_r)$, the only edges that can cross $X_j$ are $\ell_j,\cdots ,\ell_{r-1}$, with $\ell_j$ definitely being one of them}, we have
$$
\prod_{j=1}^{r-2} t_1(\ell_{j+1}) \cdots t_j(\ell_{j+1}) = \prod_{j=1}^{r-1}t_j^{b_j-1}.
$$
Define $K(X_r) = 1$ for $|X_l| = 1$ and, for ${|}X_r{|} \ge 2$,
\begin{align*}
K(X_r) =   \sum_{T \in \mathcal{T}(X_{r})} \prod_{j=1}^{r-1}(-V_{\ell_j}){\int_{[0,1]^{r-1}}d\mathbf{t}}\sum_{\substack{{(X_1,\cdots,X_{r})}\, \\ \text{compatible with}\,{T}}}  \left(\prod_{j=1}^{r-1}t_j^{b_j-1} \right)e^{-W_{X_{r}}(X_1, \dots, X_{r-1}; t_1, \dots, t_{r-1})}.
\end{align*}
Combining this with \eqref{re-sum}, 
$$
e^{-V(\mathcal{N})} = \sum_{{r}=1}^n \sum_{\substack{X_{r} \subset \mathcal{N},\\ \,|X_{r}| = r,\min\mathcal{N} \in X_r}}K(X_{r})\exp(-V({\mathcal{N}}\setminus X_{r})).
$$
Comparing this with the definition of connected part of $\exp(-V)$ \eqref{def-conn}, we can directly identify $K$ with $\exp(-V)_c$.
\end{proof}

The following properties of the weight function $p_{T,\omega}$ in the tree-graph identity is important to design the sampling algorithm. 
\begin{prop} \label{sequence-weight}
   For a fixed tree $T$ and a growing path $\omega$ of $T$. Let $b_i$ be the number of edges crossing $X_i = \{\omega(1),\cdots,\omega(i)\}$, we have 
\begin{align} \label{part1}   \int_{[0,1]^{n-1}} \dd \mathbf{t}  p_{T,\omega}(\mathbf{t}) = \frac{1}{b_1\cdots b_{n-1}}. \end{align}
Moreover,
   \begin{align}\label{part2}\sum_{\omega \in S(T)} \int_{[0,1]^{n-1}} \dd \mathbf{t}  p_{T,\omega}(\mathbf{t}) = 1.\end{align}
\end{prop}
\begin{proof}
The proof of \eqref{part1} follows from direct calculation $$\int_0^1 t_i^{b_i-1}\dd t = \frac{1}{b_i}.$$
For the second part, we sum over all possible compatible sequences sequentially: note that once $\omega(1),\cdots,\omega(i)$ is fixed, $\omega(i+1)$ has exactly $b_i$ choices, thus
$$
\sum_{\omega \in S(T)} \frac{1}{b_1\cdots b_{n-1}} = \sum_{\omega(2)}\frac{1}{b_1} \sum_{\omega(3)}\frac{1}{b_2} \cdots \sum_{\omega(n-1)}\frac{1}{b_{n-2} } = 1,
$$
where we used $b_{n-1} = 1$.
\end{proof}
\noindent We end this section with a proposition which will be crucial for deriving the determinant bound in Section \ref{sec:detbound}.
\begin{prop} \label{gram}
Given a sequence $ X_1 \subset \cdots \subset X_{n-1} \subset \mathcal{N}$ with $|X_i| = i$ and $\mathbf{t} \in \RR^{n-1}$, there is a set of unit vectors $u_1,\cdots,u_n \in \RR^{n}$ for $i \in \mathcal{N} $ such that 
$$ a_{\omega}(\mathbf{t})_{jk} =  \langle u_j,u_k\rangle,\quad \forall\,j,k\in \mathcal{N}.    $$
\end{prop}
\begin{proof}
We recall that 
\begin{align*} 
a_{\omega}(\mathbf{t})_{jk} :=
\begin{cases}
\prod_{i=1}^{n-1} t_i((j,k)), \,j\ne k \\
1,\,j = k.
\end{cases}.
\end{align*}
Without loss of generality, we assume $\mathcal{N} = \{1,\cdots,n\}$, and $X_i = \{1,\cdots,i\}$. Then for $j < k$, $a_\omega (\mathbf{t})_{jk}  = \prod_{i=j}^{k-1} t_i$. Let ${e}_1,\cdots,{e}_n$ be a orthonormal basis of $\RR^n$, define $u_1 = e_1$ and $u_j = t_{j-1}u_{j-1} + e_j\sqrt{1-t_{j-1}^2} $ for $2\le j \le n$, that is given $t_0\equiv 0$,
\begin{align*}
u_j=\sum_{\ell=0}^{j-1}  t_{\ell+1}\cdots t_{j-1} e_{\ell+1}\sqrt{1-t_\ell^2},
\end{align*}
where by convention, empty products evaluate to $1$. They satisfy the desired relations: first by a telescopic sum argument, 
\begin{align*}
\langle u_j, u_j\rangle &=\sum_{\ell,\ell'=0}^{j-1} t_{\ell+1}\cdots t_{j-1}t_{\ell'+1}\cdots t_{j-1}\sqrt{1-t_{\ell^2}}\sqrt{1-t_{\ell'}^2}\langle e_{\ell+1},e_{\ell'+1\rangle}\\
&=\sum_{\ell=0}^{j-1}(t_{\ell+1}\cdots t_{j-1})^2(1-t_\ell^2)=1,
\end{align*}
while for $j< k$,
\begin{align*}
\langle u_j,u_k\rangle &=\sum_{\ell=0}^{j-1}\sum_{\ell'=0}^{k-1} t_{\ell+1}\cdots t_{j-1}t_{\ell'+1}\cdots t_{k-1}\sqrt{1-t_{\ell^2}}\sqrt{1-t_{\ell'}^2}\langle e_{\ell+1},e_{\ell'+1\rangle}\\
&=\sum_{\ell=0}^{j-1}(t_{\ell+1}\cdots t_{j-1})^2 t_j\cdots t_{k-1}(1-t_\ell^2)=\prod_{i=j}^{k-1}t_i.
\end{align*}
The result follows. 

\end{proof}

\subsection{Grassmann Variables} 
\label{sec:grassmann_variables}

Another key element in the derivation of the determinant expansion \eqref{eq:determinant_expansion} is the use of Grassmann variables, whose algebra and integration rules we briefly recall here. A Grassmann algebra is a free algebra over $\mathbb{C}$, generated by a set of elements $\{ \phi_a^+, \phi_a^- \}_{a=1}^N$ with the anti-commutation relations:
\begin{equation}
    \{ \phi_a^{\sigma}, \phi_b^{\tau} \} = \phi_a^{\sigma} \phi_b^{\tau} + \phi_b^{\tau} \phi_a^{\sigma} = 0,
\end{equation}
for all $a,b \in \{1, \ldots, N\}$ and $\sigma, \tau \in \{+,-\}$. A direct consequence is the nilpotency property, $(\phi_a^{\sigma})^2 = 0$. We will often write these variables collectively as vectors, $\phi^+ = (\phi_1^+, \phi_2^+, \dots, \phi_N^+)$ and $\phi^- = (\phi_1^-, \phi_2^-, \dots, \phi_N^-)$. Any function of these variables can be expanded as a finite polynomial, and the algebra they generate has a dimension of $2^{2N}$: indeed, arbitrarily ordering the generators as $\phi_1,\cdots,\phi_{2N}$,
since each generator anticommutes and squares to zero, any monomial can be reordered (up to a sign) into strictly increasing order, so the set
\[
\{\,1\,\} \cup \{\,\phi_{i_1}\phi_{i_2}\cdots\phi_{i_k} : 1 \le i_1 < i_2 < \cdots < i_k \le 2N\,\}
\]
spans the Grassmann algebra. These monomials are linearly independent, so they form a basis. Therefore, the dimension is equal to the sum $\sum_{k=0}^{2N} \binom{2N}{k} = 2^{2N}$.
Another important straightforward property is that even-degree monomials of Grassmann variables commute with all elements of the algebra.

\paragraph*{Berezin Integration}
 The Berezin integral over the Grassmann algebra is defined as follow: For a single variable $\phi$, the integration is defined by the following rules:
\begin{equation}
    \int\dd\phi \, 1 = 0, \quad \int\dd\phi \, \phi = 1,
\end{equation}
and in a multivariate setting
$$
\int \dd \phi_1 (\phi_1 \cdots \phi_n) = \phi_2 \cdots \phi_n.
$$
This definition is analogous to the expectation values in a two-level fermionic system. If we identify $\phi$ with a creation operator $\psi^\dag$ and the integration measure $\dd\phi$ with the corresponding annihilation operator $\psi$, the rules correspond to taking an expectation value with respect to a vacuum state $\ket{\mathrm{vac}}$:
\begin{equation}
    \braket{\mathrm{vac}|\psi^\dagger|\mathrm{vac}} = \int \dd\phi \, 1 = 0, \quad \braket{\mathrm{vac}|\psi \psi^\dag|\mathrm{vac}} = \int \dd\phi \, \phi = 1.
\end{equation}
For multiple variables, the differential elements also anti-commute, e.g., $\dd\phi_a^+ \dd\phi_b^- = -\dd\phi_b^- \dd\phi_a^+$. 
The standard measure for our set of variables is defined as $\dd\phi = \dd\phi_1^- \dd\phi_1^+ \cdots \dd\phi_N^- \dd\phi_N^+$.
\paragraph*{Gaussian Integrals}

The evaluation of Gaussian integrals is a cornerstone of the formalism. The principal results are summarized in the following lemma.

\begin{lemma} \label{gaussian-grassmann} \cite{Charret:1996zy}
Let $A$ be an $N \times N$ invertible complex matrix. The integral measure over all $2N$ Grassmann variables is denoted by $\dd\phi \equiv \prod_{i=1}^{N} \dd\phi_i^+\dd\phi_i^-$.

\begin{enumerate}
    \item The fundamental Gaussian integral evaluates to the determinant of the matrix:
    \begin{equation}
        \int \dd\phi \, \exp\left(-\sum_{i,j=1}^N \phi_i^- A_{ij} \phi_j^+\right) = \det(A).
    \end{equation}

    \item For an integer $r$ with $1 \le r \le N$, let $\mathbf{a} = \{a_1, \ldots, a_r\}$ and $\mathbf{b} = \{b_1, \ldots, b_r\}$ be two sets of distinct indices. The integral of a monomial against the Gaussian weight is
    \begin{equation} \label{eq:gra-moment}
        \int \dd\phi \, e^{-\sum_{i,j} \phi_i^- A_{ij} \phi_j^+} \phi^-_{a_1} \phi_{b_1}^+\cdots\phi^-_{a_r}\phi^+_{b_r} = (-1)^{\sum_{k=1}^r (a_k+b_k)} \det(A_{\hat{\mathbf{a}},\hat{\mathbf{b}}}),
    \end{equation}
    where $A_{\hat{\mathbf{a}},\hat{\mathbf{b}}}$ is the $(N-r) \times (N-r)$ submatrix of $A$ obtained by deleting the rows indexed by $\mathbf{a}$ and the columns indexed by $\mathbf{b}$.
\end{enumerate}
\end{lemma}

\subsubsection{From Tree-Graph Identity to the Tree-Determinant Expansion}\label{sec:grassplustreegraph}

In this section, we combine the tree-graph identity with Gaussian integral over Grassmann variables to derive a formula for the cumulant function $\mathcal{E}_c$.
In order to do this, we reformulate the correlation function using the integration of the Grassmann variable to compute $\mathcal{E}_c$. We fix interactions $P_1,\cdots,P_s \in \mathcal{P}$ and aim to find an expression of
$\mathcal{E}_c(\{P_i,\tau_i\}_{i\in[s]})$.
 We introduce a set of Grassmann variables $(\phi_{ik}^+,\phi_{ik}^-)_{1\le i\le s,\,1\le k \le m_{P_i}}$, denote 
 $$\boldsymbol{\phi}_{i}^+ = (\phi_{i,1}^+,\cdots,\phi_{i,m_{P_i}}^+  ),\quad \boldsymbol{\phi}_{i}^- = (\phi_{i,1}^-,\cdots,\phi_{i,m_{P_i}}^-   )$$
and define
 \begin{align*}
V_{ij}=
\begin{cases}
     &  \boldsymbol{\phi}_{i}^- G_{\tau_i-\tau_j}(P_i^-,P_j^+) (\boldsymbol{\phi}_{j}^+)^{{\intercal}} +  \boldsymbol{\phi}_{j}^-G_{\tau_j-\tau_i}(P_j^-,P_i^+) (\boldsymbol{\phi}_{i}^+)^{{\intercal}} ,\quad i < j  \\
     & \boldsymbol{\phi}_{i}^- G_{0}(P_i^-,P_i^+)  (\boldsymbol{\phi}_{i}^+)^{{\intercal}} ,\quad i = j
\end{cases}
 \end{align*}
 where $G_{\tau_i-\tau_j}(P_i^-,P_i^+) $ are defined in \eqref{green-matrix}. First, using the Grassmann Gaussian integral formula of Lemma \ref{gaussian-grassmann}, the multi-point correlation function can be reformulated as the following Grassmann integral (cf. \eqref{Wickidentity}):
\begin{align}
\mathcal{E}(\{P_i, \tau_i\}_{i \in [s]}) &=\!  (-1)^{\sum_{i\in [s]} m_{P_i}(m_{P_i}-1)/2} \det(\mathbf{G}(\{P_i, \tau_i\}_{i\in [s]})) \nonumber\\
&= \!(-1)^{\sum_{i\in [s]} m_{P_i}(m_{P_i}-1)/2 } \! \int \!D[\boldsymbol{\Phi}] \exp \left(-(\boldsymbol{\phi}_1^-,\cdots,\boldsymbol{\phi}_s^-) \mathbf{G}(\{P_i, \tau_i\}_{i \in [s]}) (\boldsymbol{\phi}_1^+,\cdots,\boldsymbol{\phi}_s^+)^\intercal\right) \nonumber\\
&= \!\int D[\boldsymbol{\Phi}] \exp \left(-\sum_{1 \le i \le j \le s} V_{ij}\right) \prod_{i=1}^s \alpha_{P_i}\nonumber\\
&=\! \int D[\boldsymbol{\Phi}] \exp(-V(\mathcal{N})) \prod_{i=1}^s \alpha_{P_i}, \label{eq:correlationfunc}
\end{align}
where we recall $V(\mathcal{N})$ was defined in \eqref{eq:V(N)}, and where $D[\boldsymbol{\Phi}]  =\prod_{i\in [s]}\prod_{1\le k\le m_{P_i}} \dd \phi_{i,k}^+ \phi_{i,k}^-$ and $\alpha_{P_i} = (-1)^{m_{P_i}(m_{P_i}-1){/2}}$. The following Lemma establishes a direct connection between the Gaussian integral representation of a map and that of its connected part (see Definition \ref{def:QQc}):
\begin{lemma}
Given a commutative algebra $\mathcal{A}$, for any map $Q: 2^{[s]} \to \mathcal{A}$ and a function $f:{[s] }\to \RR$, define $Y:2^{[s]}\to\mathcal A$ as follows
$$
Y(B)=\int D[\boldsymbol{\Phi}_B] Q(B)\prod_{i\in B} f(i)\,,
$$
where $D[\boldsymbol{\Phi}_B] =\prod_{i\in B}\prod_{1\le k\le m_i} \dd\phi_{i,k}^+\dd\phi_{i,k}^-  $ for any subset $B \subset [s]$. Then, we have that
$$
Y_c(B)=  \int D[\Phi_B]\left( Q_c(B)\prod_{i\in B} f(i)\right)\,,
$$
{where $Q_c:2^{[s]}\to \mathcal{A}$, resp.~$Y_c:2^{[s]}\to \mathcal{A}$, represents the connected part of $Q$, resp.~that of $Y$.}
\begin{proof}
We note that $D[\boldsymbol{\Phi}_{B_1}]$ commutes with $D[\boldsymbol{\Phi}_{B_2}]$ for any $B_1,B_2 \subset [s]$ since each of them has even parity. Given a set $\mathcal{N}\subseteq [s]$, we recall that $\mathbf{P}_\mathcal{N}$ denotes the set of partitions of $\mathcal{N}$. For any $\Pi \in {\mathbf{P}_{\mathcal{N}}}$, we have $D[\boldsymbol{\Phi}_{\mathcal{N}}] = \prod_{B \in \Pi}D[\boldsymbol\Phi_B]$.
Thus the proof follows by linearity from the defining relation \eqref{def-conn}:
 \begin{align*}\int D[\boldsymbol\Phi_{\mathcal{N}}]\left(Q(\mathcal{N})\prod_{i\in {\mathcal{N}}} f(i)\right) &= \sum_{\Pi \in \mathbf{P}_\mathcal{N}}\int D[\boldsymbol\Phi_{\mathcal{N}}] \prod_{B \in \Pi} \left( Q_{c}(B)\prod_{i\in B} f(i) \right) \\ &= \sum_{\Pi \in \mathbf{P}_\mathcal{N}} \prod_{B\in \Pi}\int D[\boldsymbol\Phi_{{B}}] \left(  Q_{c}(B) \prod_{i\in B} f(i)\right).  \end{align*}
\end{proof}
\end{lemma}
Combining the above Lemma with the integral representation \eqref{eq:correlationfunc}, we can express $\mathcal{E}_c$ as 
\begin{align}
\begin{aligned} \label{11}
&\mathcal{E}_c(\{P_i,\tau_i\}_{i \in [s]}) = \int D[\Phi]\,{\exp(-V)_c([s])}\prod_{i=1}^s\alpha_{P_i} \\ 
& =\prod_{i=1}^s \alpha_{P_i} \int D[\Phi]\sum_{T \in \cT([s]) }\prod_{(i,j)\in T} (-V_{ij}) \sum_{\omega \in S(T)} \int_{[0,1]^{s-1}} \dd p_{T,\omega}(\mathbf{t}) \exp\left(- \sum_{u,v\in{[s]},u\le v} a_{\omega}(\mathbf{t})_{uv} V_{uv}\right),
\end{aligned}
\end{align}
where the second equation follows from the tree-graph identity of Theorem \ref{tree-graph-thm}. Now we expand $\prod_{(i,j) \in T} V_{ij} $ and compute the integral by writing it as a summation over the assignments (cf. Equations \eqref{eq:Gtog} and \eqref{eq:gxitog}):
$$
\prod_{(i,j) \in T} (-V_{ij}) = (-1)^{s-1}\sum_{\chi\in \mathcal{A}(T)} \prod_{(i,j) \in T}g_{\tau_i,\tau_j}(P_i,P_j,\chi_{ij})\,\Phi_{\chi_{ij}}
$$
where 
\begin{align*}
  \Phi_{\chi_{ij}} = 
  \begin{cases}
      & \phi_{i,k}^- \phi_{j,l}^+,\quad \text{if} \,\
      \chi_{ij} = ({-1},k,l) \\
      & \phi_{j,l}^- \phi_{i,k}^+, \quad \text{if} \,\,\chi_{ij} =({+1},k,l).
  \end{cases}
\end{align*}
Substituting this into \eqref{11} and applying the Grassmann variable Gaussian integral \eqref{eq:gra-moment}, we finally obtain the decomposition claimed in \eqref{eq:determinant_expansion}:

\begin{theorem} \label{cumulant-representation}
The cumulant function can be written as 
$$
\mathcal{E}_c(\{P_i,\tau_i\}_{i\in[s]} ) = \sum_{T \in \mathcal{T}([s])} \sum_{\chi \in \mathcal{A}(T)} \alpha_{T,\chi}\prod_{(i,j)\in T} g_{\tau_i,\tau_j}(P_i,P_j,\chi_{ij})\, h_{\boldsymbol{\tau}}(P_1,\cdots,P_s,T,\chi),
$$
where $$h_{\boldsymbol{\tau}}(P_1,\cdots,P_s,T,\chi) =\sum_{\omega \in S(T)}\int_{[0,1]^{s-1}} \dd \mathbf{t}\, p_{T,\omega}(\mathbf{t}) \det \mathbf{G}(T,\chi,\omega,\mathbf{t},\{P_i,\tau_i\}_{i\in[s]}) $$
and the definitions of $\alpha_{T,\chi}$ and $\mathbf{G}$ coincide with those given in Section \ref{sec:determinant-expansion}, namely:
where 
\begin{itemize}
\item The matrix in the determinant $ \mathbf{G}(T,\chi,\omega,\mathbf{t},\{P_i,\tau_i\}) $ is the submatrix of the weighted Green's function matrix
    $$ \begin{bmatrix}
a(\omega,\mathbf{t})_{11}G_0(P_1^-,P_1^+) & a(\omega,\mathbf{t})_{12}G_{\tau_1-\tau_2}(P_1^-,P_2^+) & \cdots & a(\omega,\mathbf{t})_{1s}G_{\tau_1-\tau_s}(P_1^-,P_s^+) \\
a(\omega,\mathbf{t})_{21}G_{\tau_2-\tau_1}(P_2^-,P_1^+) & a(\omega,\mathbf{t})_{22}G_{0}(P_2^-,P_2^+) & \cdots &a(\omega,\mathbf{t})_{2s} G_{\tau_2-\tau_s}(P_2^-,P_s^+) \\
\vdots & \vdots & \ddots & \vdots \\
a(\omega,\mathbf{t})_{s1} G_{\tau_s-\tau_1}(P_s^-,P_1^+)& a(\omega,\mathbf{t})_{s1}G_{\tau_s-\tau_2}(P_s^-,P_2^+) & \cdots & a(\omega,\mathbf{t})_{ss}G_0(P_s^-,P_s^+)
\end{bmatrix},$$
 obtained by deleting the rows and columns corresponding to the annihilation and creation operators already contracted along the tree's edges by the assignment $\chi$.
\item { $\alpha_{T,\chi}$ is the sign defined by
$$
\alpha_{T,\chi} = (-1)^{s-1}\prod_{i=1}^s \alpha_{P_i} \prod_{(i,j) \in T}\alpha_{\chi_{ij}}, 
$$
where we recall $\alpha_{P_i} = (-1)^{m_{P_i}(m_{P_i}-1)/2}$, and for $\chi_{ij} = (\pm1,k,l)$, define $\alpha_{\chi_{ij}} = (-1)^{\sum_{r=1}^{i-1}m_{P_r} + \sum_{r=1}^{j-1}m_{P_r} +k+l }$  as the sign from the Grassmann integral formula \eqref{eq:gra-moment}.}
\end{itemize}
\end{theorem}
\noindent Theorem \ref{cumulant-representation} gives a representation of the cumulant function that plays a key role in establishing the convergence result stated in \ref{thm-convergence}. 
\subsection{The Determinant Bound}\label{sec:detbound}

In this section, we are interested in the weighted Green's function matrix $\mathbf{G}(T,\chi,\omega,\mathbf{t},\{P_i,\tau_i\})$ \eqref{eq:matrixGdettobound} derived in Theorem \ref{cumulant-representation}. Proposition \ref{gram} shows that the weight of each block can be encoded in an inner product structure, so that {the matrix $\mathbf{G}(T,\chi,\omega,\mathbf{t},\{P_i,\tau_i\})$ has coefficients of the form}
\begin{align}\label{eq:GformforM}
\mathbf{G}(x,\tau,u;y,\tau',u') := g_{\tau-\tau'}(x,y)\langle u,u' \rangle,\quad x,y \in \Omega,\,\tau,\tau' \in [0,\beta],\,u,u' \in {\mathbb{S}^{s-1}},
\end{align}
{where $\mathbb{S}^{s-1}$ denotes the unit sphere in $\mathbb{R}^s$.}  
More generally, for an index set $X$, we call a map $M: X \times X \to \CC$ an $X \times X$-matrix. For $\mathbf{x}_1,\cdots,\mathbf{x}_n,\mathbf{y}_1,\cdots,\mathbf{y}_n \in X$, we use the notation $M(\mathbf{x}_k,\mathbf{y}_l)_{1\le k,l\le n}$ to denote the $n\times n$ submatrix of $M$ formed by the rows $\mathbf{x}_1,\cdots,\mathbf{x}_n$ and columns $\mathbf{y}_1,\cdots,\mathbf{y}_n$. 

In our case of interest, the matrices \eqref{eq:GformforM} we need to consider in the determinant bound can be written as $X \times X $ matrices, with $X = \Omega \times [0,\beta]\times {\mathbb{S}^{s-1}}$. Following \cite{PedraSalmhofer2008determinant} we define the determinant bound as follows:

\begin{defn}(Determinant bound)
    We say a $X\times X$ matrix $M$ satisfies the $\gamma$-uniform determinant bound if, for any integer $n$,sequence $\mathbf{x}_1,\cdots,\mathbf{x}_n,\mathbf{y}_1,\cdots,\mathbf{y}_n \in X$, 
    $$
    |\det(   M(\mathbf{x}_k,\mathbf{y}_l))_{1\le k,l\le n})| \le \gamma^{2n}.
    $$
\end{defn}
\noindent Our goal is to show that the Green's function $\mathbf{G}$ defined in \eqref{eq:GformforM} satisfies a $\gamma$-uniform determinant bound for some absolute constant $\gamma$. 

The Green's function has a discontinuity at $\tau = \tau'$, so a split is required to prove the determinant bound. We write the matrix $\mathbf{G}$ as
$
\mathbf{G} = \mathbf{G}_{< 0} +\mathbf{G}_{\ge0}  
$
with 
$$ 
\mathbf{G}_{< 0}(x,\tau,u;y,\tau',u') = g_{\tau-\tau'}(x,y)1_{\tau < \tau'}\langle u,u' \rangle,\qquad \mathbf{G}_{\ge 0}(x,\tau,y,\tau') = g_{\tau-\tau'}(x,y)1_{\tau \ge \tau'}\langle u,u' \rangle
$$
The following lemma shows that we only need to give a determinant bound for $\mathbf{G}_{\ge 0}$ and $\mathbf{G}_{< 0}$ separately:
\begin{lemma}\label{lemma:M1+M2}
If $X\times X$ matrices $M_1$ and $M_2$ satisfy the ${\gamma_1}$-uniform determinant bound and ${\gamma_2}$-uniform determinant bound, respectively, then $M_1+M_2$ satisfies the $(\gamma_1+\gamma_2)$-uniform determinant bound.
\end{lemma}
\begin{proof}
For any integer $n$ and $\mathbf{x}_1,\cdots,\mathbf{x}_n,\mathbf{y}_1,\cdots,\mathbf{y}_n \in X$, denote $A = M_1(\mathbf{x}_k,\mathbf{y}_l)_{1\le k,l\le n}$ and $B = M_2(\mathbf{x}_k,\mathbf{y}_l)_{1\le k,l\le n}$. We use the Laplace expansion to bound the determinant of $A+B$. We denote $A_{S,T}$ for $S\subset [n]$ and $T \subset [n]$ with $|S|=|T|$ as a submatrix of $A$ formed by rows $(\mathbf{x}_k)_{k \in S}$ and $(\mathbf{y}_k)_{k \in T}$; $B_{S,T}$ is similarly defined. Since $M_1$ and $M_2$ satisfy uniform determinant bounds, for $|S| = |T| = p$, we have $|\det(A_{S,T})|\le \gamma_1^{2p}$ and $|\det(B_{S,T})|\le \gamma_2^{2p}$. By the Laplace expansion, 
\begin{align*}
    \det(A+B) = \sum_{{S,T \subset [n],\, |S| = |T|}} \varepsilon(S,T) \det A_{S,T} \det B_{S^c, T^c}
\end{align*}
where $S^c = [n] \setminus S$ and some $\varepsilon(S,T) \in \{-1, 1\}$. Let $|S| = |T| = p$, we have 
\begin{align*}
    |\det A_{S,T}| \le \gamma_1^{2p},\quad  |\det B_{S^c, T^c}| \le \gamma_2^{2(n-p)}.
\end{align*}
Thus, using $\binom{n}{p}^2 \le \binom{2n}{2p}$,
\begin{align*}
    |\det(A+B)| \le \sum_{p=0}^n \binom{n}{p}^2 \gamma_1^{2p} \gamma_2^{2(n-p)} \le (\gamma_1 + \gamma_2)^{2n}.
\end{align*}
\end{proof}
\noindent Next, we derive the determinant bound for $\mathbf{G}_{< 0}$ and $\mathbf{G}_{\ge 0}$ separately. This relies on the following generalized Gram inequality:
\begin{lemma}\label{lem:genegram} (Generalized Gram Inequality) Let $\mathcal{H}$ be a Hilbert space and $n$ be an integer. For $j,j':[n] \to \RR$ and any $v_1,\cdots,v_n, w_1,\cdots,w_n \in \mathcal{H}$, consider the $n\times n$ matrix with coefficients $$M_{kl} = \langle v_k, w_l \rangle 1_{j(k)  < j'(l)}.$$ Then, we have
$$
 |\det M | \le \prod_{k=1}^n\|v_k\|\|w_k\|.
$$
The bound also holds if we replace $1_{j(k) < j'(l)}$ with $1_{j(k) \le  j'(l)}$ in the definition of $M$.
\end{lemma}
\noindent The statement of Lemma \ref{lem:genegram} relies on an expression for the determinant in the exterior algebra $\bigwedge \mathcal{H}$. We refer the interested reader to \cite[Section 3]{PedraSalmhofer2008determinant} for a proof.
This lemma is exactly useful for providing the determinant bound for $\mathbf{G}_{< 0}$ and $\mathbf{G}_{\ge 0}$. Take $\mathbf{G}_{< 0}$ for example, if we can find embedding maps $\varphi_{< 0},\varphi'_{< 0}:X\to \mathcal{H}$ with norm uniformly bounded by $\gamma$ such that the matrix elements in $\mathbf{G}_{< 0}$ can be represented by
$$
\mathbf{G}_{< 0}(x,\tau,u; y,\tau',u')1_{\tau < \tau'} = \langle \varphi_{<0}(x,\tau,u),\varphi_{<0}(y,\tau',u')\rangle 1_{\tau < \tau'},
$$
then 
$$
 \mathbf{G}_{< 0}(x_k,\tau_k,u_k,y_l,\tau_l,u_l) = \langle  \varphi_{< 0}(x_k,\tau_k,u_k), \varphi'_{< 0}(y_l,\tau_l,u_l) \rangle 1_{\tau_k < \tau_l},
$$
and the generalized Gram lemma implies a $\gamma$-uniform determinant bound for the matrix $\mathbf{G}_{< 0}$. Therefore, it suffices to construct the bounded embedding maps $\varphi_1$ and $\varphi_2$.
\begin{lemma}\label{lem.embedding}
 There is a Hilbert space $\mathcal{H}$ and embedding maps $\varphi_{< 0},\varphi'_{< 0}, \varphi_{\ge 0},\varphi_{\ge 0}':\widetilde{X}\equiv  \Omega \times [0,\beta]\to  \mathcal{H}$ with $$ \sup_{\mathbf{x} \in \widetilde{X}}\|\varphi_{< 0}(\mathbf{x})\| \le 1,\,\sup_{\mathbf{x} \in \widetilde{X}}\|\varphi'_{< 0}(\mathbf{x})\| \le 1,\,\sup_{\mathbf{x} \in \widetilde{X}}\|\varphi_{\ge  0}(\mathbf{x})\| \le 1,\,  \sup_{\mathbf{x} \in \widetilde{X}}\|\varphi'_{\ge 0}(\mathbf{x})\| \le 1$$
  such that for any $(x,\tau),(y,\tau') \in \Omega \times [0,\beta]$
  $$
 g_{\tau-\tau'}(x,y)1_{\tau < \tau'} = \langle \varphi_{< 0}(x,\tau),\varphi'_{< 0}(y,\tau')\rangle 1_{\tau< \tau'}
  $$
  and 
  $$
 g_{\tau-\tau'}(x,y)1_{\tau \ge \tau'} = \langle \varphi_{\ge 0}(x,\tau),\varphi'_{\ge 0}(y,\tau')\rangle 1_{\tau \ge \tau'}.
  $$
  Consequently, we have embedding maps $\Omega \times [0,\beta] \times {\mathbb{S}^{s-1}}\to \mathcal{H} \otimes {\mathbb{S}^{s-1}}$,{ $(x,\tau,u)\mapsto \varphi_{<0}(x,\tau)\otimes u$, $(x,\tau,u)\mapsto \varphi'_{<0}(x,\tau)\otimes u$, $(x,\tau,u)\mapsto \varphi_{\ge 0}(x,\tau)\otimes u$, $(x,\tau,u)\mapsto \varphi'_{\ge 0}(x,\tau)\otimes u$} such that
  $$
  \mathbf{G}_{<0}(x,\tau,u;y,\tau',u') = \langle \varphi_{< 0}(x,\tau) \otimes u,\varphi'_{< 0}(y,\tau')\otimes u'\rangle 1_{\tau< \tau'} 
  $$
  and 
    $$
  \mathbf{G}_{\ge 0}(x,\tau,u;y,\tau',u') = \langle \varphi_{\ge  0}(x,\tau) \otimes u,\varphi'_{\ge  0}(y,\tau')\otimes u'\rangle 1_{\tau\ge  \tau'} .
  $$
\end{lemma}

\begin{proof}
Recall that the Green's function has a matrix representation (cf. Lemma \ref{lem:Greenfunction} below)
$$ 
g_{\tau}(a,b) = (-1_{\tau < 0} e^{-\tau h}(1+e^{\beta h})^{-1} + 1_{\tau \ge 0}e^{-\tau h}(1+e^{-\beta h})^{-1} )_{ab},\quad \forall \,a,b \in \Omega.
$$
We start from a spectral decomposition of the matrix. Given $h =
\sum_{k=1}^N \epsilon_ku_ku_k^\dag $, with vectors $u_k$ of coordinates $(u_k)_x=U_{xk}$, we can write
$$ 
g_{\tau-\tau'}(x,y) = \sum_{k=1}^N (-1_{\tau < \tau'}f_{< 0}(\epsilon_k,{\tau-\tau'})  +1_{\tau \ge \tau'}f_{\ge 0}(\epsilon_k,{\tau-\tau'}))U_{{x}k}U_{{y}k}^* ,
$$
where $f_{< 0}(E,\tau) = e^{-\tau E}(1+e^{\beta E})^{-1}$ and $f_{\ge 0}(E) = e^{-\tau E}(1+e^{-\beta E})^{-1}$.
We only need to find the representations $$\sum_{k=1}^N f_{< 0}(\epsilon_k,\tau-\tau')U_{xk}U_{yk}^* = \langle \varphi_{<0}(x,\tau),\varphi'_{< 0}(y,\tau') \rangle$$ and 
$$\sum_{k=1}^N f_{\ge 0}(\epsilon_k,\tau-\tau')U_{xk}U_{yk}^* = \langle \varphi_{\ge 0}(x,\tau),\varphi'_{\ge 0}(y,\tau') \rangle.$$
For this, we need some decomposition of the scalar function $e^{-(\tau-\tau')E}$.

The decomposition can be constructed by using the fact that for $a>0$ the Fourier inverse transform of                     $ \frac{a}{\pi (a^2+x^2)} $ is  $e^{-a|x|}$:
\begin{align*}
e^{-a|x|}=\int_{\mathbb{R}}\frac{ae^{isx}}{\pi(a^2+s^2)}\,ds.
\end{align*}
Note this formula depends on the sign of both $a$ and $x$, so we will consider the 4 cases depending on the sign of $\tau - \tau'$ and $\epsilon_k$: for $\tau < \tau'$ part, we note that if $\epsilon_k< 0$,
$$
f_{< 0}(\epsilon_k,\tau-\tau') = \frac{-\epsilon_k}{1+e^{\beta \epsilon_k}}\int_\RR \frac{e^{is(\tau'-\tau)}}{\pi(s^2+\epsilon_k^2)} \dd s ;
$$
and if $\epsilon_k \ge 0$,
$$ 
f_{< 0}(\epsilon_k,\tau-\tau') = \frac{\epsilon_k}{(1+e^{-\beta \epsilon_k})}\int_\RR \frac{e^{is(\tau-\tau' + \beta)}}{\pi(s^2+\epsilon_k^2)} \dd s.$$
where we also used that $\tau-\tau'\le \beta$ by construction. Next, define $L^2([N] \times \RR)$ as the set of functions  $f:[N]\times \RR \to \CC $ with $\sum_{k=1}^N \int_\RR |f(k,s)|^2 \dd s <+\infty$ with the inner product defined by $$\langle f,g \rangle_{L^2([N]\times \mathbb{R})}: = \sum_{k=1}^N \int_\RR f^*(k,s) g(k,s)\dd s.  $$ We construct $\varphi_{< 0}\varphi_{<0}': \widetilde{X}\to L^2([N] \times\RR)$ given by
$$\varphi_{< 0}(x,\tau)(k,s) = U_{xk}^*\left(1_{\epsilon_k< 0 } \sqrt{ \frac{-\epsilon_k}{\pi(1+e^{\beta \epsilon_k})}} \frac{e^{is\tau}}{{\epsilon_k}+is} +1_{\epsilon_k \ge 0} \sqrt{\frac{\epsilon_k}{\pi(1+e^{-\beta \epsilon_k})}} \frac{e^{-is\tau}}{{\epsilon_k}+is} \right)$$
and 
$$\varphi'_{< 0}(y,\tau')(k,s)  = U_{yk}^*\left(1_{\epsilon_k< 0 } \sqrt{ \frac{-\epsilon_k}{\pi(1+e^{\beta \epsilon_k})}} \frac{e^{is\tau'}}{{\epsilon_k}+is} + 1_{\epsilon_k \ge 0} \sqrt{\frac{\epsilon_k}{\pi(1+e^{-\beta \epsilon_k})}} \frac{e^{-is(\tau'-\beta)}}{{\epsilon_k}+is}\right). $$
They satisfy
$$ 
\langle \varphi_{< 0}(x,\tau),\varphi'_{< 0}(y,{\tau'}) \rangle_{L^2([N] \times \mathbb{R})} = \sum_{k=1}^N f_{< 0}(\epsilon_k,\tau-\tau')U_{xk}U_{yk}^* = g_{\tau-\tau'}(x,y).
$$
for $\tau<\tau'$. This completes the construction for the $\mathbf{G}_{< 0}$ part. The construction for $\mathbf{G}_{\ge 0}$ follows similarly by:
$$
\varphi_{\ge 0}(x,\tau)(k,s) = -U_{xk}^*\left(1_{\epsilon_k< 0 } \sqrt{ \frac{-\epsilon_k}{\pi(1+e^{\beta \epsilon_k})}} \frac{e^{-is(\beta-\tau)}}{{\epsilon_k}+is} +1_{\epsilon_k \ge 0} \sqrt{\frac{\epsilon_k}{\pi(1+e^{-\beta \epsilon_k})}} \frac{e^{-is\tau}}{{\epsilon_k}+is} \right)
$$
and 
$$
\varphi'_{\ge 0}(y,\tau')(k,s) = U_{yk}^*\left(1_{\epsilon_k< 0 } \sqrt{ \frac{-\epsilon_k}{\pi(1+e^{\beta \epsilon_k})}} \frac{e^{is\tau'}}{{\epsilon_k}+is} +1_{\epsilon_k \ge 0} \sqrt{\frac{\epsilon_k}{\pi(1+e^{-\beta \epsilon_k})}} \frac{e^{-is\tau'}}{{\epsilon_k}+is} \right).
$$
Since $\sum_{k=1}^N |U_{xk}|^2 = 1$ and 
$$ 
\int_\RR \frac{E}{\pi(E^2+s^2)} \dd s = 1,
$$
we directly get that the embedding functions are bounded by 1. 
\end{proof}

{\begin{proof}[Proof of Lemma \ref{lem-det-bound}]
Lemma \ref{lem-det-bound} directly follows the above lemmas: first by the embedding maps introduced in Lemma \ref{lem.embedding} and the generalized Gram inequality \ref{lem:genegram}, we get that both $\mathbf{G}_{<0}$ and $\mathbf{G}_{\ge 0}$ satisfy the $1$-uniformly determinant bound. Then, by \ref{lemma:M1+M2}, we get that $\mathbf{G}=\mathbf{G}_{<0}+\mathbf{G}_{\ge 0}$ satisfies the $2$-uniform determinant bound. 
\end{proof}}

 \subsection{Dyson Expansions}\label{sec:technlemm}

In this section, we provide a proof of the Dyson expansions 
\eqref{eq:Dyson} and \eqref{eq:cumulant_expansion}.

\begin{lemma}[Dyson expansion]\label{lem:dyson}
In the notations of Sec. \ref{sec.2},
\begin{align*}
\frac{Z}{Z_0} = \sum_{s=0}^\infty \frac{(-1)^s}{s!}\sum_{P_1,\cdots,P_s \in \mathcal{P}} v_{P_1}\cdots v_{P_s} \int_{[0,\beta]^s} \dd \tau_1 \cdots \dd \tau_s \langle \cT (\Psi_{P_1}(\tau_1) \cdots \Psi_{P_s}(\tau_s) )\rangle_0.
\end{align*}
\end{lemma}

\begin{proof}
Define $W(\tau)\equiv e^{\tau H_0}\,e^{-\tau H}$, $\tau\in[0,\beta]$.
Clearly, $W(0)=\mathbf{1}$ and $e^{-\beta H} = e^{-\beta H_0}\,W(\beta)$. Therefore $Z = \operatorname{Tr}\!\left( e^{-\beta H} 
\right)= \operatorname{Tr}\!\left(e^{-\beta H_0}\,W(\beta)\right)=Z_0\langle W(\beta)\rangle_0$.
Differentiating $W(\tau)$:
\begin{align*}
\frac{d}{d\tau}W(\tau)= H_0\,e^{\tau H_0}e^{-\tau H} - e^{\tau H_0}He^{-\tau H} = -\,e^{\tau H_0}(H-H_0)e^{-\tau H} = -\,e^{\tau H_0} V e^{-\tau H}.
\end{align*}
Inserting $1=e^{-\tau H_0}e^{\tau H_0}$ on the right yields
\begin{equation}
\frac{d}{d\tau}W(\tau) 
= - \underbrace{e^{\tau H_0}V e^{-\tau H_0}}_{V(\tau)}\,W(\tau),
\qquad W(0)=1.
\label{eq:diffeq}
\end{equation}
Equation \eqref{eq:diffeq} is a first-order operator equation with a non-commuting coefficient $V(\tau)$. Therefore,
\begin{align*}
W(\beta) 
&= 1 
- \int_0^\beta d\tau \, e^{\tau H_0} V e^{-\tau H}=1-\int_0^\beta d\tau  e^{\tau H_0}V e^{-\tau H_0}W(\tau) .
\end{align*}
Iterating this, we get the series expansion
\begin{align}
W(\beta)& =\sum_{s=0}^\infty (-1)^s \int_0^\beta \int_0^{\tau_1}\dots \int_0^{\tau_{s-1}}\,V(\tau_1)\dots V(\tau_s)d\tau_1\dots d\tau_s\nonumber \\
&=\sum_{s=0}^\infty \frac{(-1)^s}{s!}\,\left[\prod_{j=1}^s\int_0^\beta d\tau_j\right]\, \mathcal{T} \left[\prod_{j=1}^s V(\tau_j)\right]\,d\tau_1\dots d\tau_s\label{eq:ZZ0exp}
\end{align}
Next, by decomposition $V=\sum_{P\in \mathcal{P}}v_P\,\Psi_P$ and developing the products above, we get 
\begin{align*}
W(\beta)=\sum_{s=0}^\infty \frac{(-1)^s}{s!}\,\sum_{P_1,\dots, P_s\in\mathcal{P}}v_{P_1}\dots v_{P_s}\int_{[0,\beta]^s} \Psi_{P_1}(\tau_1)\dots \Psi_{P_s}(\tau_s)\,d\tau_1\dots d\tau_s\,.
\end{align*}
The result follows after taking the trace against the unperturbed Gibbs state.
\end{proof}

\begin{lemma}[Dyson expansion for log-partition function]\label{lem:dysonlog}
In the notations of Sec. \ref{sec.2},
\begin{align*}
\log\left(\frac{Z}{Z_0}\right)= \sum_{s=1}^\infty \frac{(-1)^{s}}{s!} \sum_{P_1,\cdots,P_s \in \mathcal{P}} v_{P_1}\cdots v_{P_s} \int_{[0,\beta]^s} \dd \tau_1\cdots \dd \tau_s  \,\mathcal{E}_{c}(\{P_i,\tau_i\}_{i\in[s]}),
 \end{align*}
where the cumulant function $\mathcal{E}_c$ is defined as
\begin{align}\label{eq:defcumulant}
\mathcal{E}_{c}(\{P_i,\tau_i\}_{i\in[s]}):=\sum_{\Pi\in\mathbf{P}_s}(-1)^{|\Pi|-1}(|\Pi|-1)!\prod_{B=\{j_1,\cdots , j_{|B|}\}\in\Pi}\langle \mathcal{T}(\Psi_{P_{j_1}}(\tau_{j_1})\cdots \Psi_{P_{j_{|B|}}}(\tau_{j_{|B|}}) )\rangle_0,
\end{align}
where $\mathbf{P}_s$ stands for the set of partitions of $\{1,\cdots ,s\}$ and $|\Pi|$ denotes the number of sets in a partition $\Pi$. Moreover, $\mathcal{E}_c$ is the connected part of the moment function $\mathcal{E}$ and satisfies:
\begin{align}\label{eq:connected}
\mathcal{E}(\{P_i,\tau_i\}_{i\in[s]}) = \sum_{\Pi\in \mathbf{P}_s} \prod_{B\in\Pi} \mathcal{E}_c(\{P_i,\tau_i\}_{i\in B}).
\end{align}
\end{lemma}

\begin{proof}
Given an integrable source term $J$ on $[0,\beta]$, we introduce the time-ordered exponential
\begin{align*}
\mathcal{G}[J]:=\sum_{s=0}^\infty \frac{1}{s!}\, \int_{[0,\beta]^s}\, J(\tau_1)\,\cdots J(\tau_s)\, \langle \mathcal{T}(V(\tau_1)\cdots V(\tau_s))\rangle_0\,d\tau_1\cdots d\tau_s 
\end{align*}
where $V$ is defined as in Eq.\eqref{eq:diffeq}. We see directly from here that setting $J$ to the constant function $J=-1$ yields back the ratio function $Z/Z_0$ whereas $\log \mathcal{G}[J]$ yields the quantity we are after. We also directly have that $\mathcal{G}[0]=1$. 
Taking the Fréchet derivatives of $\log\mathcal{G}[J]$ at $0$ directly yields to the cumulant functions $\mathcal{E}_c$: first, by Taylor series expansion, we have that
\begin{align*}
\log\left(\frac{Z}{Z_0}\right)=\log\mathcal{G}[-1]&=\sum_{s=1}^\infty \frac{(-1)^s}{s!}\, D^s\log\mathcal{G}[0](-1,\cdots ,-1)\\
&=\sum_{s=1}^\infty \frac{(-1)^s}{s!}\int_{[0,\beta]^s}\left.\frac{\delta^s \log\mathcal{G}[J]}{\delta J(\tau_1)\cdots \delta J(\tau_s)}\right|_{J=0}\,d\tau_1\cdots d\tau_s
\end{align*}
where the second line follows from Riesz representation theorem, with
\begin{align*}
\left.\frac{\delta^s \log\mathcal{G}[J]}{\delta J(\tau_1)\cdots \delta J(\tau_s)}\right|_{J=0}:=D^m\log \mathcal{G}[0](\delta_{\tau_1},\cdots,\delta_{\tau_s}).
\end{align*}
Next, by the Fàa-di-Bruno identity, and since $\mathcal{G}[0]=1$, we get that
\begin{align*}
\left.\frac{\delta^s \log\mathcal{G}[J]}{\delta J(\tau_1)\cdots \delta J(\tau_s)}\right|_{J=0}=\sum_{\Pi\in\mathbf{P}_s} (-1)^{|\Pi|-1}(|\Pi|-1)!\,\prod_{B\in\Pi}\left.\frac{\delta^{|B|}\mathcal{G}}{\prod_{i\in B}\delta J(\tau_i)}\right|_{J=0}.
\end{align*}
 Moreover, recalling the expansion found in Eq. \eqref{eq:ZZ0exp}, the Fréchet derivatives of $\mathcal{G}$ at $0$ can be expressed by standard interchanges of series, integrals and differentiations as
\begin{align*}
\left.\frac{\delta^s \mathcal{G}[J]}{\delta J(\tau_1)\cdots \delta J(\tau_s)}\right|_{J=0}=\langle \mathcal{T}(V(\tau_1)\cdots V(\tau_s))\rangle_0\,.
\end{align*}
This directly yields Eq. \eqref{eq:defcumulant} after expressing each $V$ as a sum over interaction terms. In order to get Eq. \eqref{eq:connected}, we make one more use of the Faà-di-Bruno identity which yields the following expansion of derivatives of $\mathcal{G}[J]=e^{\log\mathcal{G}[J]}$:
\begin{align*}
\frac{\delta^s\mathcal{G}}{\delta J(\tau_1)\cdots \delta J(\tau_s)}[J]=\mathcal{G}[J]\,\sum_{\Pi\in\mathbf{P}_s}\,\prod_{B\in \Pi}\,\frac{\delta^{|B|}\log\mathcal{G}}{\prod_{i\in B}\delta J(\tau_i)}[J].
\end{align*}
This directly gives \eqref{eq:connected} after evaluating the equation at $J=0$.
\end{proof}

\subsection{Green's Function and Wick's Theorem}

This section is devoted to a proof of Equation \eqref{Wickidentity}. First, in the next Lemma, we recall the derivation of Green's functions associated to a Gaussian fermionic state:

\begin{lemma}[Green's function]\label{lem:Greenfunction}

The Green's functions defined in Definition \ref{def:Green} admit a matrix representation 
$$ g_\tau := {-}\mathbf{1}_{\tau {<} 0}e^{-\tau h}(1+e^{\beta h})^{-1} {+} \mathbf{1}_{\tau{\ge}0} e^{-\tau h}(1+e^{-\beta h})^{-1} . $$

\end{lemma}
\begin{proof}

We denote the imaginary-time Heisenberg fields
\[
\psi_a(\tau)=e^{\tau H_0}\psi_a e^{-\tau H_0},\qquad
\psi_a^\dagger(\tau)=e^{\tau H_0}\psi_a^\dagger e^{-\tau H_0}.
\]
Since
\begin{align*}
[H_0,\psi_a]&=\sum_{b,c}h_{bc}[\psi_b^\dagger\psi_c,\psi_a]\\
&=\sum_{b,c}h_{bc}\left(\psi_b^\dagger [\psi_c,\psi_a]+[\psi_b^\dagger,\psi_a]\psi_c\right)\\
&=\sum_{b,c}h_{bc}\left(-\psi_b^\dagger\psi_a\psi_c-\psi_a\psi_b^\dagger\psi_c\right)\\
&=-\sum_{b,c}h_{bc}\{\psi_b^\dagger,\psi_a\}\psi_c\\
&=-\sum_{c}h_{ac}\psi_c\,,
\end{align*}
we get that $\partial_\tau \psi_a(\tau)=-\sum_c h_{ac}\psi_c(\tau)$. In vector form (with $\bm{\psi}=(\psi_a)_a$), 
\begin{align}\label{eq:psitau}
\bm{\psi}(\tau)=e^{-\tau h}\,\bm\psi(0).
\end{align}
Moreover,
\begin{align}\label{eq:psitaudag}
\bm\psi^\dagger(\tau)=e^{\tau H_0}\bm\psi^\dagger(0) e^{-\tau H_0}=\left(e^{-\tau H_0}\bm\psi(0) e^{\tau H_0}\right)^\dagger=\left(e^{\tau h}\bm\psi(0)\right)^\dagger=\bm\psi^\dagger(0)\,e^{+\tau h}.
\end{align}
Next, we diagonalize $h=U\varepsilon U^\dagger$ and set $\bm c=U^\dagger \bm\psi$. By construction, the operators $c_{\alpha}$ also satisfy the canonical anticommutation relations. Therefore, $H_0=\sum_{\alpha,\beta}\epsilon_{\alpha\beta}c_\alpha^\dagger c_\beta= \sum_\alpha \epsilon_{\alpha} c_\alpha^\dagger c_\alpha$, where we set $\epsilon_\alpha=\epsilon_{\alpha\alpha}$, and by commutativity and the projection property of $c^\dagger_{\alpha}c_\alpha$,
\[
e^{-\beta H_0}=e^{-\beta \sum_\alpha \epsilon_\alpha c_\alpha^\dagger c_\alpha}=\prod_\alpha e^{-\beta \epsilon_\alpha c_\alpha^\dagger c_\alpha}=\prod_\alpha\left(1+c_\alpha^\dagger c_\alpha \big(e^{-\beta \epsilon_\alpha}-1\big)\right).
\]
Hence,
\begin{align*}
\langle \psi_a\psi_b^\dagger\rangle_0&=\Tr\big( e^{-\beta H_0}\psi_a\psi_b^\dagger\big)/Z_0\\
&=\sum_{\alpha,\beta}\frac{U_{a\alpha}U^\dagger_{\beta b}}{Z_0}\,\Tr\left(\prod_{\alpha'}\Big(1+c_{\alpha'}^\dagger c_{\alpha'}\big(e^{-\beta\epsilon_{\alpha'}}-1\big)\Big)c_\alpha c_\beta^\dagger\right)\\
&\overset{(1)}{=}\sum_\alpha \frac{U_{a\alpha}U_{\alpha b}^\dagger}{Z_0} \Tr\left(\prod_{\alpha'\ne \alpha}\Big(1+c_{\alpha'}^\dagger c_{\alpha'}\big(e^{-\beta \epsilon_{\alpha'}}-1\big)\Big)\,c_\alpha c_\alpha^\dagger\right)\\
&=\sum_\alpha \frac{U_{a\alpha}U_{\alpha b}^\dagger}{Z_0} \Tr\left(\prod_{\alpha'\ne \alpha}\Big(1+c_{\alpha'}^\dagger c_{\alpha'}\big(e^{-\beta \epsilon_{\alpha'}}-1\big)\Big)(1-c_\alpha^\dagger c_\alpha)\right)
\end{align*}
where in (1) we used that the trace vanishes for $\alpha\ne\beta$ due to the anticommutation relations. Developing the trace via the Fock basis \cite{bravyi2002fermionic,bravyi2004lagrangian}, we get
\begin{align*}
\langle \psi_a\psi_b^\dagger\rangle_0=\sum_\alpha \frac{U_{a\alpha}{U_{\alpha b}^\dagger }}{Z_0}\,\Tr\left(\prod_{\alpha'\ne \alpha}\,\Big(1+c_{\alpha'}^\dagger c_{\alpha'}(e^{-\beta \epsilon_{\alpha'}}-1)\Big) \right)\,.
\end{align*}
Since the trace as well as the partition function $Z_0$ factorize, the above expression simplifies as
\begin{align}\label{eq:twopointcorrs}
\langle \psi_a\psi_b^\dagger\rangle_0=\sum_\alpha \frac{U_{a\alpha}U_{\alpha b}^\dagger}{\Tr_\alpha\Big(1+c_\alpha^\dagger c_\alpha(e^{-\beta\epsilon_\alpha}-1)\Big)}= \sum_\alpha U_{a\alpha}\frac{1}{1+e^{-\beta \epsilon_\alpha}}U_{\alpha b}^\dagger=(1+e^{-\beta h})^{-1}_{ab}
\end{align}
where $\Tr_\alpha$ stands for trace over a single-mode Fock basis associated to mode $\alpha$. Moreover, 
\begin{align}\label{eq:twopointcorrs2}
\langle \psi_b^\dagger \psi_a\rangle_0=\delta_{ab}-\langle \psi_a\psi_b^\dagger\rangle_0=(1+e^{\beta h})^{-1}_{ab}\,.
\end{align}
Combining with the previously derived simple form for the time evolved creation and annihilation operators, we get
\begin{align*}
&\langle \psi_b^\dagger \psi_a(\tau)\rangle_0=\sum_\alpha (e^{-\tau h})_{a\alpha}\langle\psi_b^\dagger \psi_\alpha\rangle_0=\left(e^{-\tau h}(1+e^{\beta h})^{-1}\right)_{ab}\\
&\langle \psi_a(\tau)\psi_b^\dagger\rangle_0=\sum_\alpha (e^{-\tau h})_{a\alpha}\langle \psi_\alpha\psi_b^\dagger\rangle_0=(e^{-\tau h}(1+e^{-\beta h})^{-1})_{ab}\,.
\end{align*}

\end{proof}

\begin{lemma}[Wick's theorem \cite{sims2016decay,araki1971quasifree,bach1994generalized}]\label{lemm:Wick}

In the notations of Sec. \ref{sec.2}, 
\[
\mathcal{E}(\{P_i,\tau_i\}_{i \in [s]}) = (-1)^{\sum_{i=1}^s m_{P_i}(m_{P_i} - 1)/2}\operatorname{det}(\mathbf{G}(\{P_i,\tau_i\}_{i \in [s]}),
\]
\end{lemma}

 \begin{proof}

We argue by multilinearity of the determinant. Here we denote $m_k\equiv m_{P_k}$. By definition, we have, given { $\tau_{\sigma(1)}\ge \dots\ge \tau_{\sigma(s)}$},
\begin{align*}
&\mathcal{E}(\{P_i,\tau_i\}_{i \in [s]})\\
&\quad =\Big\langle \!\prod_{j=1}^{m_{\sigma(1)}}\psi^\dagger_{P^+_{\sigma(1)}(j)}(\tau_{\sigma(1)})\prod_{j=1}^{m_{{\sigma(1)}}}\psi_{P^-_{\sigma(1)}(j)}(\tau_{\sigma(1)})\cdots \prod_{j=1}^{m_{\sigma(s)}}\psi^\dagger_{P^+_{\sigma(s)}(j)}(\tau_{\sigma(s)})\prod_{j=1}^{m_{{\sigma(s)}}}\psi_{P^-_{\sigma(s)}(j)}(\tau_{\sigma(s)})\Big\rangle_0
\end{align*}
By \eqref{eq:psitau}, \eqref{eq:psitaudag}, we have that, for any $k\in[s]$ and $j\in [m_{\sigma(k)}]$,
\begin{align*}
&C_{\sum_{r=0}^{k-1}2m_{\sigma(r)}+j}:=\psi^\dagger_{P^+_{\sigma(k)}(j)}(\tau_{\sigma(k)})=\sum_{\alpha}\psi^\dagger_\alpha\,(e^{\tau_{\sigma(k)}h})_{\alpha P^+_{\sigma(k)}(j)}\\
&C_{\sum_{r=0}^{k-1}2m_{\sigma(r)}+m_{\sigma(k)}+j}:=\psi_{P^-_{\sigma(k)}(j)}(\tau_{\sigma(k)})=\sum_{\alpha} (e^{-\tau_{\sigma(k)}h})_{P^-_{\sigma(k)}(j)\alpha}\psi_{\alpha},
\end{align*}
and hence, since the average in $\mathcal{E}(\{P_i,\tau_i\}_{i \in [s]})$ is taken over a quasi-free Gaussian state, we have that \cite{sims2016decay,araki1971quasifree,bach1994generalized}, since $m=\sum_{k\in [s]}m_{k}$,
\begin{align*}
\mathcal{E}(\{P_i,\tau_i\}_{i \in [s]})=\langle C_1\cdots C_{2m}\rangle_0=\operatorname{pf}(\mathcal{C}), 
\end{align*}
where $\operatorname{pf}(\mathcal{C})$ stands for the pfaffian of the skew-symmetric matrix 
\begin{align*}
\mathcal{C}=(\langle C_uC_v\rangle_0)_{1\le u,v\le 2m}.
\end{align*}
This expression can be further simplified due to the block structure of the matrix $\mathcal{C}$. Indeed, since $\langle \psi_a\psi_b\rangle_0=\langle \psi_a^\dagger\psi_b^\dagger\rangle_0=0$ for all $a,b\in\Omega$,
\begin{align*}
\mathcal{C}=\begin{bmatrix}
0 &\mathcal{C}^{+-}_{11}  & 0 &\mathcal{C}^{+-}_{12}&0 &\cdots &0&\mathcal{C}^{+-}_{1s}  \\
-\mathcal{C}^{+-}_{11} &0  & \mathcal{C}^{-+}_{12}& 0&\mathcal{C}_{13}^{-+}&\cdots &\mathcal{C}^{-+}_{1{ s}} &0\\
0&- \mathcal{C}^{-+}_{12} &  0&\mathcal{C}_{22}^{+-}&0&\cdots &0&\mathcal{C}_{2s}^{+-}      \\
\vdots & \vdots & \ddots&\ddots &\ddots& \ddots&\ddots&\vdots \\
0 & -\mathcal{C}^{-+}_{1{ s}} & \cdots & \cdots&\cdots&\cdots &0& \mathcal{C}_{ss}^{+-}\\
-\mathcal{C}^{+-}_{1s}&0&\cdots&\cdots &\cdots&\cdots&-\mathcal{C}_{ss}^{+-}&0
\end{bmatrix}
\end{align*}
where $\mathcal{C}^{+-}_{ij}$ and $\mathcal{C}^{-+}_{ij}$ are blocks of size $m_{\sigma(i)}\times m_{\sigma(j)}$. Next, we consider the permutation $\Pi$ that first groups all $+$ blocks together, and next all $-$ blocks:
\begin{align*}
(1^+,1^-,2^+,2^-,\cdots, s^+,s^- )\overset{\Pi}{\longrightarrow }(1^+,2^+,\cdots ,s^+,1^-,2^-,\cdots s^-).
\end{align*}
Under this permutation, we have 
\begin{align}\label{eq:rotatedC}
\Pi C \Pi^{\intercal} =\begin{bmatrix}
0&X\\
-X^\intercal &0
\end{bmatrix}
\end{align}
where
\begin{align*}
X=\begin{bmatrix}
\mathcal{C}^{+-}_{11}   &\mathcal{C}^{+-}_{12} &\cdots &\mathcal{C}^{+-}_{1s-1}&\mathcal{C}^{+-}_{1s}\\
-\mathcal{C}^{-+}_{12}&\mathcal{C}^{+-}_{22}&\cdots &\mathcal{C}^{+-}_{2s-1}&\mathcal{C}^{+-}_{2s}\\
-\mathcal{C}^{-+}_{13}&-\mathcal{C}^{-+}_{23}&\mathcal{C}^{+-}_{33}&\cdots& \mathcal{C}^{+-}_{3s}\\
\vdots &\ddots&\ddots &\ddots &\vdots \\
- \mathcal{C}^{-+}_{1s} &\cdots&\cdots &- \mathcal{C}^{-+}_{s-1s}& \mathcal{C}^{+-}_{ss}
\end{bmatrix}.
\end{align*}
Hence, using the standard result that for any $2m\times 2m$ real matrices $A,B$ with $A^\intercal=-A$, $\operatorname{pf}(B^\intercal AB)=\operatorname{det}(B)\operatorname{pf}(A)$, we get
\begin{align}\label{eq:decomCPiX}
\operatorname{pf}(\mathcal{C})=\operatorname{pf}(\Pi^\intercal \Pi \mathcal{C}\Pi^\intercal \Pi)=\operatorname{det}(\Pi)\,\operatorname{pf}(\Pi \mathcal{C}\Pi^\intercal )
\end{align}
Another standard identity on Pfaffians of matrices of the form of \eqref{eq:rotatedC} directly implies that 
\begin{align}\label{eq:pftodet}
\operatorname{pf}(\Pi\mathcal{C}\Pi^\intercal )=\operatorname{det}(X)(-1)^{\frac{1}{2}m(m-1)},
\end{align}
where $m=\sum_k m_{P_k}$.

Finally, since the determinant of a permutation coincides with its sign, to compute $\operatorname{det}(\Pi)$ we simply need to decompose it into 
transpositions, as follows
\smallskip
\[
(1^+\tikzmarknode{R}{,}\,1^-,
  \tikzmarknode{L}{2}^+,\,2^-,\,
  \cdots,\, s^+,\,s^-)
\qquad \rightarrow \qquad
(1^+,\,2^+\tikzmarknode{U}{,}\,1^-,\,2^-,\,\tikzmarknode{V}{3}^+,\,3^- ,\,\cdots ,\,s^+,\, s^-)
\qquad \rightarrow \qquad \cdots
\]
\UnderArrow[<-]{R}{L}{}
\UnderArrow[<-]{U}{V}{}

\noindent The first basis transformation above involves $m_{\sigma(1)}\times m_{\sigma(2)}$ transpositions, the second involves $(m_{\sigma(1)}+m_{\sigma(2)})\times m_{\sigma(3)}$ etc. In total, the permutation $\Pi$ can be written as a composition of 
\begin{align*}
\sum_{1< j\le  s} (m_{\sigma(1)}+\cdots +m_{\sigma(j-1)})m_{\sigma(j)}=\sum_{1\le i<j\le s}m_{\sigma(i)}m_{\sigma(j)}
&=\frac{1}{2}\left(\Big(\sum_{i=1}^sm_{\sigma(i)}\Big)^2-\sum_{i=1}^s m_{\sigma(i)}^2\right)\\
&=\sum_{1\le i<j\le s}m_{i}m_{j}.
\end{align*}
Note that this is independent of the permutation $\sigma$ associated to the time ordering $\mathcal{T}$. Thus, we get 
\begin{align}\label{eq:detperm}
\operatorname{det}(\Pi)=(-1)^{\sum_{1\le i<j\le s}m_{i}m_{j}}
\end{align}
Combining \eqref{eq:decomCPiX}, \eqref{eq:pftodet} and \eqref{eq:detperm}, we have shown that
\begin{align*}
\operatorname{pf}(\mathcal{C})=(-1)^{\sum_{1\le i<j\le s}m_{i}m_{j}+\frac{1}{2}m(m-1)}\operatorname{det}(X) =  (-1)^{\sum_{i=1}^s m_{i}(m_{i} - 1)/2} \det X.
\end{align*}
Now, by definition, for $i\le j$, $X_{ij}=\mathcal{C}^{+-}_{ij}$ coincides with $g_{\boldsymbol{\tau}}(P^-_{\sigma(j)},P^+_{\sigma(i)})=\mathbf{G}(\{P_i,\tau_i\}_{i \in [s]})_{\sigma(j)\sigma(i)}$, whereas for $i>j$, $X_{ij}=-\mathcal{C}_{ji}^{-+}$, which again coincides with $\mathbf{G}(\{P_i,\tau_i\}_{i \in [s]})_{\sigma(j)\sigma(i)}$. Hence
\begin{align*}
\operatorname{det}(X)=\operatorname{det}(X^\intercal)=\operatorname{det}(\Sigma^\intercal \mathbf{G}(\{P_i,\tau_i\}_{i \in [s]})\Sigma)=\operatorname{det}(\mathbf{G}(\{P_i,\tau_i\}_{i \in [s]})),
\end{align*}
where we denoted by $\Sigma$ the matrix representation of the permutation $\sigma$ of the indices $i\in[s]$. The proof follows. 

\end{proof}

\subsection{Decay of Green's Function} \label{green-decay}

The goal of this section is to prove Lemma \ref{lemma-green-decay}. Recall that the non-interacting Hamiltonian is $H_0 = \sum_{i,j \in \Omega}h_{ij} \psi_i^\dag \psi_j$ and the Green's function matrix is given by $g_\tau := - \mathbf{1}_{\tau < 0}e^{-\tau h}(1+e^{\beta  h})^{-1} + \mathbf{1}_{\tau \ge 0} e^{-\tau h}(1+e^{-\beta h})^{-1}.$ Here, we suppose that $|h_{ij}|\le 1$ for all $i,j\in\Omega$, and $h_{ij}=0$ unless $\operatorname{dist}(i,j)<r_1$. We aim to prove that the matrix entries $g_{\tau}(a,b)$ are bounded by a function exponentially decaying with $\mathrm{dist}(a,b)$ uniformly over $\tau \in [-\beta,\beta]$.

For this, we define two scalar functions $f_{<0},f_{\ge0}:[-\|h\|,\|h\|] \to \RR$ by $$f_{< 0}(x) = e^{-\tau x}(1+e^{\beta x})^{-1} ,\qquad f_{\ge 0}(x) = e^{-\tau x}(1+e^{-\beta x})^{-1},$$
where we omit the dependence on $\beta,\tau$. The Green's function matrix $g_\tau$ can be written as $g_\tau = - \mathbf{1}_{\tau < 0 }f_{<0}(h) + \mathbf{1}_{\tau \ge 0} f_{\ge 0}(h)$ and we can consider the $\tau < 0$ and $\tau\ge 0$ parts separately. The proof idea is to approximate the functions $f_{<0},f_{\ge 0}$ by polynomials.

\paragraph*{The local Hamiltonian Case}

We use the truncated Chebyshev series to provide a uniform polynomial approximation for the functions $f_1$ and $f_2$. The following can be found in \cite[Theorem 8.2]{doi:10.1137/1.9781611975949}:
\begin{lemma}[Chebyshev Polynomial Approximation]\label{lem:cheb} For $\rho > 1$, define $E_\rho$ to be the closed area enclosed by the ellipse
$
 \left\{ \frac{1}{2}\left(z + z^{-1}\right) : |z| \le \rho \right\}.
$
Let $f$ be an analytical function on $[-1,1]$. If $f$ possesses an analytical continuation on $E_\rho$ and  $|f(z)| \le C$ for $z \in E_\rho$, then the error of Chebyshev polynomial approximation of $f$ is bounded by 

$$
\sup_{x\in[-1,1]}|f_n(x) - f(x) | \le \frac{2C\rho^{-n}}{{\rho-1}},
$$
where $f_n = \sum_{k=0}^n a_k P_k$ is the projection of $f$ onto $\mathrm{span}(P_0,P_1,\cdots,P_k) $  with $P_k$ denoting the $k$-th order Chebyshev polynomial.
\end{lemma}

Now we apply this general theorem to the functions $f_1$ and $f_2$.
\begin{lemma} \label{poly-approx}
There exist degree $n$ polynomials $f_{n,<0}$ and $f_{n,\ge 0}$ such that 
$$ \sup_{x\in[-\|h\|,\|h\|]} |f_{<0}(x) -f_{n,<0}(x) | \le \frac{2\rho^{-n}}{\rho - 1},\qquad   \sup_{x\in[-\|h\|,\|h\|]} |f_{\ge 0}(x) - f_{n,\ge0}(x)|  \le \frac{2\rho^{-n}}{\rho - 1}$$
where $\rho = \frac{\pi}{2\beta \|h\|} + \sqrt{1+ \frac{\pi^2}{4\beta^2 \|h\|^2}} $.
\end{lemma}
\begin{proof}
Consider the rescaling $\tilde{f}_{<0}(x) = f_{<0}(x\|h\|) $ and $\tilde{f}_{\ge 0}(x) = f_{\ge 0}(x\|h\|)$. Then $\tilde{f}_{<0}$ and $\tilde{f}_{\ge 0}$ are defined on $[-1,1]$. Moreover, given $\rho = \frac{\pi}{2\beta \|h\|} + \sqrt{1+ \frac{\pi^2}{4\beta^2 \|h\|^2}}$, for any $z =a+bi\in \CC$, $a,b\in\mathbb{R}$, in the closed area enclosed by $E_\rho$, we have $|\mathrm{Im}(\beta \|h\| z)| \le \frac{\pi}{2}$. Indeed, denoting the polar decomposition $\xi=re^{i\theta}$ with $0\le r\le \rho$ and $\theta \in\mathbb{R}$, and $z=1/2(\xi+\xi^2)$, we have that $|b|\le 1/2(\rho+\rho^{-1})$. Therefore
\begin{align*}
\beta\|h\||b|\le \beta\|h\|\frac{\rho-\rho^{-1}}{2}.
\end{align*}
Putting $x:=\pi/(2\beta\|h\|)$, we have that $\rho=x+\sqrt{1+x^2}$, and since $\rho^{-1}=\sqrt{1+x^2}-x$, we have
\begin{align*}
\rho-\rho^{-1}=2x=\frac{\pi}{\beta\|h\|}
\end{align*}
and the claim follows. This implies
$ |\tilde{f}_{<0}(z)| \le 1 $ and $|\tilde{f}_{\ge 0}(z)| \le 1$. Indeed, for any  $z = a+bi$ with $\beta \| h\| |b| \le \frac{\pi}{2} $,
$$ 
|\tilde{f}_{<0}(z)| =   \frac{e^{-\tau \|h\| a}}{\sqrt{(1+e^{\beta\|h\|a}\cos(\beta\|h\|b))^2 + e^{2\beta\|h\|a} \sin^2 (\beta \|h\|b)}} \le \frac{e^{-\tau \|h\|a}}{\sqrt{1+\,e^{2\beta \|h\|a}}} \le 1,
$$
where we used $\cos (\beta \|h\|b) \ge 0$ and $\tau\in[-\beta,\beta]$. Similar calculation shows $|\tilde{f}_{\ge 0}|$ is bounded by 1. The result follows from a direct application of Lemma \ref{lem:cheb}.
\end{proof}
\begin{lemma} \label{spectral-norm-bound}
    Suppose $H_0 = \sum_{i,j\in \Omega}h_{ij}\psi_i^\dag \psi_j$ is such that $|h_{ij}|\le 1$ for all $i,j\in\Omega$, and $h_{ij}=0$ unless $\operatorname{dist}(i,j)<r_1$. Then, the spectral norm of the matrix $h$ is bounded by
    $$ \|h\| \le C_dr_1^d.$$
\end{lemma}
\begin{proof}
Since each row (or column) has at most $C_dr_1^d$ non-zero entries and each entry is bounded by 1, we have 
$$\|h\|_1 = \|h\|_{\infty}\le C_dr_1^d$$
where $\|h\|_1$, resp. $\|h\|_\infty$ correspond to the max column sum and the max row sum, respectively. By the standard bound $\|h\| \le \sqrt{\|h\|_1\|h\|_\infty} \le C_dr_1^d$.
\end{proof}

As a direct corollary, we obtain the decay of the Green's functions of the non-interacting Hamiltonian $H_0$ that were stated in Lemma \ref{lemma-green-decay}, which we repeat below for convenience:

\begin{corollary}[First part of Lemma \ref{lemma-green-decay}]\label{lemma-green-decay1}
If the non-interacting part satisfies $|h_{ab}| \le 1$ and $h_{ab} = 0$ unless $\mathrm{dist}(a,b) < r_1$, the Green's function is $(K,\xi)$-exponentially decaying with 
$$
 K = O(\beta r_1^d)\qquad \text{ and }\qquad  \xi =O(\beta r_1^{d+1}).
$$

\end{corollary}
\begin{proof}
For fixed $a,b \in \Omega$, let $n = \lfloor\frac{\mathrm{dist}(a,b)}{r_1}\rfloor$.  Supposing the fact that $|h_{ij}|\le 1$ for all $i,j\in \Omega$ local, for any degree $n$ polynomials $f_{n,<0}$ and $f_{n,\ge0}$, we have
$$ (f_{n,<0}(h))_{ab} = (f_{n,\ge 0}(h))_{ab} = 0. $$
Choosing $f_{n,<0}$ and $f_{n,\ge0}$ to be the Chebyshev polynomial approximations of $f_{<0}$ and $f_{\ge 0}$ that satisfy the error bound in Lemma \ref{poly-approx}, we have
\begin{align*}
|g_\tau(a,b)| &  = |\mathbf{1}_{\tau \le 0}(f_{<0}(h) - f_{n,<0}(h))_{ab}|  -|\mathbf{1}_{\tau>0} (f_{\ge 0}(h) - f_{n,\ge0}(h))_{ab}|  \\
& \le \mathbf{1}_{\tau \le 0} \| f_{<0}(h) - f_{n,<0}(h)\| + \mathbf{1}_{\tau >  0}\| f_{\ge 0}(h) - f_{n,\ge 0}(h)\| \\
& \le \frac{2\rho^{-n}}{\rho -1} \\
&=\frac{2e^{-\log(\rho)\lfloor \operatorname{dist}(a,b)/r_1\rfloor}}{\rho-1}\\
&\le \frac{2\rho e^{-\log(\rho) \operatorname{dist}(a,b)/r_1}}{\rho-1}\\
& \le \left(2+\frac{4\beta\|h\|}{\pi} \right)\,e^{- \mathrm{dist}(a,b)/\xi},
\end{align*}
where $\xi = r_1(\log \rho)^{-1} \le \Or(r_1\beta\|h\|) $. Combining this with Lemma \ref{spectral-norm-bound}, we complete the proof.
\end{proof}

\paragraph*{The Long-Range Case}
\begin{lemma}[Second Part of Lemma \ref{lemma-green-decay}]\label{secondpartlemmadecay}
 If $H_0$ satisfies 
$$
 \max_{a \in \Omega}\sum_{b\in \Omega} |h_{ab}| (e^{\theta \mathrm{dist}(a,b)} -1) \le \frac{\pi}{4\beta\|h\|},
$$
for some constant $\theta > 0$,  the Green's function is $(K,\xi)$-exponentially decaying with
$$
K = O(\beta\|h\|^2)\qquad\text{ and }\qquad  \xi = \theta^{-1}.
$$
\end{lemma} 
\begin{proof} Let $\sigma(h)$ denote the spectrum of the Hermitian matrix $h$. For $f \in \{f_{\ge 0}, f_{<0}\} $, we aim to prove the decaying of the matrix elements in $f(h)$.  Let $\delta = \frac{\pi}{2\beta \|h\|}$. Consider a contour $\Gamma$ that encloses the rectangle
$$ \{z:|\Re(z)| \le \|h\|+1,\,|\Im(z)| \le  \delta\}$$
Then the distance between points in $\Gamma$ and $\sigma(h)$ is at least $\delta$. We consider the contour integral representation
$$
f(h) = \frac{1}{2\pi i}\int_\Gamma f(z)(zI-h)^{-1} \dd z
$$
We now aim to estimate the decaying for each resolvent $((zI-h)^{-1})_{xy}$.

Fix $x_0 \in \Omega$ and define the multiplication operator $U_\theta$ by:
\[
(U_\theta v)_x= e^{\theta d(x,x_0)} v_x,\quad \forall  v \in \mathbb{C}^N.
\]
Define $h_\theta = U_\theta h U_\theta^{-1}$. The matrix elements in $h_\theta$ is given by:
\[
(h_\theta)_{xy} = h_{xy} e^{\theta (d(x,x_0) - d(y,x_0))}.
\]
Thus we have 
\[
(z I- h)^{-1} = U_\theta^{-1} (zI - h_\theta)^{-1} U_\theta.
\]
and the matrix elements of $(zI-h)^{-1}$ can be expressed as
\[
(zI - h)^{-1}_{x_0 y} = e^{-\theta d(y,x_0)} (z I- h_\theta)^{-1}_{x_0 y},
\]
which implies the estimate 
$$
|(z I- h)^{-1}_{x_0 y}| = e^{-\theta d(y,x_0)} \left\|(zI-h_\theta)^{-1} \right\| \le \frac{e^{-\theta \mathrm{dist}(x_0,y)}}{\delta - \|h-h_\theta\|}.$$
Now we estimate $\|h-h_\theta\|$. Note $|(h_\theta - h)_{xy}| \le |h_{xy}||e^{\theta d(x,y)} -1| $, using the condition of the Lemma we have
$$
\sum_{x}|(h_\theta - h)_{xy} | \le \delta/2
$$

\end{proof}

\subsection{Sampling Algorithms}
\subsubsection{Belief Propagation} \label{sec:bp}
Recall that we need to sample $P_1,\cdots,P_s \in \mathcal{P}$ from the following distribution.
$$
\mathbb{P}(P_1,\cdots,P_s) = \frac{1}{Z_s}|v_{P_1} \cdots v_{P_s}|\prod_{(i,j)\in T} M_{\tau_i,\tau_j}(P_i,P_j)
$$
The procedure of belief propagation algorithm involves the following steps:
\begin{enumerate}
    \item An \textbf{upward message pass} from the leaves to the root, where each node sends a summary of its subtree to its parent.
    \item Calculation of the \textbf{normalizing constant $Z_s$} at the root using the incoming messages.
    \item A \textbf{downward sampling pass} (ancestral sampling) from the root to the leaves, where each variable is sampled conditioned on the value of its parent.
\end{enumerate}
The message sent from a node $i$ to its parent $a(i)$, denoted $m_{i \to a(i)}(P_{a(i)})$, represents the aggregated information from the subtree rooted at $i$.

Without loss of generality, we assume $P_1,\cdots,P_s$ forms a growing path. Otherwise, we can simply reorder the vertices. We describe the algorithm as follows:
\begin{algorithm}[H]
\caption{Belief Propagation for Sampling and Normalization on a Tree}
\label{alg:bp_tree}
\begin{algorithmic}[1]
\REQUIRE Parent function $a(i) \in \{1,\cdots,i-1\}$ for each node $2\le i \le s$. Vertex factors $|v_{P_i}|$ and edge factors $M_{\tau_i,\tau_j}(P_i,P_j)$ for $P_i, P_j
\in \mathcal{P}$.
\STATE 
\ENSURE The normalizing constant $Z_s$ and a sample $(\hat{P}_1, \dots, \hat{P}_s)$.

\STATE
\STATE \textit{// Phase 1: Upward message pass (from leaves to root)}
\FOR{$i = s, s-1, \dots, 2$}
    \STATE \textit{// Compute the message vector $m_{i \to a(i)}$ for the parent $a(i)$.}
    \FOR{all $P_{a(i)} \in \mathcal{P}$}
        \STATE Let $C(i) \leftarrow \{k \in V \mid a(k)=i\}$ be the children list of $i$.
        \STATE $m_{i \to a(i)}(P_{a(i)}) \leftarrow \sum_{P_i \in \mathcal{P}} M_{\tau_{a(i)},\tau_i}(P_{a(i)}, P_i) |v_{P_i}| \prod_{k \in C(i)} m_{k \to i}(P_i)$
    \ENDFOR
\ENDFOR

\STATE
\STATE \textit{// Phase 2: Compute normalizing constant $Z_s$}
\STATE Let $C(1) \leftarrow \{k \in V \mid a(k)=1\}$.
\STATE Compute the unnormalized marginal for the root for each $P_1 \in \mathcal{P}$:
\STATE \qquad $b_1(P_1) \leftarrow |v_{P_1}| \prod_{k \in C(1)} m_{k \to 1}(P_1)$
\STATE $Z_s \leftarrow \sum_{P_1 \in \mathcal{P}} b_1(P_1)$

\STATE
\STATE \textit{// Phase 3: Downward sampling pass (from root to leaves)}
\STATE Compute the root distribution $\mathbb{P}(P_1) \leftarrow b_1(P_1) / Z_s$ for all $P_1 \in \mathcal{P}$.
\STATE Sample a state $\hat{P}_1$ from the distribution $\mathbb{P}(P_1)$.
\FOR{$i = 2, 3, \dots, s$}
    \STATE Let $C(i) \leftarrow \{k \in V \mid a(k)=i\}$.
    \STATE Compute the unnormalized conditional probability for each $P_i \in \mathcal{P}$:
    \STATE \qquad $b(P_i | \hat{P}_{a(i)}) \leftarrow M_{\tau_{a(i)},\tau_i}(\hat{P}_{a(i)}, P_i) |v_{P_i}| \prod_{k \in C(i)} m_{k \to i}(P_i)$
    \STATE Normalize to get the conditional distribution $\mathbb{P}(P_i | \hat{P}_{a(i)})$.
    \STATE Sample a state $\hat{P}_i$ from the distribution $\mathbb{P}(P_i | \hat{P}_{a(i)})$.
\ENDFOR

\STATE
\STATE \textbf{return} $Z_s$ and $(\hat{P}_1, \dots, \hat{P}_s)$.
\end{algorithmic}
\end{algorithm}
The total time complexity of the procedure is $O(s \cdot |\mathcal{P}|^2)$. This cost is dominated by the upward message pass (Phase 1). For each of the $s-1$ non-root vertices, the algorithm computes a message vector for its parent. The size of this vector is $|\mathcal{P}|$, and for each entry of the vector, a summation over all $|\mathcal{P}|$ states of the current vertex is required. Therefore, computing the full message vector for a single vertex costs $O(|\mathcal{P}|^2)$ and the total complexity of the upward pass $O(s|\mathcal{P}|^2)$. The other phases are faster, with Phase 2 taking $O(|\mathcal{P}|)$ time and Phase 3 taking $O(s |\mathcal{P}|)$ time. 

\subsubsection{Sampling Uniformly Random Labeled Tree}\label{sec:prufer}

Fix $s\ge2$ and let $P=(p_1,\dots,p_{s-2})\in[s]^{\,s-2}$.
Define $\Phi(P)$ to be the (unique) labeled tree $T$ on $[s]$ whose edge set
is constructed as follows: first, for any $v\in [s]$, 
\medskip
\[
d(v)\coloneqq 1+\#\{i\in\{1,\dots,s-2\}:p_i=v\}\,
\]
Here, $d(v)$ coincides with the final degree that vertex $v$ will have in the finished tree. Initially $\ell\coloneqq\{v\in[s]: d(v)=1\}$ is the set of leaves. Next, for $i=1,2,\dots,s-2$:
\begin{enumerate}
\item Let $v$ be the smallest element of $\ell$ (the smallest current leaf).
\item Add the edge $\{v,p_i\}$ to the edge set $E(T)$.
\item Update $d(v)\gets d(v)-1$ (so $d(v)=0$) and $d(p_i)\gets d(p_i)-1$.
\item Remove $v$ from $\ell$. If after the decrement $d(p_i)=1$, insert $p_i$ into $\ell$.
\end{enumerate}
After the loop, exactly two vertices $x\ne y$
satisfy $d(x)=d(y)=1$. Indeed, since the tree over $s$ vertices has $(s-1)$ edges, $\sum_{v}d(v)=2(s-1)$. Since at each step we decrease this sum by $2$, after $s-2$ steps, it remains $2(s-1)-2(s-2)=2$. With one edge left to add, those two remaining degree ``units'' must be supplied by the two endpoints of that final edge. A single vertex with 
$d(v)=2$ would require adding a loop (edge from the vertex to itself) to use both units, which is impossible in a simple tree. Therefore the only possibility is that two distinct vertices have $d=1$ and all others have $d=0$. Set $E(T)\gets E(T)\cup\{\{x,y\}\}$. The resulting graph $T$ is $\Phi(P)$.

\begin{lemma}[Pr\"ufer code: bijection and uniform sampling]
For $s\ge2$, the map $\Phi:[s]^{\,s-2}\to\mathcal{T}([s])$ that sends a sequence
$P=(p_1,\dots,p_{s-2})$ to the labeled tree obtained by the standard Pr\"ufer
\emph{decoding} is a bijection. Consequently, if $P$ is sampled uniformly from
$[s]^{\,s-2}$, then $\Phi(P)$ is a uniformly random labeled tree on $[s]$.
Moreover, $\Phi$ can be evaluated in $O(s)$ time and $O(s)$ space.
\end{lemma}

\begin{proof}
 Given a tree $T$ on $[s]$, the {encoding} iteratively removes the smallest leaf and records its unique neighbor; by repeating this $s-2$ times, exactly two vertices remain which correspond to the final edge in decoding. No symbol is recorded for them.
 This produces a
sequence in $[s]^{s-2}$. This procedure and the decoding map $\Phi$ described in the last paragraph are inverse to each other: encoding followed by
decoding (and vice versa) returns the original object. Hence $\Phi$ is a
bijection and $|\mathcal{T}([s])|=|[s]^{s-2}|=s^{s-2}$.
 If $P$ is uniform on $[s]^{s-2}$ and $\Phi$ is a bijection,
then $\Phi(P)$ is uniform on $\mathcal{T}([s])$ by pushforward of the uniform
measure. In decoding, each of the $s$ vertices has its degree
initialized once, each of the $s-2$ symbols is processed once, and each vertex
enters/leaves 
the leaf-set at most once. Maintaining $\ell$ via a moving pointer yields total $O(s)$ time and $O(s)$ space.
\end{proof}

\section*{Acknowledgments} We thank Zhen Huang and Yuanran Zhu for helpful discussions.
We thank Chenxuan Li for proof reading part of the manuscript.
J.C., G.K.C., 
H.C. and L.Y. are supported by the U.S. Department of Energy, Office of Science, Accelerated Research in Quantum Computing Centers, Quantum Utility through Advanced Computational Quantum Algorithms, grant no. DE-SC0025572.
S.O.S. is supported by the U.S. Department of Energy, Office of Science, National Quantum Information Science Research Centers, Quantum Science Center. C.R. is supported by the DFG under the grant TRR 352–Project-ID 470903074. J.J is supported by DOE QSA grant $\#$FP00010905. Y.Z. is supported by the National
Science Foundation, grant no. PHY-2317110.

\bibliographystyle{ieeetr}
\bibliography{ref}

\end{document}